\documentclass{article}
\usepackage{arxiv}
\usepackage[utf8]{inputenc}
\usepackage[T1]{fontenc}
\usepackage[numbers,sort&compress]{natbib}

\usepackage{hyperref}
\usepackage{xurl}
\usepackage{doi}
\usepackage{amssymb}
\usepackage{amsthm}
\usepackage{float}
\usepackage{flafter}
\usepackage{caption}

\AtBeginDocument{%
  }

\renewcommand{\shorttitle}{TrajMark}
\renewcommand{\headeright}{}
\renewcommand{\undertitle}{}
\date{}
\hypersetup{
  colorlinks=true,
  linkcolor=blue,
  citecolor=blue,
  urlcolor=blue,
  pdftitle={TrajMark: Ownership Attribution and Segment-Level Tamper Localization for Coding-Agent Trajectories},
  pdfauthor={Bokang Zeng, Zheng Gao, Xiaoyu Li, Xiaoyan Feng, Jiaojiao Jiang},
  pdfkeywords={coding agents, behavioral watermarking, provenance, ownership attribution, tamper detection, trajectory localization}
}
\newtheorem{theorem}{Theorem}[section]
\theoremstyle{definition}
\newtheorem{definition}[theorem]{Definition}
\theoremstyle{plain}
\providecommand{\Description}[2][]{}

\usepackage{mathtools}
\usepackage{booktabs}
\usepackage{longtable}
\usepackage{multirow}
\usepackage{threeparttable}
\usepackage{array}
\usepackage{tabularx}
\usepackage{adjustbox}
\usepackage{graphicx}
\graphicspath{{figures/}}
\usepackage{placeins}
\usepackage{xcolor}
\usepackage{listings}
\usepackage[most]{tcolorbox}
\usepackage{xspace}
\usepackage{tikz}
\usetikzlibrary{arrows.meta,positioning,fit,calc}
\usepackage{enumitem}
\usepackage{microtype}
\usepackage{algorithm}
\usepackage{algpseudocode}

\algrenewcommand\algorithmicrequire{\textbf{Require:}}
\algrenewcommand\algorithmicensure{\textbf{Ensure:}}
\algrenewcommand\algorithmiccomment[1]{\hfill\(\triangleright\)~\textit{#1}}

\newcommand{\tool}{TrajMark\xspace}
\newcommand{\ownerid}{\ensuremath{\mathbf{w}}}
\newcommand{\GF}{\ensuremath{\mathrm{GF}(2)}}

\newcommand{\hmac}{\operatorname{HMAC}}

\newcommand{\agentmarku}{AgentMark-U\xspace}
\newcommand{\acthook}{ActHook-style\xspace}
\definecolor{trajdetailframe}{HTML}{6F748F}
\definecolor{trajdetailback}{HTML}{F7F7FA}
\newtcblisting{trajdetail}[1]{
  enhanced,
  breakable,
  listing only,
  listing engine=listings,
  sharp corners,
  colback=trajdetailback,
  colframe=trajdetailframe,
  colbacktitle=trajdetailframe,
  coltitle=white,
  fonttitle=\bfseries\small,
  title={#1},
  title after break={#1\ (continued)},
  boxrule=0.65pt,
  left=2.4mm,
  right=2.4mm,
  top=1.2mm,
  bottom=1.2mm,
  before skip=7pt,
  after skip=8pt,
  listing options={
    basicstyle=\ttfamily\footnotesize,
    columns=fullflexible,
    keepspaces=true,
    showstringspaces=false,
    breaklines=true,
    breakatwhitespace=false,
    emph={Purpose,Public,Common,Masked,Direct,Visible,Verifier,Inputs,Recovery,Exact,Trigger,Replay,Commitment,SYSTEM,PROMPT,SETTING,USER,TASK,ISSUE,AGENT,PROPOSAL,TrajMark,wrapper,RELEASED,ACTION,VISIBLE,OBSERVATION,NEXT},
    emphstyle=\bfseries\color{trajdetailframe},
    tabsize=2
  }
}
\newcommand{\OwnerSelectionProbability}{0.30}

\newcommand{\PrimaryExactRecoveryConditions}{27}
\newcommand{\PrimaryConditions}{27}

\begin{document}

\title{TrajMark: Ownership Attribution and Segment-Level Tamper Localization for Coding-Agent Trajectories}

\author{%
  \begin{minipage}[t]{0.29\textwidth}\centering
    \textbf{Bokang Zeng}\\[2pt]
    \normalfont\small School of Computer Science\\
    and Engineering\\
    University of New South Wales\\
    Sydney, Australia\\
    \scriptsize\texttt{bokang.zeng@student.unsw.edu.au}
  \end{minipage}
  \And
  \begin{minipage}[t]{0.29\textwidth}\centering
    \textbf{Zheng Gao}\\[2pt]
    \normalfont\small School of Computer Science\\
    and Engineering\\
    University of New South Wales\\
    Sydney, Australia\\
    \scriptsize\texttt{zheng.gao1@unsw.edu.au}
  \end{minipage}
  \And
  \begin{minipage}[t]{0.29\textwidth}\centering
    \textbf{Xiaoyu Li}\\[2pt]
    \normalfont\small School of Computer Science\\
    and Engineering\\
    University of New South Wales\\
    Sydney, Australia\\
    \scriptsize\texttt{xiaoyu.li2@unsw.edu.au}
  \end{minipage}
  \AND
  \begin{minipage}[t]{0.36\textwidth}\centering
    \textbf{Xiaoyan Feng}\\[2pt]
    \normalfont\small School of Information and\\
    Communication Technology\\
    Griffith University\\
    Brisbane, Australia\\
    \scriptsize\texttt{xiaoyan.feng@griffithuni.edu.au}
  \end{minipage}
  \And
  \begin{minipage}[t]{0.36\textwidth}\centering
    \textbf{Jiaojiao Jiang}\\[2pt]
    \normalfont\small School of Computer Science\\
    and Engineering\\
    University of New South Wales\\
    Sydney, Australia\\
    \scriptsize\texttt{jiaojiao.jiang@unsw.edu.au}
  \end{minipage}
}

\maketitle

\begin{abstract}
Watermarking the final patch produced by a coding agent provides provenance evidence for the submitted artifact, but does not authenticate the visible process that produced it. Recent behavioral watermarking methods extend provenance to agent trajectories by embedding ownership evidence into observable actions. However, these methods primarily provide a global detection or identifier-recovery signal, so a locally edited trajectory may retain sufficient ownership evidence without revealing which protected region has become inconsistent. To address this limitation, we propose \textbf{\tool{}},
a training-free, symmetric-key, visible-only trajectory watermarking framework
that separates robust ownership attribution from fragile local integrity verification. Our framework consists of two complementary layers: a sparse owner layer that encodes a six-bit deployment identifier by rewriting a keyed subset of naturally occurring READ actions into masked linear equations, and a localization layer that inserts linked Q12 ordinary, group, and terminal seals to commit to protected critical-action segments. This separation allows ownership evidence to accumulate robustly across trajectories, while local modifications perturb nearby keyed commitments and expose the affected protocol region. We further provide a design-level analysis of owner recoverability, integrity collision probability, structural overhead, and localization behavior. Across three coding-agent frameworks and three LLMs, \tool{} recovers the exact owner in all evaluated clean full-watermark batches. Under exhaustive eligible single-site attacks it detects 95.5\%–100\% of edits, and under random single-action corruption it localizes 95.8\% of modified sites to an accepted protocol region rather than to the individual action. Owner marking adds no trajectory actions; the integrity layer adds explicit read-only seals, and matched Pass@1 is 26.9\% versus 26.3\% for unwatermarked runs. Overall, \tool{} preserves robust batch-level ownership evidence while providing segment-level tamper detection and localization.

\end{abstract}


\section{Introduction}

Coding agents~\cite{yang2024sweagent,wang2024openhands} are increasingly used for
software-engineering tasks such as automated bug fixing, test generation, and
repository maintenance. Unlike conventional code-generation systems, a coding-agent
run produces more than a final patch: it leaves a \emph{visible trajectory}, an ordered
record of observable actions including repository exploration, search, file reads,
edits, test executions, reproductions, and final submission. As such trajectories are
increasingly retained or released for auditing and provenance, two questions become
fundamental: \emph{who produced the trajectory}, and \emph{whether the released
trajectory has been modified after generation}. Reliable provenance therefore requires
protecting not only the final software artifact, but also the visible process that
produced it.

Existing watermarking mechanisms primarily protect generated content rather than the
underlying execution process. Token-level watermarks embed provenance signals into
generated text~\cite{kirchenbauer2023watermark,zhao2024provable}, code watermarks
modify source-code structure or lexical choices~\cite{collberg1999software}, and
cryptographic signatures can authenticate a submitted patch. These approaches provide
evidence about \emph{what} was produced, but not about \emph{how} it was produced.
The same patch may arise from substantially different debugging trajectories, while a
released action record may be selectively rewritten without changing the final patch.
Consequently, artifact-level provenance alone cannot authenticate the visible
development process.

Behavioral watermarking for agents takes an important step toward process-level
provenance by embedding ownership evidence directly into agent actions. AgentMark and
AGENTWM encode identifiers through behavioral or tool-execution
choices~\cite{huang2026agentmark,wang2026agentwm}; ActHook studies keyed hook behavior
in agent trajectories~\cite{meng2026acthook}; and SeqWM introduces sequential
transition patterns designed to retain watermark detection under trajectory
corruption~\cite{an2026seqwm}. These methods demonstrate that an agent's behavior,
rather than only its textual output, can serve as a watermark carrier. However, their
verification signals are primarily \emph{global}: they determine whether a trajectory
or collection of trajectories carries a watermark, or recover an embedded identifier.
A global ownership signal may therefore remain detectable even after part of the
released trajectory has been rewritten, without revealing which protected region has
become inconsistent.

This limitation exposes a more fundamental tension in trajectory provenance:
\emph{ownership and integrity require different evidence structures}. Ownership
evidence should be \emph{robust}: it should accumulate across naturally varying
trajectories and remain recoverable when individual actions are missing, altered, or
provide limited evidence. Tamper-localization evidence should instead be
\emph{fragile}: modifying a protected region should disturb a nearby commitment and
thereby reveal where the inconsistency occurred. Attempting to use a single global
watermark signal for both objectives conflates these opposing requirements.

Our key observation is therefore that trajectory provenance should not be represented
by a single homogeneous watermark channel. Instead, robust ownership attribution and
fragile local integrity should be treated as complementary but explicitly separated
protocol layers. Under such a design, ownership evidence can remain stable under
partial corruption, while local edits deliberately disrupt region-specific
commitments. This separation allows a verifier to retain confidence about
\emph{whose} trajectory collection it is while simultaneously identifying
\emph{where} an individual released trajectory has become inconsistent.

Motivated by this principle, we propose \tool{}, a training-free, visible-only
trajectory-watermarking framework for coding agents. \tool{} operates on a public
canonical projection of observable actions rather than assuming access to the complete
semantics or hidden state of an execution. Its \emph{owner layer} embeds a six-bit
deployment identifier into a keyed subset of naturally occurring READ actions. Letting
the deployment identifier lie in $\mathrm{GF}(2)^6$, each selected normalized
filename contributes a masked linear equation over $\mathrm{GF}(2)$ with a nonzero
coefficient row. The equation is encoded by rewriting only the READ subtype while
preserving its semantic family and file target. Because these carriers reuse actions
that the agent already intended to execute, owner-only mode introduces no additional
trajectory steps, while sparse equations can be aggregated across multiple trajectories
for batch-level identifier recovery.

In parallel, the \emph{integrity and localization layer} introduces deliberately
fragile commitments around critical portions of the visible trajectory. \tool{}
partitions the canonical action stream into protected segments and emits linked Q12
ordinary seals at eligible critical-segment boundaries. Pairs of ordinary seals are
protected by group commitments, and a terminal seal closes the remaining suffix and
observed seal sequence. Consequently, an eligible modification to a protected segment
% changes the inputs to one or more nearby keyed commitments, causing the corresponding
changes the inputs to one or more nearby keyed commitments, which can cause the corresponding
ordinary, group, or terminal replay check to become inconsistent. The two layers thus
serve complementary purposes: ownership evidence is sparse and redundant across
trajectories, whereas integrity evidence is intentionally sensitive to local structural
changes.

Verification requires only the published visible trajectory, a secret key, and public
canonicalization rules. It does not rely on model logits, hidden reasoning, private
wrapper metadata, authenticated side information, or access to the original execution
environment. The verifier reconstructs the ownership equations from visible READ actions and
independently replays the ordinary, group, and terminal seal protocol. It therefore
produces two distinct outputs: a batch-level deployment-identifier decision and
trajectory-level alarms identifying the protected segment, segment pair, or suffix
whose keyed replay has become inconsistent. Importantly, \tool{} authenticates this
public canonical projection rather than claiming to reconstruct deleted content or
certify the complete semantics of an execution.

We implement \tool{} across three independently engineered coding-agent frameworks and
evaluate it with three LLMs. Across all 27 evaluated full-watermark
agent--provider--seed groups, the visible-only decoder recovers the exact embedded
six-bit identifier. Under exhaustive eligible single-site attacks, \tool{} detects
95.5\%--100\% of edits across the nine agent--LLM conditions; for random one-action
payload corruption, accepted localization covers 95.8\% of eligible modified sites,
while matched Pass@1 reaches 26.9\% versus 26.3\% without watermarking. These results
demonstrate that robust ownership
attribution and fragile local integrity can coexist within the same visible
coding-agent trajectory without requiring model retraining or hidden verifier inputs.

We summarize our contributions as follows:
\begin{itemize}
    \item \textbf{Problem and design principle.}
    We formulate coding-agent trajectory provenance over a public canonical projection
    of released action records and identify the fundamental distinction between robust
    batch-level ownership attribution and fragile trajectory-level integrity.

    \item \textbf{Framework.}
    We propose \tool{}, a training-free, visible-only framework that combines sparse
    filename-indexed linear ownership equations over natural READ actions with linked
    Q12 ordinary, group, and terminal commitments. This two-layer design enables
    batch-level owner recovery together with protocol-level tamper localization.

    \item \textbf{Analysis and empirical evaluation.}
    We provide a design-level analysis of owner recoverability, integrity collision
    behavior, structural overhead, and localization, and conduct an extensive
    evaluation across three coding-agent frameworks and three LLMs, covering exact
    identifier recovery, random trajectory corruption, exhaustive single-site
    tampering, protocol-unit localization, structural cost, and repository-level task
    performance.
\end{itemize}

\section{Background and Motivation}
\label{sec:background}

\subsection{Coding-Agent Trajectories}
\label{sec:background-trajectories}

A coding agent repeatedly observes a repository state, asks a language model for the
next decision, executes a tool, and records the resulting action and observation.
Interfaces differ: SWE-agent exposes a text-action agent--computer interface
\cite{yang2024sweagent}; OpenHands represents interactions as typed events
\cite{wang2024openhands}; and OpenDev exposes structured tool calls. Nevertheless, their
visible behavior contains recurring semantic families:

\begin{itemize}
  \item \textbf{EXPLORE}: inspect directory structure or repository metadata;
  \item \textbf{SEARCH}: find names, symbols, imports, or matching content;
  \item \textbf{LOCATE}: read all or a bounded region of a file;
  \item \textbf{CRITICAL}: edit code, refactor, reproduce a defect, or run tests; and
  \item \textbf{TERMINAL}: submit or end the run.
\end{itemize}

Although these runtimes may record both actions and observations, \tool{} verifies
a public canonical projection of visible action records, including public linkage and
execution status; observation contents beyond these fields are not detector inputs.
The first three families are read-only and often admit
multiple surface realizations. \tool{} uses paired realizations that retain the READ
family, file target, and read-only character, although they may return different
observations and their utility impact must therefore be evaluated empirically.
Agent-proposed natural READs provide Q6 ownership carriers. Separately, framework
adapters instantiate additional read-only actions from the Q12 alphabet as integrity
seals; Section~\ref{sec:trajmark-carriers} defines both alphabets and their roles.

\subsection{Three Provenance Mechanisms}
\label{sec:background-provenance}

\paragraph{Artifact watermarks.}
Text and code watermarks modify a generated artifact and detect a statistical or
structural pattern in that artifact \cite{kirchenbauer2023watermark,
zhao2024provable,cox2007digital,collberg1999software}. They are complementary to
\tool{}: an artifact watermark carries provenance evidence about the final text or
code, whereas a trajectory watermark carries evidence over a protected canonical
representation of the action record.

\paragraph{Cryptographic audit logs.}
Hash chains, forward-secure logs, and signed commitments provide strong integrity when
a trusted logger can attach authenticators and downstream verifiers preserve and
validate them \cite{schneier1999secure}. They remain preferable in that setting, and
their authenticators may be retained in-band or alongside the log. \tool{} instead
studies the restricted observation contract formalized in Section~\ref{sec:problem}:
the released interface preserves a canonical action trajectory but exposes no
verifiable signature, hash-chain state, or authenticated wrapper metadata. \tool{}
therefore complements rather than replaces cryptographic audit logging and does not
provide the same freshness, replay, or non-repudiation guarantees.

\paragraph{Behavioral watermarks.}
AgentMark uses behavior distributions and conditional sampling to encode multi-bit
identifiers while preserving utility \cite{huang2026agentmark}. AGENTWM constructs
functionally equivalent execution paths \cite{wang2026agentwm}. ActHook studies
secretly activated hook behavior in trajectory data and later detects that behavior in
a trained agent \cite{meng2026acthook}. SeqWM uses sequential transition patterns to
retain detection under trajectory corruption \cite{an2026seqwm}. These methods protect
different objects and provide global detection or identifier-recovery evidence; they
do not provide an explicit protocol-level output identifying which protected segment,
seal pair, or suffix has become inconsistent.

\subsection{Design Rationale for Two Layers}
\label{sec:background-two-layers}

Ownership and localization favor different evidence structures. Ownership benefits
from sparse, redundant evidence aggregated across trajectories, whereas localization
requires region-specific commitments within an individual trajectory. Reusing one
globally aggregated signal and decision rule for both objectives creates a tension:
robustness favors tolerance to missing carriers, while local diagnosis requires
additional visible evidence tied to protocol-defined regions. \tool{} therefore makes
the two roles separately deployable. Owner-only mode inserts no additional visible
actions but rewrites selected natural READ subtypes; integrity mode adds overt
read-only seal actions at an explicit structural cost. Combined deployment retains
separate carrier roles and detector outputs, although owner rewrites remain part of the
canonical stream committed by the integrity layer. Because either intervention can
change observations and downstream behavior, utility is evaluated empirically.

\section{Problem Formulation and Threat Model}
\label{sec:problem}

We study provenance for a released coding-agent trajectory, where the verifier sees
the agent's public actions but not its hidden reasoning or runtime state.  The central
constraint is therefore observational: every carrier, role, and decision must be
recoverable from the released action stream.  We first fix the trust boundary and
online run, then define the canonical stream, verification tasks, and adversary.

\subsection{System Model and Trust Boundary}
\label{sec:system-model}

\paragraph{Parties and attribution scope.}
\tool{} is a symmetric-key mechanism operated inside a trusted verification domain.
A trusted embedding service runs the coding agent with a watermarking wrapper, and a
trusted verifier holds the same secret key \(K\).  The key authority assigns a
six-bit, deployment-scoped payload \(\ownerid\in\GF^6\); it is not a publicly
verifiable legal or organizational identity.  Any holder of \(K\) can
produce a valid mark for any label.  The construction consequently provides neither
public verifiability nor non-repudiation, and it does not cover a dishonest embedding
service, a dishonest verifier, or key compromise.  The claimed label
\(\mathbf w_c\), configuration \(\theta\), and verifier-selected batch membership are
trusted inputs fixed before decoding rather than inferred from the submitted stream.

\paragraph{Online generation.}
Let \(\mathcal A\) be a coding agent, \(\mathcal E\) its execution environment, \(x\)
a repository task, and \(\omega\) the run randomness.  Watermarking is part of the
agent loop:
\begin{equation}
  \tau_w \leftarrow
  \mathsf{Run}(\mathcal A,\mathcal E,x,\omega;
  \mathsf{Wrap}_{K,\ownerid,\theta}).
  \label{eq:online-run}
\end{equation}
At decision step \(t\), the agent proposes a request from the already marked history
\(h^w_{t-1}\).  The wrapper may rewrite an eligible READ before execution, update its
state after the visible outcome, and execute integrity seals before the next agent
decision.  A \emph{natural READ} is therefore an agent proposal observed immediately
before the current rewrite, not an action recovered from a completed counterfactual.
The No-WM run is a separate experimental control rather than an input to
Eq.~\eqref{eq:online-run}.

\subsection{Canonical Visible Observation}
\label{sec:visible-observation}

A released trajectory is an ordered sequence of visible records
\(\tau=(a_1,\ldots,a_n)\).  A public, versioned causal transducer joins only
publicly linkable records.  For example, it joins an OpenHands action with its
subsequent observation, then emits canonical events,
\begin{equation}
  \begin{aligned}
  s_\theta(\tau)&=\Pi_\theta(a_1,\ldots,a_n)=(e_1,\ldots,e_m),\\
  e_i&=(\mathsf{class}_i,\mathsf{subtype}_i,
   \mathsf{target}_i,\mathsf{status}_i,\mathsf{header}_i).
  \end{aligned}
  \label{eq:stream}
\end{equation}
where unavailable subtypes, targets, or headers use fixed null symbols and status is
\(\mathsf{success}\), \(\mathsf{failed}\), or \(\mathsf{unknown}\).  Every visible
action has a class, including \(\mathsf{OTHER}\) outside the protocol taxonomy.
The optional header retains public carrier fields required by a native profile;
OpenDev publishes a role, ordinal, and authentication tag in its tool-call text.

\begin{definition}[Visible-only observation]
\label{def:visible}
The verifier's trajectory-derived input is exactly \(s_\theta(\tau)\).  In addition,
it receives \(K\), \(\theta\), the claimed label, and the batch definition as trusted
inputs.  Candidate syntax and filename eligibility are determined from the released
prefix and \(\theta\).  Profile-specific role replay may also use \(K\) to recognize
owner carriers or authenticate public carrier headers.  Runtime checks such as file
existence are safety checks, not verifier-side eligibility conditions.
\end{definition}

The retained view contains class, subtype, canonical target, public status, any
profile-defined public carrier header, and order.
It excludes hidden reasoning, logits, private wrapper state, debug metadata, full
commands or outputs, and observation contents beyond public linkage and status.  The
complete field contract is listed in Appendix~\ref{app:protocol-details}.

\subsection{Public Replay and Protected Projections}
\label{sec:public-replay}

Generation knows which actions were proposed by the agent and which were inserted by
the wrapper; verification may not use those private labels.  For an edited release
\(\tilde\tau=\mathcal T(\tau_w)\), the verifier replays the declared placement rules
from the beginning:
\begin{equation}
  (\widehat N,\widehat L,D_{\mathrm{struct}})
  \leftarrow
  \mathsf{ProtocolReplay}_{K,\theta}(s_\theta(\tilde\tau)).
  \label{eq:role-parser}
\end{equation}
Here \(\widehat L\) assigns expected ordinary, group, and terminal roles at public
protocol positions, \(D_{\mathrm{struct}}\) records missing or invalid expected
carriers, and
\[
  \widehat N=\{i\notin\widehat L:
  \mathsf{Elig}_\theta(e_i,e_{<i})=1\}
\]
contains replay-assigned ownership positions.  Given the replay-assigned roles,
eligibility uses only the public prefix and \(\theta\); it does not recover
wrapper-private origins after editing.  Owner
decoding excludes every integrity role.

The two channels consume different projections of the same canonical stream.  The
owner projection \(\rho_{\mathrm{own}}\) retains eligible, replay-assigned non-seal
READ families, their Q6 subtypes, and normalized basenames.  The first eligible
occurrence of a basename in each trajectory supplies its carrier record.  For protected
segment \(S_k\), the integrity projection is
\begin{equation}
  \rho_{\mathrm{int}}(S_k)=
  \mathsf{counts}_{12}(S_k)
  \parallel
  \bigl((j,\mathsf{class}_j,\mathsf{tgt}_\theta(e_j))\bigr)_{
    \mathsf{class}_j\in\mathcal C_{\mathrm{crit}}},
  \label{eq:integrity-projection}
\end{equation}
where
\begin{equation*}
  \mathcal C_{\mathrm{crit}}=
  \{\mathsf{GENERATE\_FIX},\mathsf{REFACTOR},
    \mathsf{RUN\_TESTS},\mathsf{REPRODUCE}\}.
\end{equation*}
Here \(j\) is the zero-based ordinal among critical events in \(S_k\), and
\(\mathsf{tgt}_\theta(e_j)\) is the normalized basename of the event's explicit public
target, or the fixed null token when none is visible.
Critical-event order, class, and canonical target are retained.  Exploratory READ
filenames, READ order within a
segment, \(\mathsf{OTHER}\) contents, complete commands, and observation contents are
not protected by this projection.  Seal roles, subtypes, public status, and order are
checked separately during replay.

For later use, define the protected protocol view
\begin{equation}
  P_{\mathrm{int}}(\tau)=
  \bigl((\rho_{\mathrm{int}}(S_k))_k,
  ((\mathsf{role}_\ell,\mathsf{subtype}_\ell,
  \mathsf{status}_\ell))_\ell,
  D_{\mathrm{struct}}\bigr),
  \label{eq:protected-view}
\end{equation}
with seal order represented by sequence position.  Projection-preserving edits are
outside the integrity claim by definition.

\subsection{Verification Tasks and Adversarial Scope}
\label{sec:verification-tasks}

For a verifier-defined batch \(\widetilde{\mathcal B}\) claimed to share one key,
configuration, and deployment label, owner verification returns
\[
  \mathsf{OwnerVerify}_{K,\theta}
  (\widetilde{\mathcal B},\mathbf w_c)
  \rightarrow (\hat{\mathbf w},d_{\mathrm{owner}}),
\]
where
\(d_{\mathrm{owner}}\in\{\mathsf{ACCEPT},\mathsf{REJECT},\mathsf{ABSTAIN}\}\).
Acceptance means that sufficient retained batch evidence is consistent with the
registered key and claimed label.  A different unique label that clears the same rank
and evidence gates yields rejection; insufficient or ambiguous evidence yields
abstention.  Unique six-bit
recovery is not open-world attribution and does not attribute every member of the
batch individually.

For one trajectory, integrity replay returns
\[
  \mathsf{IntegrityReplay}_{K,\theta}(\tilde\tau)
  \rightarrow(d_{\mathrm{int}},\mathcal A_{\mathrm{int}}),
\]
where \(d_{\mathrm{int}}\in\{\mathsf{CLEAN},\mathsf{ALARM}\}\) and
\(\mathcal A_{\mathrm{int}}\) contains segment-, group-, suffix-, or structural
alarms.  A mismatched or missing expected carrier, or one whose status is disallowed by
\(\theta\), produces an alarm.  The reference protocol requires successful carrier
execution; the evaluated profiles' status and termination rules are specified in
Appendix~\ref{app:framework-details}.  An
alarm establishes inconsistency under \(K\) and \(\theta\); it does not distinguish
malicious editing from a wrong key or configuration, an unmarked input, or a benign
runtime failure.

\paragraph{Post-publication adversary.}
The editor knows the algorithms, claimed label, and public configuration and may
observe other marked trajectories, but it does not know \(K\), control the trusted
wrapper, alter trusted verifier inputs, or obtain adaptive localized feedback.  It may
insert, delete, or replace visible records and may target publicly recognizable
protocol positions.  The integrity probability statements in
Section~\ref{sec:analysis} concern a fixed, one-shot edit; adaptive
verification-oracle attacks are outside the model.

\paragraph{Evaluated corruption scopes.}
Ownership experiments apply non-key-aware insertion, deletion, and replacement to
homogeneous batches; integrity experiments edit protected payloads or visible carriers
relative to each trajectory's paired-clean replay baseline.
Section~\ref{sec:analysis} proves only fixed,
role-preserving single-segment substitutions.  Insert/delete/carrier cases, targeted
rank suppression, and adaptive claims are empirical or outside scope.

\paragraph{Explicit exclusions.}
An intact replay of a valid marked trajectory remains valid evidence for its embedded
label.  \tool{} does not establish freshness, execution uniqueness, or binding to an
external task, repository commit, patch, user, or session.  Whole-object replay,
same-label substitution, arbitrary mixed-label batches, cross-trajectory splicing,
arbitrary compound attacks on integrity or localization, complete evidence erasure,
and key or wrapper compromise are not covered.  Applications that need context binding
must authenticate that context separately.  Removing enough owner evidence may force
\(\mathsf{ABSTAIN}\); visible-only
verification cannot reconstruct deleted content.

\paragraph{Design objective.}
\tool{} seeks robust batch-level recovery when enough independent owner evidence
survives, fragile localization when the protected view changes, and abstention when
evidence is insufficient.  All three outcomes must be derived from the released stream,
\(K\), public configuration, and predeclared claim and batch definition.

\section{\tool{}: Dual-Channel Visible-Trajectory Watermarking}
\label{sec:trajmark}

\subsection{Method Overview}
\label{sec:trajmark-architecture}

Figure~\ref{fig:overview} shows the two-channel loop.  The robust owner channel maps
selected agent-proposed READs to linear equations over \(\GF^6\), preserving their
family and target and adding no action.  The fragile integrity channel closes public
segments with ordinary, group, and terminal Q12 seals.  From the released stream, the
verifier solves a batch label (``who'') and independently localizes inconsistent
protocol units (``where'').

\begin{figure}[H]
\centering
\includegraphics[width=\textwidth]{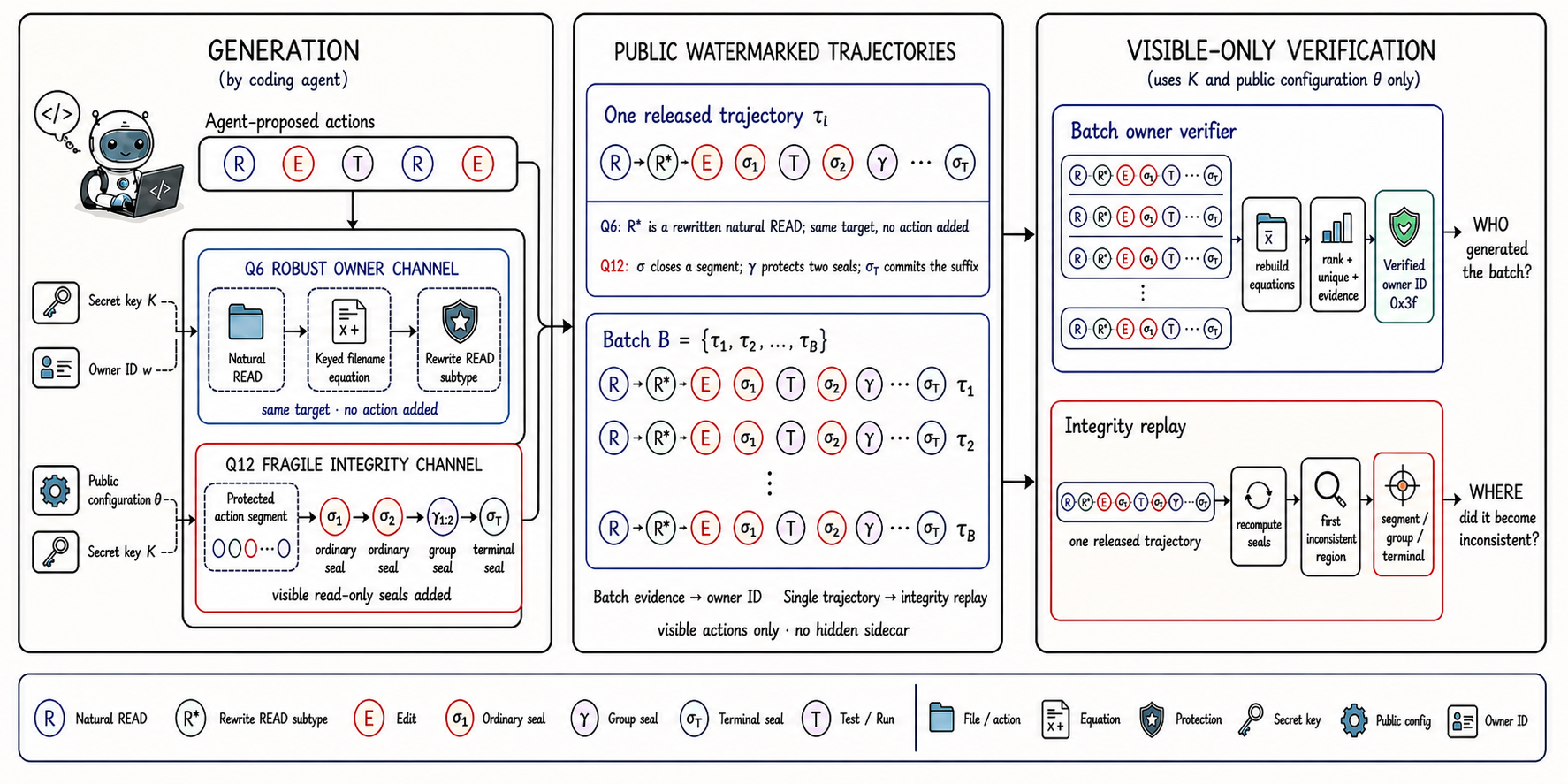}
\caption{\tool{} separates robust ownership evidence from fragile projected
integrity. The owner channel rewrites selected natural READs and adds no actions,
whereas the integrity channel emits ordinary, group, and terminal Q12 seals. The
released public trajectories carry both channels: batch evidence supports owner-ID
recovery, while one trajectory supports independent integrity replay.}
\Description{A two-channel pipeline in which natural actions pass through an owner
encoder and an integrity-seal encoder.  A batch decoder recovers a six-bit label, while
visible replay reports inconsistent segments, segment pairs, or suffixes.}
\label{fig:overview}
\end{figure}

We use \(\operatorname{Enc}_\theta\) as abstract notation for a public,
unambiguous profile encoding.  The evaluated adapters use two profile-specific
terminal serializers; Appendix~\ref{app:framework-details} explains their shared
fields and adapter-level differences.  Thus \(\parallel\) below denotes the encoding
selected by the declared framework profile.

\subsection{Public Carrier Space and Role Separation}
\label{sec:trajmark-carriers}

Natural READs assigned to \(N(\tau)\) may carry ownership equations; replay-assigned
seal positions in \(L(\tau)\) carry integrity tags and never contribute owner
equations.  Independent HMAC domain tags separate selection, row generation, masking,
orientation, ordinary seals, group seals, and the terminal seal.

The public Q6 alphabet provides one binary pair per READ family; Q12 exposes all twelve
subtypes to integrity seals.  The concrete mapping appears in
Appendix~\ref{app:protocol-details}; secrecy lies only in keyed choices and values.

\subsection{Robust Ownership Channel}
\label{sec:owner}

\subsubsection{Causal Filename Slots}
\label{sec:owner-slots}

An owner candidate is a natural READ that (i) has an explicit regular repository-file
target, (ii) belongs to EXPLORE, SEARCH, or LOCATE, and (iii) is the first eligible
READ of its normalized basename in that trajectory.  Write
\[
  s_f=\operatorname{norm}(\operatorname{basename}(f)).
\]
The filename-once rule prevents one loop over a file from dominating a trajectory and
allows both runtime and verifier to decide eligibility from the current public prefix.
The public configuration supplies a real-valued threshold \(\bar p_o\).  The key selects
a basename by
\begin{align}
u_f &= \frac{\hmac_K(\texttt{owner-select}\parallel s_f)_{64}}{2^{64}},\\
f\text{ is selected} &\Longleftrightarrow u_f<\bar p_o .
\label{eq:owner-select}
\end{align}
The implementation compares this normalized 64-bit draw with the configured threshold;
write \(p_o\) for the resulting selection probability.  The primary configuration uses
\(\bar p_o=\OwnerSelectionProbability\); profile-level realizations are summarized in
Appendix~\ref{app:framework-details}.

\subsubsection{Masked Linear-Equation Injection}
\label{sec:owner-equations}

For selected basename \(s_f\), rejection sampling maps a domain-separated HMAC output
uniformly to a nonzero row
\[
  \mathbf a_f\in\GF^6\setminus\{\mathbf0\}.
\]
The unmasked payload equation and keyed mask are
\begin{align}
b_f&=\langle\mathbf a_f,\ownerid\rangle\bmod2,
\label{eq:owner-equation}\\
m_f&=\hmac_K(\texttt{owner-mask}\parallel s_f)\bmod2,
\qquad c_f=b_f\oplus m_f.
\label{eq:owner-mask}
\end{align}
For READ family \(F\), let
\(\mathsf{Q6}_F[0],\mathsf{Q6}_F[1]\) be its public pair.  A further domain-separated
bit randomizes the public orientation:
\begin{equation}
o_{f,F}=\hmac_K(\texttt{owner-orient}\parallel F\parallel s_f)\bmod2,
\qquad j_f=c_f\oplus o_{f,F}.
\label{eq:owner-orientation}
\end{equation}
The wrapper executes \(\mathsf{Q6}_F[j_f]\) with the original target.  The same public
pair therefore has no fixed visible meaning across filenames.

Equations~\eqref{eq:owner-mask}--\eqref{eq:owner-orientation} specify the masked
profile used by all three evaluated adapters.  SWE-agent and OpenHands preserve the
natural READ family, while OpenDev's evaluated owner carrier uses the Locate pair.
The adapters use profile-specific canonical slots and pair contexts but expose the
same decoded equation \((\mathbf a_f,b_f)\); their visible realization differences are
summarized in Appendix~\ref{app:framework-details}.

% Repeated basenames add observations but not rank; distinct basenames may also collide
% on one of 63 rows.  The decoder therefore reports raw observations and unique rows.
Each basename defines a slot.  Coefficient collisions do not merge slots or
increase rank; only coefficient types are bounded by 63.

\subsubsection{Batch Decoding and Three-Way Decision}
\label{sec:owner-decoding}

For a selected visible carrier with observed Q6 index \(j_f\), the verifier reconstructs
\begin{equation}
  y_f=j_f\oplus o_{f,F}\oplus m_f
     =\langle\mathbf a_f,\ownerid\rangle
\label{eq:owner-recover}
\end{equation}
% on an unedited carrier.  Repeated observations of the same basename slot are
% aggregated by majority, and a tied slot is erased.  Let \(U\) be the number of retained
% slots and \(H(\mathbf x)\) the number of their equations satisfied by candidate
% \(\mathbf x\in\GF^6\).  The decoder enumerates all 64 candidates and chooses a unique
% maximizer when one exists.
on an unedited carrier.  Majority aggregation retains one equation per basename
slot; ties are erased.  Let \(U\) count retained slots and \(H(\mathbf x)\) count
their equations satisfied by \(\mathbf x\in\GF^6\).  The decoder scores all 64
candidates and selects a unique maximizer when one exists.

For the predeclared claim \(\mathbf w_c\), the descriptive evidence score is
\begin{equation}
z_{\mathrm{owner}}(\mathbf w_c)
=\frac{H(\mathbf w_c)-U/2}{\sqrt{U/4}}.
\label{eq:owner-z}
\end{equation}
Given threshold \(\tau_z\), first require that the retained matrix has rank six, the
Hamming-score maximizer \(\hat{\mathbf w}\) is unique, and
\(z_{\mathrm{owner}}(\hat{\mathbf w})\ge\tau_z\).  If any gate fails, the verifier
abstains; otherwise it accepts when \(\hat{\mathbf w}=\mathbf w_c\) and rejects when
\(\hat{\mathbf w}\ne\mathbf w_c\).  Rank six is a payload
capacity condition, not by itself an acceptance condition or a robustness guarantee;
Section~\ref{sec:analysis} makes both distinctions precise.

\subsection{Fragile Projected-Integrity Channel}
\label{sec:localization}

\subsubsection{Segments and Public State}
\label{sec:local-segments}

The reference integrity state machine processes replay-assigned non-seal events.
Its state contains the current segment buffer, the
ordinary index \(k\), the observed subtype of the previous ordinary seal, the ordinary
digests and subtypes, the group count, and a Boolean \(c_{\mathrm{prev}}\), initially
false.  After adding a non-seal event \(e\) to the current payload, set
\(c=\mathbf1[\mathsf{class}(e)\in\mathcal C_{\mathrm{crit}}]\).  The boundary rule is
\[
\mathsf{Boundary}_\theta(e,c_{\mathrm{prev}})=c\land\neg c_{\mathrm{prev}},
\qquad c_{\mathrm{prev}}\gets c.
\]
Thus a critical event at stream start or after a noncritical event closes a segment immediately;
subsequent consecutive critical events enter the next buffer without another immediate
seal.  Ordinary and group roles do not update \(c_{\mathrm{prev}}\).  After an ordinary
role, the buffer is cleared, and every second ordinary role schedules a group role
before the next natural event.  The reference schedule advances over attempted roles
and checks status separately.  The terminal role commits the remaining buffer, including
an empty suffix.  The evaluated profiles' admission, boundary-state, failure, and closure
rules are specified in Appendix~\ref{app:framework-details}.

READ events contribute their Q12 subtype counts to the buffer.  If the critical events
in a segment are numbered \(j=0,\ldots,r\) in their visible order, they contribute
tokens
\[
t_j=j\parallel \mathsf{class}_j\parallel
\mathsf{tgt}_\theta(e_j),
\]
using the public target rule in Section~\ref{sec:public-replay}.  The resulting segment
payload is
\begin{equation}
\mathsf{payload}(S_k)=
\mathsf{counts}_{12}(S_k)\parallel(t_0,\ldots,t_r).
\label{eq:segment-payload}
\end{equation}
This is exactly the projection in Eq.~\eqref{eq:integrity-projection}; changing data
outside that projection is not an integrity event claimed by \tool{}.

\subsubsection{Ordinary, Group, and Terminal Commitments}
\label{sec:ordinary-seal}

Let \(\tilde\sigma_{k-1}\) be the observed subtype of the preceding ordinary role and
set \(\tilde\sigma_0=\mathsf{GENESIS}\).  The \(k\)-th ordinary digest and visible
subtype are
\begin{align}
d_k &= \hmac_K(\texttt{ordinary}\parallel k\parallel
        \tilde\sigma_{k-1}\parallel\mathsf{payload}(S_k)),\\
\sigma_k &= \mathsf{Q12}[d_k\bmod12].
\label{eq:ordinary-seal}
\end{align}
Chaining the \emph{observed} predecessor confines an ordinary mismatch: replay of the
next segment starts from the symbol actually present in the stream rather than an
unseen expected symbol.

After ordinary roles \(2g-1\) and \(2g\), the group protector is
\begin{align}
g_g = \hmac_K(&\texttt{group}\parallel g\parallel
  d_{2g-1}\parallel d_{2g}\parallel
  \tilde\sigma_{2g-1}\parallel\tilde\sigma_{2g}),\\
\gamma_g &= \mathsf{Q12}[g_g\bmod12].
\label{eq:group-seal}
\end{align}
The group input contains the full ordinary digests, not only their visible residues,
so a payload change that collides modulo 12 can still create fresh group evidence.

Finally, let \(\mathbf d\), \(\tilde{\boldsymbol\sigma}\), and
\(\tilde{\boldsymbol\gamma}\) be the ordinary digests and the observed ordinary and
group subtype sequences.  With empty sequences encoded canonically and
\(\tilde\sigma_m=\mathsf{GENESIS}\) when \(m=0\), define
\[
C_T=H_\theta(\mathbf d\parallel\tilde{\boldsymbol\sigma}
      \parallel\tilde{\boldsymbol\gamma})
\]
and
\begin{equation}
\sigma_T=\mathsf{Q12}\!\left[
\hmac_K(\texttt{terminal}\parallel m\parallel g\parallel
\tilde\sigma_m\parallel C_T\parallel\mathsf{payload}(S_T))
\bmod12\right].
\label{eq:terminal-seal}
\end{equation}
Here \(H_\theta\) is a fixed public hash identified by the protocol version.  The
terminal closes the final suffix and, conditional on collision resistance, binds the
seal counts and visible seal sequence.

\subsubsection{Visible-Only Replay and Localization}
\label{sec:local-replay}

Reference replay reconstructs each segment and consumes the next visible record at every expected
role.  Missing, mismatched, unsuccessful, or trailing records raise structural alarms.
An ordinary mismatch reports one segment, a group mismatch its adjacent pair, and a
terminal mismatch a suffix or closure inconsistency.  These are protocol units rather
than byte-level locations.  The complete transition rules appear in
Appendix~\ref{app:protocol-details}.

\subsection{End-to-End Online Protocol}
\label{sec:detection}

The online wrapper alternates agent actions with any scheduled integrity roles; the
verifier first reconstructs public roles, then decodes owner equations and replays
integrity commitments.  The owner channel adds no action, whereas the integrity channel
is overt and incurs visible step cost.  Either channel may be enabled independently.
Reference pseudocode is given in Appendix~\ref{app:protocol-details};
Section~\ref{sec:implementation} distinguishes it from the evaluated framework profiles.

\newtheorem{assumption}[theorem]{Assumption}
\newtheorem{remark}[theorem]{Remark}

\definecolor{trajblueframe}{RGB}{68,104,132}
\definecolor{trajbluefill}{RGB}{239,247,251}
\definecolor{trajredframe}{RGB}{154,91,83}
\definecolor{trajredfill}{RGB}{252,243,241}
\newsavebox{\trajblueboxstore}
\newsavebox{\trajredboxstore}
\newenvironment{trajtheorembox}{%
  \par\medskip\noindent
  \begin{lrbox}{\trajblueboxstore}%
  \begin{minipage}{\dimexpr\linewidth-2\fboxsep-2\fboxrule\relax}%
}{%
  \end{minipage}\end{lrbox}%
  \noindent\fcolorbox{trajblueframe}{trajbluefill}{\usebox{\trajblueboxstore}}%
  \par\medskip
}
\newenvironment{trajremarkbox}{%
  \par\smallskip\noindent
  \begin{lrbox}{\trajredboxstore}%
  \begin{minipage}{\dimexpr\linewidth-2\fboxsep-2\fboxrule\relax}%
}{%
  \end{minipage}\end{lrbox}%
  \noindent\fcolorbox{trajredframe}{trajredfill}{\usebox{\trajredboxstore}}%
  \par\smallskip
}

\section{Theoretical Analysis}
\label{sec:analysis}

We now state the guarantees that explain the two-channel design.  The main text keeps
only the assumptions, headline results, and their interpretation; recurrences and
complete proofs are deferred to the appendices.
The owner result separates a key-independent frozen-pool model from decoding of an
observed batch, while the integrity result concerns one fixed role-preserving edit.

\begin{assumption}[Canonical replay]
\label{ass:canonical-replay}
Generation and verification use the same deterministic public transducer, role state
machine, and typed, length-prefixed, injective encoding \(\operatorname{Enc}_\theta\).
Boundary placement, null targets, role order, and successful-status requirements are
fixed by \(\theta\).
\end{assumption}

\begin{assumption}[Idealized keyed outputs]
\label{ass:ideal-prf}
Domain-separated \(\lambda\)-bit HMAC outputs on distinct fresh canonical inputs are
independent and uniform.  The tested claim, batch, and frozen opportunity pool are fixed
independently of the test key.  The integrity editor makes one attempt, without key or
verification-oracle access.  A real-PRF formulation adds its distinguishing advantage.
\end{assumption}

For \(\mathcal E(\tau)=(\rho_{\mathrm{own}}(\tau),P_{\mathrm{int}}(\tau))\), owner
verification depends only on the first projection and integrity replay only on the
second.  This functional separation does not imply statistical independence during
online generation.  In particular, capacity calculations use a key-independent frozen
pool, because an early keyed rewrite can change the later trajectory.

\subsection{Ownership Guarantees}
\label{sec:owner-budget}

Fix a frozen set \(\mathcal S\) of \(D_{\mathcal B}\) distinct eligible basenames.  If
basename \(s\) appears in \(c_s\) trajectories, let
\(C_{\mathcal B}=\sum_s c_s\).  Under Assumption~\ref{ass:ideal-prf}, the selected
basename count \(M\) and expected owner rewrite count \(R\) satisfy
\begin{align}
M&\sim\operatorname{Binomial}(D_{\mathcal B},p_o),
&\mathbb E[M]&=p_oD_{\mathcal B},
\label{eq:owner-selected-budget}\\
\mathbb E[R]&=p_oC_{\mathcal B}.
\label{eq:owner-rewrite-budget}
\end{align}
% Let \(Q_m(r,u)\) be the probability that \(m\) independent draws from the 63 nonzero
% rows of \(\GF^6\) have rank \(r\) and contain \(u\) distinct rows, and write
% \(P_m(6)=\sum_uQ_m(6,u)\).  Appendix~\ref{app:owner-capacity} gives the exact
% recurrence, useful rank specializations, and the proof of the capacity result.
Let \(P_m(6)\) be the full-rank probability for \(m\) independent draws from the 63
nonzero vectors of \(\GF^6\), one draw per selected basename slot.
Appendix~\ref{app:owner-capacity} gives the exact rank and collision recurrences
and the proof of the capacity result.

% \begin{trajtheorembox}
% \begin{theorem}[Frozen-pool ownership capacity]
% \label{thm:clean-acceptance}
% Suppose the frozen-pool audit contributes one clean equation for every selected
% basename, without feedback into the pool, and the claimed label is the embedded label.
% For \(\tau_z\ge0\), let \(L_\tau=\lceil\tau_z^2\rceil\).  Under
% Assumptions~\ref{ass:canonical-replay} and~\ref{ass:ideal-prf},
% \begin{equation}
% \Pr[\mathsf{ACCEPT}]
% =\sum_{m=0}^{D_{\mathcal B}}
% \binom{D_{\mathcal B}}m p_o^m(1-p_o)^{D_{\mathcal B}-m}
% \sum_{u=L_\tau}^{63}Q_m(6,u).
% \label{eq:clean-acceptance}
% \end{equation}
% The corresponding rank-six capacity is
% \begin{equation}
% P_{\mathrm{rank6}}(D_{\mathcal B},p_o)=
% \sum_{m=0}^{D_{\mathcal B}}
% \binom{D_{\mathcal B}}m p_o^m(1-p_o)^{D_{\mathcal B}-m}P_m(6).
% \label{eq:rank-six-probability}
% \end{equation}
% Thus clean payload recovery requires rank six, while clean acceptance additionally
% requires at least \(L_\tau\) distinct retained rows.
% \end{theorem}
% \end{trajtheorembox}
% 
% At the evaluated \(\tau_z=3.09\), the evidence gate requires ten distinct rows.  Rank is
% therefore a capacity condition rather than the complete attribution decision.
\begin{trajtheorembox}
\begin{theorem}[Frozen-pool ownership capacity]
\label{thm:clean-acceptance}
Suppose the frozen-pool audit contributes one clean equation for every selected
basename, without feedback into the pool, and the claimed label is the embedded label.
For \(\tau_z\ge0\), let \(L_\tau=\max\{6,\lceil\tau_z^2\rceil\}\).  Under
Assumptions~\ref{ass:canonical-replay} and~\ref{ass:ideal-prf},
\begin{equation}
\Pr[\mathsf{ACCEPT}]
=\sum_{m=L_\tau}^{D_{\mathcal B}}
\binom{D_{\mathcal B}}m p_o^m(1-p_o)^{D_{\mathcal B}-m}P_m(6),
\label{eq:clean-acceptance}
\end{equation}
where an empty sum is zero.  The corresponding rank-six capacity is
\begin{equation}
P_{\mathrm{rank6}}(D_{\mathcal B},p_o)=
\sum_{m=0}^{D_{\mathcal B}}
\binom{D_{\mathcal B}}m p_o^m(1-p_o)^{D_{\mathcal B}-m}P_m(6).
\label{eq:rank-six-probability}
\end{equation}
In a clean batch, \(U=M\); for \(M>0\), \(z_{\mathrm{owner}}=\sqrt M\).
Acceptance requires rank six and at least \(L_\tau\) retained basename slots.
\end{theorem}
\end{trajtheorembox}

At the evaluated \(\tau_z=3.09\), the evidence gate requires ten retained basename
slots.  Their coefficient vectors may repeat, but the matrix must still have rank six.

% \begin{trajtheorembox}
% \begin{theorem}[Owner robustness and fixed-claim soundness]
% \label{thm:owner-guarantees}
% \label{thm:owner-distance}
% \label{thm:owner-null}
% Let \(A\in\GF^{U\times6}\) and \(y\in\GF^U\) be the retained system after replay,
% row aggregation, and tie erasure.
% 
% \emph{(i) Robust recovery.}  Define
% \begin{equation}
% d(A)=\min_{v\in\GF^6\setminus\{0\}}\|Av\|_0.
% \label{eq:owner-distance}
% \end{equation}
% If \(A\) has rank six, \(y=A\ownerid\oplus\eta\), and
% \(2\|\eta\|_0<d(A)\), then \(\ownerid\) is the unique maximum-satisfaction decoder
% output.  If \(s\) rows of an original matrix \(A_0\) are erased, the sufficient
% original-code condition is \(2\|\eta\|_0+s<d(A_0)\).
% 
% \emph{(ii) Fixed-claim soundness for the masked profile.}  Fix the visible stream and claim
% \(\mathbf w_c\) independently of a wrong test key.  Under
% Assumption~\ref{ass:ideal-prf}, conditional on the realized selection and row pattern
% and on any retained set of \(U\ge1\) non-tied row buckets, the satisfied-row count obeys
% \(H_c\sim\operatorname{Binomial}(U,1/2)\).  Hence
% \begin{equation}
% \Pr[\mathsf{false\ ACCEPT}\mid U]
% \le 2^{-U}\sum_{h=h_\tau(U)}^U\binom Uh,
% \qquad
% h_\tau(U)=\left\lceil\frac U2+\frac{\tau_z\sqrt U}{2}\right\rceil.
% \label{eq:owner-false-accept}
% \end{equation}
% \end{theorem}
% \end{trajtheorembox}
\begin{trajtheorembox}
\begin{theorem}[Owner robustness and fixed-claim soundness]
\label{thm:owner-guarantees}
\label{thm:owner-distance}
\label{thm:owner-null}
Let \(A\in\GF^{U\times6}\) and \(y\in\GF^U\) retain one equation per basename
slot after replay, majority aggregation, and tie erasure; coefficient vectors may repeat.

\emph{(i) Robust recovery.}  Define
\begin{equation}
d(A)=\min_{v\in\GF^6\setminus\{0\}}\|Av\|_0.
\label{eq:owner-distance}
\end{equation}
If \(A\) has rank six and \(y=A\ownerid\oplus\eta\), then
\(2\|\eta\|_0<d(A)\) ensures unique recovery of \(\ownerid\).
With \(s\) slot erasures from \(A_0\), a sufficient condition is
\(2\|\eta\|_0+s<d(A_0)\).

\emph{(ii) Fixed-claim soundness for the masked profile.}  Fix the visible stream and
claim \(\mathbf w_c\) independently of a wrong test key.  Suppose owner-mask outputs
do not affect role assignment, carrier selection, or basename-slot grouping.
Under Assumption~\ref{ass:ideal-prf}, conditional on all other domain outputs and
\(U\ge1\) retained slots, \(H_c:=H(\mathbf w_c)\sim\operatorname{Binomial}(U,1/2)\).
This holds even when coefficient vectors repeat.  Hence
\begin{equation}
\Pr[\mathsf{false\ ACCEPT}\mid U]
\le 2^{-U}\sum_{h=h_\tau(U)}^U\binom Uh,
\qquad
h_\tau(U)=\left\lceil\frac U2+\frac{\tau_z\sqrt U}{2}\right\rceil.
\label{eq:owner-false-accept}
\end{equation}
\end{theorem}
\end{trajtheorembox}

Part (i) guarantees payload recovery, not the separate claim and evidence gates.
Part (ii) is an exact finite-sample, fixed-claim bound; the normal score
\(1-\Phi(z_{\mathrm{owner}})\) is only a descriptive approximation.
% Appendix~\ref{app:owner-decoding} proves both parts using the Hamming-distance and
% bucket-symmetry arguments.
Appendix~\ref{app:owner-decoding} proves both parts by Hamming distance and
basename-mask symmetry.

\subsection{Projected-Integrity Guarantee}
\label{sec:failure-probabilities}

Because 12 does not divide \(2^\lambda\), the largest probability mass of one Q12
subtype is
\begin{equation}
q_{12,\lambda}=
\frac{\lceil2^\lambda/12\rceil}{2^\lambda}
\le\frac1{12}+2^{-\lambda}.
\label{eq:q12-accidental-match}
\end{equation}

\begin{trajtheorembox}
\begin{theorem}[Projected-integrity guarantees]
\label{prop:clean-replay}
\label{thm:integrity-localization}
Under Assumption~\ref{ass:canonical-replay}, an unchanged marked trajectory whose
expected carriers execute successfully replays cleanly.

Now consider a one-shot, role-preserving substitution that changes exactly one completed
ordinary segment \(S_k\), while boundaries, carrier positions, subtypes, and successful
statuses remain fixed.  If its recomputed ordinary input is fresh relative to every
prior same-key ordinary-domain query, then under Assumption~\ref{ass:ideal-prf},
\begin{equation}
\Pr[S_k\notin\mathcal A_{\mathrm{int}}]\le q_{12,\lambda}.
\label{eq:ordinary-localization-bound}
\end{equation}
If the ordinary role belongs to a retained completed pair and, outside a full
ordinary-digest collision, the recomputed group input is fresh relative to every prior
same-key group-domain query, then
\begin{equation}
\Pr[\text{neither the segment nor its pair is reported}]
\le 2^{-\lambda}+q_{12,\lambda}^2.
\label{eq:group-localization-bound}
\end{equation}
\end{theorem}
\end{trajtheorembox}

\begin{trajremarkbox}
\begin{remark}[Scope of the integrity bound]
No terminal factor is claimed because \(H_\theta\) has no quantitative collision model
or global terminal-query freshness assumption.  Role-changing single-site insertion,
deletion, and carrier edits are evaluated empirically.  Compound and
projection-preserving edits are outside both this theorem and the empirical integrity
claim.
\end{remark}
\end{trajremarkbox}

Appendix~\ref{app:integrity-proofs} proves clean replay and both localization clauses,
then records the boundary cases that are intentionally excluded from the probability
product.

\paragraph{Operational consequences.}
With \(m\) ordinary roles, the integrity channel emits
\begin{equation}
A_{\mathrm{seal}}(m)=m+\left\lfloor\frac m2\right\rfloor+1
\label{eq:seal-overhead}
\end{equation}
seal actions.  Ordinary alarms identify one segment, group alarms one adjacent pair,
and terminal alarms a suffix or closure inconsistency.  For total canonical byte length
% \(B\), generation and replay take expected \(O(B)\) time; owner scoring adds
% \(O(2^6U)\), with \(U\le63\).
\(B\), generation and replay take expected \(O(B)\) time; owner scoring adds
\(O(2^6U)\) for \(U\) retained basename slots.

\section{Cross-Framework Implementation}
\label{sec:implementation}

We implement \tool{} in SWE-agent, OpenHands, and OpenDev using a shared semantic
schema and versioned framework profiles.  The schema fixes the Q6/Q12 inventories,
owner-equation abstraction, and integrity commitments.  Each evaluated profile also
fixes event admission, role parsing, failure handling, closure hooks, literal hash inputs,
terminal serialization, and native action templates.  The reference contract below
states the conditions for the clean-replay analysis; the evaluated profiles have the
specific differences described in Appendix~\ref{app:framework-details}.

\subsection{Reference Adapter Contract}
\label{sec:adapter-contract}

Let \(\mathsf{Realize}_F(e)\) be framework \(F\)'s visible realization of canonical
event \(e\), and let \(\Pi_{F,\theta}\) be its public causal transducer.  For every
finite event sequence and every prefix, a conforming adapter satisfies
\begin{equation}
\Pi_{F,\theta}\!\left(
  \mathsf{Realize}_F(e_1)\cdots\mathsf{Realize}_F(e_n)\right)
=(e_1,\ldots,e_n).
\label{eq:adapter-conformance}
\end{equation}
This contract requires unique public action--observation joins, order preservation,
regular-file targets, and read-only Q6/Q12 realizations.  Missing observations map to
\(\mathsf{unknown}\), and non-success expected carriers raise alarms.  These are
reference requirements, rather than identical failure semantics in all evaluated
adapters.

\begin{table}[!htbp]
  \centering
  \caption{Native action surfaces of the evaluated framework profiles.}
  \label{tab:adapter-implementation}
  \small
  \begin{tabularx}{\textwidth}{@{}l
    >{\raggedright\arraybackslash}X
    >{\raggedright\arraybackslash}X
    >{\raggedright\arraybackslash}X@{}}
    \toprule
    Framework & Owner integration point & Integrity realization & Public verifier input \\
    \midrule
    SWE-agent
      & Agent-emitted request before computer execution
      & Read-only request executed through the ordinary path
      & Published action, outcome status, and order \\
    OpenHands
      & Agent-sourced action entering the event stream
      & Typed read action and observation executed by the runtime
      & Publicly joined action and observation events \\
    OpenDev
      & Structured tool call before dispatch
      & Read-only tool call and response appended by the dispatcher
      & Canonically indexed tool-call order and status \\
    \bottomrule
  \end{tabularx}
\end{table}

SWE-agent hooks textual agent-to-computer requests, OpenHands joins typed actions with
their observations, and OpenDev uses structured tool calls with a public index.  In all
three, owner rewriting occurs before execution and scheduled seals traverse the normal
runtime.  SWE-agent replay expects a final Q12 READ; OpenHands schedules a terminal READ
after a visible finish action, so an exit that bypasses this hook has no terminal
expectation in its evaluated replay.  OpenDev recognizes roles through a public
authenticated role/ordinal header in the tool call.  This header is retained verifier
input; auxiliary debug records are excluded.

The evaluated OpenHands profile admits only successful non-seal actions to Q12 state.
Its expected-role replay checks the next action's subtype and advances the role count,
but does not separately reject that carrier's failed observation.  OpenDev instead
records failed carriers and advances its ordinary commitment state only after a
successful ordinary carrier.  These distinctions affect clean baselines and deletion
coverage; the empirical results use each evaluated profile's paired replay.

\subsection{Configuration and Validation}
\label{sec:impl-configuration}

The evaluated configuration uses owner \(0x3f\), six owner bits,
\(\bar p_o=\OwnerSelectionProbability\), filename-once selection, uniform nonzero rows,
Q6 owner substitution, Q12 segment seals, group size two, and terminal closure.
All three adapters use the masked-orientation owner profile.  SWE-agent and OpenHands
preserve the natural READ family, while OpenDev's evaluated carrier uses the Locate
pair.  Ordinary and group payloads share one literal grammar; the OpenHands terminal
commitment uses compact sorted JSON and the other two profiles use newline-delimited
fields.

Validation covers Q6/Q12 realizations and profile-specific replay with auxiliary debug
metadata removed and native action text retained.  Fixed public vectors cover owner rows, empty segments, the \(m=0\)
terminal case, and ordinary/group/terminal payloads.  Appendix~\ref{app:framework-details}
documents the adapter contracts, while Appendix~\ref{app:watermark-details} shows
sanitized excerpts of the resulting released trajectories.

\section{Evaluation Methodology}
\label{sec:setup}

This section describes the empirical protocol used to evaluate \tool{}. The
probabilistic behavior of the owner equations and the integrity seals is analyzed in
Section~\ref{sec:analysis}; here we specify the experimental subjects, comparison
conditions, parameter choices, measurements, attacks, and artifact gates. Parameter
selection uses a separate offline carrier-capacity analysis. The three claim-bearing
evidence scopes are the primary full-watermark study, matched comparisons and ablations,
and a paired cross-agent and cross-task generalization study with No-WM and owner-WM
grids.

\subsection{Research Questions}
\label{sec:setup-rqs}

Our evaluation asks five questions.

\begin{description}[leftmargin=3.7em,style=nextline]
  \item[\textbf{RQ1: Ownership.}]
  Can the visible-only decoder recover the exact six-bit owner from batches of
  complete watermarked trajectories, and how does recovery change with batch size?

  \item[\textbf{RQ2: Robustness.}]
  Does batch-level owner recovery remain stable when visible actions are randomly
  inserted, deleted, or replaced?

  \item[\textbf{RQ3: Integrity.}]
  Do untouched trajectories replay without alarms, and do edits to protected payloads
  or seal carriers induce a new Q12 alarm under single-site and proportional attacks?

  \item[\textbf{RQ4: Cost and utility.}]
  How many visible actions do the two layers add, and do marked runs preserve the
  agent's repository-level task performance on the same tasks?

  \item[\textbf{RQ5: Comparison and generalization.}]
  How does \tool{} compare with matched controls and behavioral baselines, and does
  natural owner-carrier capacity remain available across agent, model-provider, and
  task-family boundaries?
\end{description}

RQ1 and RQ2 use a batch as the unit of attribution because the owner layer is
deliberately sparse. RQ3 uses a trajectory and its protected segments as the unit of
integrity detection. RQ4 is interpreted only within matched agent, model, task-list,
and seed conditions; raw action counts from different agent interfaces are not treated
as equivalent computational units.

\subsection{Evaluation Design and Scope}
\label{sec:setup-scope}

Table~\ref{tab:evaluation-scope} summarizes the three empirical scopes. The
\emph{primary full-WM study} evaluates trajectories generated with both the Q6 owner
layer and the Q12 ordinary, group, and terminal seals. It is the evidence source for
the paper's ownership, integrity-detection, and overhead claims. Matched controls
and ablations use the same frozen task lists wherever complete live artifacts are
available. The separate \emph{cross-pair generalization study} first runs a No-WM
control grid and then enables only the owner layer while varying agent, provider, and
task family at seed 42. The paired grids distinguish unsupported agent--task
conditions from limitations of natural READ embedding and batch decoding; neither
grid is used to support full-protocol integrity or utility claims.

\begin{table}[h]
  \caption{Evaluation scopes. The paired 27-cell grids cross every agent--task
  pairing and are analyzed separately from the primary full-WM study.}
  \label{tab:evaluation-scope}
  \small
  \begin{tabularx}{\columnwidth}{@{}p{0.22\columnwidth}X p{0.25\columnwidth}@{}}
    \toprule
    \textbf{Scope} & \textbf{Design} & \textbf{Claims} \\
    \midrule
    Primary full-WM
      & Three agents and three model providers; three fixed 50-task lists
        (seeds 42, 45, and 48) for each aligned agent--model condition; both layers
        enabled
      & RQ1--RQ4 \\
    Matched comparisons
      & No-WM, owner-only, two live behavioral baselines, and full \tool{} on
        identical task identifiers where artifacts pass the gates
      & RQ4--RQ5 \\
    Crossed grids
      & Paired No-WM and owner-WM grids over \(3\) agents \(\times\) \(3\)
        providers \(\times\) \(3\) task families, seed 42, 50-task batches
      & Carrier availability in RQ5 \\
    \bottomrule
  \end{tabularx}
\end{table}

This separation prevents two common attribution errors. First, owner recovery is
evaluated with Q6 groups, while integrity replay is evaluated separately with Q12.
Second, an unsuccessful downstream task outcome does not by itself invalidate
watermark evidence when the complete visible trajectory and protocol carriers are
present. We therefore report watermark validity, task outcome, and technical artifact
completeness as distinct fields.

\subsection{Agents, Models, Tasks, and Seeds}
\label{sec:setup-subjects}

We evaluate three independently engineered coding-agent frameworks: SWE-agent,
OpenHands, and the Rust implementation of OpenDev\footnote{\url{https://github.com/opendev-to/opendev}}.The initial JavaScript full-watermark batches used a locally modified
Rust implementation of OpenDev based on upstream commit
\texttt{e56cf85769a1}. A subset of subsequent local Owner-only
generalization reruns used upstream commit \texttt{a90fbda26e8c}.
Their released action surfaces differ substantially:
SWE-agent records text-oriented shell and ACI actions, OpenHands publishes typed
actions and observations, and OpenDev records structured tool calls. This variation is
central to the study because the detector is restricted to the visible representation
published by each framework.

The primary full-WM study pairs each framework with the task family used by its
aligned implementation:

\begin{itemize}
  \item \textbf{SWE-agent} uses 50-instance lists from SWE-bench Python
        \cite{jimenez2024swebench,yang2024sweagent};
  \item \textbf{OpenHands} uses 50-instance lists from the Java partition of
        SWE-PolyBench \cite{wang2024openhands,rashid2025swepolybench}; and
  \item \textbf{OpenDev} uses 50-instance lists from the JavaScript partition of
        SWE-PolyBench.
\end{itemize}

For each aligned agent we evaluate DeepSeek V4 Flash, GPT-5 mini, and MiniMax M3 on
the frozen task-list seeds 42, 45, and 48. Task identifiers are fixed before live
generation and reused within every matched comparison. We do not derive the reported
denominator from this nominal grid: a group enters a result only after the artifact and
method gates in Section~\ref{sec:setup-gates} pass, and each table reports its actual
eligible denominator.

The cross-pair generalization study uses the same three agents and providers but
crosses each of them with all three task families: SWE-bench Python, SWE-PolyBench
Java, and SWE-PolyBench JavaScript. Both 27-cell grids use seed 42 and frozen
50-instance inputs. The No-WM grid verifies that every crossed agent--task condition
materializes a usable visible trajectory; the owner-WM grid then asks whether the
owner mechanism can use the natural filenames exposed by that condition. Because
integrity is disabled, these cells are never pooled with the primary full-WM
localization results.

\subsection{Parameter Configuration and Selection}
\label{sec:setup-parameters}

We tune only the owner-selection probability \(p_o\). The owner payload, equation
construction, carrier alphabets, baseline profiles, and integrity schedule are fixed
across agents, models, task families, and attack strengths.
Table~\ref{tab:parameter-configuration} reports the density-selection evidence.

\begin{table}[!htbp]
  \caption{Owner-density parameter selection based on carrier-capacity analysis
  of fixed No-WM trajectories. A condition passes the rank-six, exact-owner,
  and batch-size-30 recovery criteria. We select the smallest \(p_o\) among
  the five candidates shown that passes all nine conditions.}
  \label{tab:parameter-configuration}
  \centering
  \small
  \begin{tabular}{@{}crrrl@{}}
    \toprule
    \(p_o\) (\%) & Rank-6 cells & Passing conditions &
    Min. \(B{=}30\) recovery & Decision \\
    \midrule
    15 & 27/27 & 6/9 & 0.440 & Fail \\
    20 & 27/27 & 7/9 & 0.771 & Fail \\
    25 & 27/27 & 7/9 & 0.928 & Fail \\
    \textbf{30} & \textbf{27/27} & \textbf{9/9} & \textbf{0.970} & \textbf{Pass} \\
    45 & 27/27 & 9/9 & 0.995 & Pass \\
    \bottomrule
  \end{tabular}
\end{table}

\paragraph{Owner-density selection.}
Using fixed No-WM trajectories, we evaluate carrier capacity at
\(p_o\in\{0.15,\allowbreak 0.20,\allowbreak 0.25,\allowbreak 0.30,\allowbreak 0.45\}\) by constructing owner equations from
the observed eligible READ opportunities. We hold the task identifiers,
owner identifier, key derivation, filename-once rule, equation family, and decoding
procedure fixed. Each rate is evaluated over 27 seed-level cells grouped into nine
agent--LLM conditions. We record rank-six coverage, the number of conditions that
pass all selection gates, and the minimum condition-level exact-owner recovery rate
at batch size 30.
These density-selection artifacts are not pooled into the primary detection tables.

Selection is constraint-based rather than based on the largest raw hit count. A
candidate passes only when all 27 seed-level cells reach rank six and recover the
exact owner, and every agent--LLM condition has a pooled batch-size-30 recovery rate
of at least \(0.95\). Each condition pools 1,000 resampling trials from each of
the three seeds.
Table~\ref{tab:parameter-configuration} shows that rank six alone is insufficient:
\(p_o\leq 0.25\) leaves at least two
conditions below the complete gate and has a minimum \(B{=}30\) recovery below
\(0.95\). Both \(0.30\) and \(0.45\) pass all nine conditions. We therefore select
\(p_o=0.30\), the smallest passing density among the five candidates, to minimize
intervention in natural READs.

\paragraph{Transfer to other agents, models, and task families.}
We do not retune \(p_o\) for each deployment condition. Instead, for every cell in the
No-WM generalization grid we apply the same visible canonicalizer and count
\(D_{\mathcal B}\), the distinct normalized basenames in eligible natural READs, and
\(C_{\mathcal B}\), the filename-once carrier opportunities. With \(p_o=0.30\), these
counts determine the expected number of selected equation generators and rewrites via
Eqs.~\eqref{eq:owner-selected-budget} and~\eqref{eq:owner-rewrite-budget}, and the
idealized rank-six probability via Eq.~\eqref{eq:rank-six-probability}. This capacity
audit uses natural No-WM behavior; task utility and owner recovery are measured in the
primary full-WM groups. The same \(p_o=0.30\), without per-cell adjustment, reaches rank six and the
exact owner in all 27 agent--provider--seed groups. The combination of the No-WM
capacity audit and the primary full-WM result is the basis for using \(p_o=0.30\)
throughout the evaluated settings.

\paragraph{Baseline operating points.}
ActHook-style uses a trajectory-level selection ratio of \(R=0.05\), following
ActHook's default watermark ratio, and adds exactly one keyed, read-only hook to a
selected trajectory.
\agentmarku{} has no density parameter: it uses every eligible existing Q6 READ and
changes only its visible subtype. This dense profile gives the adaptation its maximum
available equation budget while preserving zero added actions. Both profiles are
frozen before live runs and remain unchanged across agents and models.

\paragraph{Integrity operating point.}
The integrity layer is configured from the design analysis rather than selected on
the attack results. Q12 is the largest READ-subclass alphabet implemented consistently
by all three adapters and gives an approximate \(1/12\) accidental match probability
for an altered keyed seal. Ordinary seals protect segments ending at the fixed critical
classes GENERATE\_FIX, REFACTOR, RUN\_TESTS, and REPRODUCE; consecutive critical
actions are coalesced to avoid sealing every step in a burst. A group protector commits
to each pair of ordinary seals, so \(G=2\) keeps every ordinary seal near a second keyed
check while adding approximately one protector per two segments. The terminal seal
commits to the final suffix, seal count, and observed seal sequence. Detection reports
whether replay introduces an ordinary, group, terminal, missing-carrier,
failed-carrier, or unexpected-carrier alarm relative to the clean trajectory.

\subsection{Methods and Comparison Conditions}
\label{sec:setup-conditions}

The full method uses the frozen configuration in
Section~\ref{sec:impl-configuration}. On matched task identifiers, we compare
an unwatermarked control, two prior-method baselines, an ownership-only
ablation, and the complete \tool{} method.

\begin{description}[leftmargin=2.8em,style=nextline]
  \item[\textbf{Unwatermarked control (No-WM).}]
  The same agent--LLM pair is executed on the matched tasks with watermarking
  completely disabled. The resulting trajectories contain no owner-channel
  encoding or integrity seals. This condition provides the negative control
  and the reference for measuring utility and execution overhead.

  \item[\textbf{Prior-method baseline (ActHook-style).}]
  A live visible-action analogue of ActHook \cite{meng2026acthook}. A keyed
  subset of eligible runs receives a recognizable read-only hook action, and
  the detector checks whether visible hook evidence survives. This condition
  implements ActHook's action-insertion principle in the live visible-action
  setting.

  \item[\textbf{Prior-method baseline (\agentmarku{}).}]
  An adaptation of AgentMark \cite{huang2026agentmark} to the native tool-dispatch
  interface. AgentMark elicits an explicit planning-time probability list \(P_t\)
  from the LLM and uses it for sampling and decoding. The evaluated native interfaces
  provide selected tool calls rather than this probability list, so
  \agentmarku{} applies keyed uniform selection over a public set of
  functionally equivalent visible choices. An honest abstention caused by too
  few opportunities remains a method outcome. It uses every eligible Q6 READ
  and therefore operates at the maximum carrier density available to the
  adaptation rather than at a tuned sparse operating point.

  \item[\textbf{Ownership-only ablation (Owner-only).}]
  This ablation of \tool{} enables the filename-linear Q6 owner-ID channel
  while disabling the Q12 integrity seals, group-of-two protectors, and
  terminal commitment. It isolates the owner-ID recovery provided by the Q6
  channel, including its property of adding no trajectory actions.

  \item[\textbf{Complete method (\tool{}).}]
  The complete method combines the filename-linear Q6 owner-ID channel with
  Q12 ordinary segment seals, group-of-two protectors, and a terminal
  commitment.
\end{description}

All comparison methods run during agent generation rather than being added to a
No-WM log after the fact. This matters because an inserted action or altered READ can
change the next observation seen by the model. We compare only conditions with the
same agent, model, seed, and task identifiers; utility comparisons remain within the
same agent--model row.

\subsection{Metrics}
\label{sec:setup-metrics}

\paragraph{Ownership recovery.}
% For every batch we report the number of visible equation observations, the number of
% unique coefficient rows after basename and row aggregation, algebraic rank, decoded
% owner, and the method-specific decoder outcome. For \agentmarku{}, exact recovery
% requires a consistent rank-six GF(2) system whose solution is \(0x3f\). For Owner-only
% and \tool{}, it requires a unique maximum-consensus solution equal to \(0x3f\).
% The score in Eq.~\eqref{eq:owner-z} reports a descriptive evidence margin.
% Batch-size curves preserve the frozen trajectory order
% or use a stated resampling procedure and report abstention when rank is insufficient.
For every batch we report raw equation observations, retained basename slots
(majority aggregation with ties erased), distinct coefficient vectors, rank, decoded
owner, and decoder outcome. For \agentmarku{}, exact recovery
requires a consistent rank-six GF(2) system whose solution is \(0x3f\). For Owner-only
and \tool{}, it requires a unique maximum-consensus solution equal to \(0x3f\).
The score in Eq.~\eqref{eq:owner-z} reports a descriptive evidence margin.
Batch-size curves preserve the frozen trajectory order
or use a stated resampling procedure and report abstention when rank is insufficient.

\paragraph{Cross-method carrier evidence.}
For the matched detection comparison, let \(N\) be the number of visible binary
carrier observations and \(H\) the number that match the method's keyed target.
We report
\[
  z=\frac{H-N/2}{\sqrt{N/4}},\qquad
  p=1-\Phi(z),\qquad
  \mathrm{Hit}=H/N.
\]
For ActHook-style, an observation is an exact keyed-hook match; for
\agentmarku{}, Owner-only, and \tool{}, it is an equation satisfied by the claimed
owner. Within each agent--LLM cell, Table~\ref{tab:baseline-current} reports the mean
seed-level \(z\), the largest normal-reference \(p\)-value across seeds, and the pooled
\(\sum H/\sum N\). Following the random-agreement reporting convention used in prior
agent-watermarking evaluations \cite{an2026seqwm}, the \(z\)-score and one-sided
reference \(p\)-value are standardized summaries of accumulated carrier evidence.
They are not used as a cross-method efficiency test or as the operative group-level
decision rule; each method's decision rule is defined in
Section~\ref{sec:setup-conditions}.

\paragraph{Owner robustness.}
For \agentmarku{}, Owner-only, and \tool{}, the unit of success is a complete
corrupted batch whose decoded identifier remains exactly \(0x3f\). The robustness
curves retain the three LLM columns. For random visible-action corruption, every
method--attack--LLM--corruption point contains 4,500 corrupted batches (500 trials
for each of three agents and three seeds). ActHook-style appears only as a
detection-only reference because it encodes hook presence rather than a six-bit
% identifier. We additionally evaluate exact-budget white-box carrier deletion. For
% each 50-trajectory batch, the attacker deletes
% \(k=\lfloor\delta N+0.5\rfloor\) owner-carrier occurrences. For an ID-bearing
% method it exhaustively checks the 63 nonzero six-bit null directions and chooses
% the minimum-cost cut that makes the retained equation system rank-deficient; Q12
% integrity seals are excluded. Each white-box point therefore contains nine
% deterministic seed batches per LLM rather than Monte Carlo replicates.
identifier. We additionally evaluate exact-budget white-box carrier deletion. In
each 50-trajectory batch, \(N\) is the number of clean owner-carrier occurrences
and \(k=\lfloor\delta N+0.5\rfloor\). The attacker constructs up to 63
deterministic deletion sets using the clean coefficient rows, or hook orderings
for ActHook-style, with clean Q12 seal positions excluded. Each candidate removes
exactly \(k\) action atoms from the visible trajectories, including cause-linked
observations for OpenHands. We then re-extract the evidence and rerun the native
decoder, reporting whether recovery or hook detection survives the strongest
tested candidate. Appendix~\ref{app:ownership-results} specifies candidate
construction. Each point contains nine deterministic seed batches per LLM;
candidate searches are not counted as additional trials.

The disaggregated robustness figures retain all nine agent--LLM cells and the three
primary seeds in each cell. For every random attack and corruption level, 500
deterministic trials are run per seed. Each plotted random-attack point therefore
contains 1,500 corrupted \(B=50\) batches and uses the method-specific exact owner-ID
recovery criterion above, with ActHook-style as a binary hook-detection reference.

\paragraph{Integrity detection.}
Clean replay is the fraction of untouched full-WM trajectories with no ordinary,
group, terminal, missing-carrier, failed-carrier, or unexpected-carrier alarm. Attack
trials require a complete, classifiable paired-clean replay baseline and an eligible
modification site. Alarms already present in that baseline do not count as
attack-induced detections. For exhaustive
single-site attacks, we report whether a payload insertion, deletion, replacement, or
seal removal induces at least one new Q12 protocol alarm. For proportional attacks,
we apply the same any-alarm decision to random insertion, deletion, and replacement,
and to targeted Q12-seal removal under the same nominal visible-action budget. The
seal-removal attack spends its budget only on existing ordinary, group, and terminal
seals, so its realized deletion rate saturates once all seals have been removed. These
any-alarm rates are separate from localization. For the disaggregated results in
Table~\ref{tab:localization-summary}, the agent-native Q12 report set is fixed from
paired clean and attacked replay before the modified positions are supplied.
Each report is mapped to an inclusive interval \([\ell,r]\) of paired-clean visible
actions. For every eligible payload modification, the accepted hierarchy contains its
ordinary segment, the corresponding enclosing group when present, and the terminal
region. We select the narrowest newly reported mapped interval in this hierarchy that
geometrically covers the modified action or insertion boundary; its width is
\(r-\ell+1\). A site with no accepted covering report is a localization miss: it
remains in the coverage denominator and does not enter the conditional width mean or
median. The one-action column uses exactly one selected action or insertion boundary
and 50 deterministic trials per trajectory and attack. Percentage columns retain the
frozen proportional schedules. Randomly selected Q12 seal actions remain in deletion
and replacement mutants but are excluded from the payload-site denominator; targeted
seal removal is evaluated only by the paired-clean any-alarm decision. Each displayed
condition--attack cell micro-pools eligible sites over the three primary seeds; the
overall row additionally pools the three random attack types. The interval construction
follows the ordinary, group, and terminal replay units in
Appendix~\ref{app:protocol-details}.

In the full-watermark strength sweep, integrity is evaluated per attacked
trajectory rather than per batch. Each agent--LLM--attack--corruption point contains
7,500 trajectory trials: 50 trajectories, 50 deterministic trials, and three seeds. A trial
is detected when the attacked visible stream introduces an ordinary, group,
terminal, missing-carrier, failed-carrier, or unexpected-carrier alarm not present in
that trajectory's clean replay. This clean-replay difference prevents a genuine
runtime outcome already present in the released trajectory from being counted as an
attack-induced detection.

\paragraph{Cost and task outcome.}
We count owner rewrites separately from added integrity actions. The owner-only layer
adds no structural action; full \tool{} reports ordinary, group, and terminal READs and
their within-condition percentage overhead. Where a task-family harness is available,
we report paired benchmark resolution. We define Pass@1 as the fraction of benchmark
tasks for which the single submitted patch passes
the corresponding benchmark's repository-level test harness; empty or invalid patches
count as failures. Solve-rate comparisons pair TrajMark and No-WM runs by agent, LLM,
seed, and task identifier, and use the same frozen harness configuration. We report the
Pass@1 percentages and their difference in percentage points. To assess whether those
changes have a consistent direction independently of the pooled magnitude, we apply
an exact two-sided sign test to the nine matched agent--LLM percentage changes,
excluding ties. Token and wall-clock
measurements are reported only when the same adapter
exposes comparable accounting, because API queueing, retries, and repository commands
otherwise dominate the difference. For OpenDev records with missing native cost measurements, turns and
output tokens are taken from matched cost-only re-executions;
visible-action counts are taken from the canonical runs.

\subsection{Attack Protocol}
\label{sec:setup-attacks}

We use separate attack suites for the robust owner layer and the fragile integrity
layer.

\paragraph{Random batch corruption.}
For RQ2, an attacker inserts, deletes, or replaces a fraction
\(\delta\in\{0.05,0.10,0.20,0.30,0.50\}\) of visible actions. Positions are sampled
without access to the secret key. Insertions use syntactically valid visible actions;
deletions remove observed actions; replacements change the visible action
representation. The appropriate detector is rerun on every corrupted batch: exact
six-bit recovery for \agentmarku{}, Owner-only, and \tool{}. The owner-ID baseline
comparison uses matched groups from SWE-agent, OpenHands, and OpenDev, retaining
identical task identifiers and seeds within each agent--LLM condition. ActHook-style is evaluated
separately as a binary hook detector. This suite measures tolerance to non-targeted
log damage, not an oracle attack that knows the selected filename equations.

The cross-agent strength sweep applies the same five values of \(\delta\) to the
complete full-WM trajectories from every agent--LLM--seed group. In each trial, the
requested fraction is rounded separately for each trajectory and sampled from that
trajectory's released visible actions. Insertions and replacements draw non-protocol
actions from the same clean seed group. The owner detector then re-decodes the full
50-trajectory batch, while the integrity detector evaluates every attacked
trajectory against its own clean replay. Consequently, the sweep shows both how the
robust owner channel degrades with attack strength and how quickly the fragile layer
reacts to the same corruption.

For proportional seal removal, a trajectory with \(N\) released visible actions receives
the nominal budget
\(k=\min\{N,\max(1,\lceil\delta N\rceil)\}\). The attacker samples
without replacement from the trajectory's existing Q12 ordinary, group, and terminal
seals and removes \(\min(k,S)\) of the \(S\) available seals. Thus \(\delta\) denotes a
nominal visible-action budget; the realized deletion fraction can be smaller after the
seal set is exhausted. This attack is scored only by the paired-clean Q12 any-alarm
decision.

\paragraph{Exhaustive single-site integrity attacks.}
For RQ3, we enumerate eligible attack sites in every trajectory with a complete,
classifiable paired-clean replay baseline. Payload
attacks insert a READ into a protected segment, delete a protected READ or critical
action, replace a READ subtype, or replace a critical target. Protector attacks delete
or replace each ordinary seal, group protector, and terminal commitment. Every mutant
changes one eligible site and is verified from the modified visible stream. We retain
the exact trial denominator for each attack family because trajectories expose
different numbers of protected segments and seal carriers.

Insertion, deletion, and replacement are the primary payload-tampering families.
Attacks on ordinary, group, and terminal carriers are reported separately because they
target the integrity mechanism itself. Suffix loss is evaluated through the terminal
role and its commitment to the final payload, seal counts, and observed seal sequence.

\subsection{Artifact and Method Gates}
\label{sec:setup-gates}

A run enters the analysis only when its selected identifiers match the frozen input,
provider and model identity
match the manifest, required trajectory and prediction artifacts are non-empty, and
the public protocol configuration and implementation hashes match the frozen manifest.
For full \tool{}, the owner method must be the filename-once linear Q6 protocol and
the integrity method must be the group-of-two Q12 protocol described in
Sections~\ref{sec:owner} and \ref{sec:localization}. Only runs satisfying these
method gates are pooled.

The visible-only gate removes debug records and reconstructs owner equations,
protected segments, and seals solely from the released action order. A full-WM
trajectory must provide a complete, classifiable clean replay baseline before it
contributes attack sites, and only alarm signatures newly introduced by a mutant
count as attack detection. Provider, authentication, setup, and missing-artifact
failures are treated as technical exclusions. Completed model interactions, timeouts,
and failed tests remain honest task outcomes when their visible watermark artifacts
are complete.

Finally, every reported group is regenerated from synchronized local artifacts. The
evaluation snapshot stores
the frozen input hash, runtime and detector identifiers, eligible denominators, and the
result file used to generate each table.

\paragraph{AI-assisted research workflow.}
We used ChatGPT and Codex (OpenAI) to assist with discussions of the watermarking
and verification design, implementation and debugging of the watermarking and
evaluation code, and development of scripts for aggregating recorded experimental
outputs and producing result tables and plots. OpenAI image-generation tools assisted
in preparing the method-overview illustration. AI-assisted code was checked using
automated tests and comparisons against recorded verification outputs. The numerical
results reported in this paper were computed from the recorded runs and verifier
outputs by the evaluation and analysis scripts.

\section{Evaluation Results}
\label{sec:evaluation-results}

We report exact attribution, robustness, Q12 tamper detection, and utility in that
order. Each result is kept at its valid unit of analysis; agent--LLM conditions are
not merged unless an overall value is explicitly labeled.

\subsection{Cross-Agent and Cross-Task Owner Attribution}
\label{sec:ablation-results}
\label{sec:results-generalization}

The crossed grids test whether the frozen Q6 owner layer transfers beyond the three
aligned agent--task pairings.

\begin{table}[H]
  \caption{Cross-agent and cross-task owner attribution. Each cell is one
  agent--provider--task-family batch.}
  \label{tab:generalization-current}
  \small
  \begin{tabular}{@{}lll@{}}
    \toprule
    Condition & Measure & Result \\
    \midrule
    No-WM & Valid visible batch & 27/27 \\
    Owner-only & Exact owner-ID recovery & 27/27 \\
    \bottomrule
  \end{tabular}
\end{table}

All 27 No-WM cells produce usable visible batches, and Owner-only recovers the exact
identifier \texttt{0x3f} in all 27 corresponding cells. Thus the owner layer transfers
across all tested agents, providers, and task families without per-cell retuning.

\FloatBarrier
\subsection{Trace Footprint and Runtime Cost}
\label{sec:results-execution-cost}

Table~\ref{tab:execution-cost-current} shows the realized actions, model turns, and
output tokens for each matched agent--LLM row. Complete \tool{} has more visible
actions than No-WM in all nine rows, reflecting the explicit Q12 seals;
Section~\ref{sec:results-step-cost} isolates their exact structural cost from ordinary
run-to-run variation. Model turns change by \(-9.5\) to \(+5.7\), and output tokens by
\(-1.23\)k to \(+0.62\)k, so neither increases systematically across the matched rows.
Turns and output tokens use framework-native accounting.

\begingroup
\footnotesize
\setlength{\tabcolsep}{2.4pt}
\renewcommand{\arraystretch}{1.00}
\newcommand{\costup}[1]{\textcolor[HTML]{B04A3A}{\raisebox{0.35ex}{\scriptsize$\uparrow$#1}}}
\newcommand{\costdown}[1]{\textcolor[HTML]{2A8C91}{\raisebox{0.35ex}{\scriptsize$\downarrow$#1}}}
\setlength{\LTcapwidth}{\textwidth}
\begin{longtable}{@{}p{\textwidth}@{}}
\caption{Trace footprint and runtime cost across three coding agents and three LLMs
(three matched 50-task seeds; 150 runs per full cell). Superscripts show absolute
changes from No-WM within the same agent--LLM row (cyan: lower; red: higher).
AH, AM-U, Owner, and TM denote ActHook-style, AgentMark-U, Owner-only, and
complete \tool{}.}
\label{tab:execution-cost-current} \\
\endfirsthead
\multicolumn{1}{@{}l@{}}{\small Table~\thetable{} (continued)} \\
\endhead
\endfoot
\endlastfoot
\begin{minipage}{\linewidth}
\begin{tabular*}{\linewidth}{@{\extracolsep{\fill}}llrrrrr@{}}
  \toprule
  \multicolumn{7}{c}{\textbf{Avg. Visible Actions} $\downarrow$} \\
  \cmidrule(lr){1-7}
  LLM & Agent & No-WM & AH & AM-U & Owner & TM \\
  \midrule
  \multirow{3}{*}{\shortstack[l]{DeepSeek V4\\Flash}}
    & SWE-agent & 41.1 & 40.1\costdown{1.0} & 51.0\costup{9.9} & 41.8\costup{0.7} & 47.7\costup{6.6} \\
    & OpenHands & 33.4 & 31.7\costdown{1.7} & 33.7\costup{0.4} & 32.7\costdown{0.6} & 34.4\costup{1.1} \\
    & OpenDev & 33.8 & 45.3\costup{11.5} & 47.7\costup{13.9} & 47.6\costup{13.8} & 51.3\costup{17.5} \\
  \addlinespace[1pt]
  \multirow{3}{*}{GPT-5 mini}
    & SWE-agent & 16.2 & 16.0\costdown{0.2} & 39.9\costup{23.7} & 16.6\costup{0.4} & 27.0\costup{10.8} \\
    & OpenHands & 25.1 & 25.6\costup{0.5} & 31.3\costup{6.2} & 26.7\costup{1.5} & 29.9\costup{4.7} \\
    & OpenDev & 18.3 & 21.4\costup{3.1} & 19.5\costup{1.2} & 15.7\costdown{2.6} & 19.8\costup{1.4} \\
  \addlinespace[1pt]
  \multirow{3}{*}{MiniMax M3}
    & SWE-agent & 43.1 & 41.5\costdown{1.5} & 43.5\costup{0.4} & 41.5\costdown{1.6} & 46.9\costup{3.9} \\
    & OpenHands & 34.1 & 33.4\costdown{0.7} & 35.7\costup{1.6} & 31.9\costdown{2.2} & 37.6\costup{3.5} \\
    & OpenDev & 68.6 & 72.7\costup{4.1} & 74.9\costup{6.3} & 55.0\costdown{13.6} & 73.3\costup{4.7} \\
  \bottomrule
\end{tabular*}
\end{minipage} \\
\begin{minipage}{\linewidth}
\begin{tabular*}{\linewidth}{@{\extracolsep{\fill}}llrrrrr@{}}
  \toprule
  \multicolumn{7}{c}{\textbf{Avg. Turns} $\downarrow$} \\
  \cmidrule(lr){1-7}
  LLM & Agent & No-WM & AH & AM-U & Owner & TM \\
  \midrule
  \multirow{3}{*}{\shortstack[l]{DeepSeek V4\\Flash}}
    & SWE-agent & 41.4 & 40.5\costdown{0.9} & 51.6\costup{10.2} & 42.2\costup{0.8} & 40.4\costdown{1.0} \\
    & OpenHands & 31.4 & 29.8\costdown{1.6} & 31.8\costup{0.4} & 30.7\costdown{0.7} & 30.8\costdown{0.6} \\
    & OpenDev & 40.1 & 40.0\costdown{0.2} & 41.9\costup{1.8} & 44.6\costup{4.5} & 39.4\costdown{0.7} \\
  \addlinespace[1pt]
  \multirow{3}{*}{GPT-5 mini}
    & SWE-agent & 18.1 & 17.5\costdown{0.6} & 47.7\costup{29.6} & 18.8\costup{0.7} & 23.8\costup{5.7} \\
    & OpenHands & 24.1 & 24.5\costup{0.4} & 30.3\costup{6.2} & 25.7\costup{1.5} & 26.2\costup{2.0} \\
    & OpenDev & 23.3 & 26.4\costup{3.0} & 24.5\costup{1.2} & 18.5\costdown{4.8} & 18.7\costdown{4.6} \\
  \addlinespace[1pt]
  \multirow{3}{*}{MiniMax M3}
    & SWE-agent & 44.3 & 42.7\costdown{1.6} & 44.7\costup{0.4} & 42.7\costdown{1.6} & 42.2\costdown{2.0} \\
    & OpenHands & 32.9 & 32.0\costdown{0.9} & 34.3\costup{1.4} & 30.5\costdown{2.4} & 34.7\costup{1.8} \\
    & OpenDev & 70.8 & 74.0\costup{3.2} & 72.9\costup{2.0} & 54.7\costdown{16.2} & 61.3\costdown{9.5} \\
  \bottomrule
\end{tabular*}
\end{minipage} \\
\begin{minipage}{\linewidth}
\begin{tabular*}{\linewidth}{@{\extracolsep{\fill}}llrrrrr@{}}
  \toprule
  \multicolumn{7}{c}{\textbf{Avg. Output Tokens (k)} $\downarrow$} \\
  \cmidrule(lr){1-7}
  LLM & Agent & No-WM & AH & AM-U & Owner & TM \\
  \midrule
  \multirow{3}{*}{\shortstack[l]{DeepSeek V4\\Flash}}
    & SWE-agent & 5.43 & 5.41\costdown{0.02} & 5.50\costup{0.07} & 4.81\costdown{0.62} & 5.28\costdown{0.14} \\
    & OpenHands & 7.34 & 7.38\costup{0.04} & 4.60\costdown{2.74} & 6.57\costdown{0.77} & 6.43\costdown{0.91} \\
    & OpenDev & 20.52 & 21.34\costup{0.81} & 22.05\costup{1.52} & 19.77\costdown{0.75} & 21.14\costup{0.62} \\
  \addlinespace[1pt]
  \multirow{3}{*}{GPT-5 mini}
    & SWE-agent & 0.74 & 0.69\costdown{0.05} & 1.80\costup{1.05} & 0.80\costup{0.06} & 1.13\costup{0.39} \\
    & OpenHands & 6.17 & 6.21\costup{0.03} & 4.20\costdown{1.98} & 6.46\costup{0.29} & 6.58\costup{0.41} \\
    & OpenDev & 4.09 & 4.80\costup{0.71} & 4.60\costup{0.51} & 3.59\costdown{0.50} & 3.80\costdown{0.29} \\
  \addlinespace[1pt]
  \multirow{3}{*}{MiniMax M3}
    & SWE-agent & 2.32 & 2.28\costdown{0.04} & 1.54\costdown{0.78} & 2.19\costdown{0.14} & 1.96\costdown{0.36} \\
    & OpenHands & 6.41 & 5.70\costdown{0.71} & 3.97\costdown{2.44} & 5.80\costdown{0.61} & 6.40\costdown{0.01} \\
    & OpenDev & 12.89 & 14.36\costup{1.47} & 13.38\costup{0.50} & 9.89\costdown{3.00} & 11.65\costdown{1.23} \\
  \bottomrule
\end{tabular*}
\end{minipage} \\
\end{longtable}
\endgroup

\FloatBarrier
\subsection{Behavioral Watermark Evidence}
\label{sec:results-baselines}

\begingroup
\setlength{\intextsep}{3pt}
\captionsetup{skip=4pt}
\begin{table}[H]
  \centering
  \caption{Detection evidence across three coding agents and three LLMs. Each
  cell pools three seed-matched 50-trajectory batches and reports mean seed-level
  \(z\), the largest normal-reference \(p\)-value \(1-\Phi(z)\) across seeds, and pooled
  keyed-carrier Hit\%. Exact owner recovery is reported separately in
  Section~\ref{sec:results-owner-only}.}
  \label{tab:baseline-current}
  \scriptsize
  \setlength{\tabcolsep}{2.3pt}
  \renewcommand{\arraystretch}{1.00}
  \resizebox{\textwidth}{!}{%
  \begin{tabular}{@{}ll*{3}{r>{\centering\arraybackslash}m{6.2em}r}@{}}
    \toprule
    & & \multicolumn{3}{c}{SWE-agent}
      & \multicolumn{3}{c}{OpenHands}
      & \multicolumn{3}{c}{OpenDev} \\
    \cmidrule(lr){3-5}\cmidrule(lr){6-8}\cmidrule(lr){9-11}
    LLM & Method
      & \(z\uparrow\) & \(p\downarrow\) & Hit\%\(\uparrow\)
      & \(z\uparrow\) & \(p\downarrow\) & Hit\%\(\uparrow\)
      & \(z\uparrow\) & \(p\downarrow\) & Hit\%\(\uparrow\) \\
    \midrule

    \multirow{4}{*}{\shortstack[l]{DeepSeek V4\\Flash}}
      & \acthook{}
        & 1.28 & \(1.59{\times}10^{-1}\) & 100.0
        & 1.82 & \(4.2{\times}10^{-2}\) & 100.0
        & 1.62 & \(1.59{\times}10^{-1}\) & 100.0 \\
      & \agentmarku{}
        & 29.51 & \(1.5{\times}10^{-179}\) & 100.0
        & 35.83 & \(8.4{\times}10^{-273}\) & 100.0
        & 20.88 & \(6.7{\times}10^{-92}\) & 100.0 \\
      & Owner-only
        & 6.82 & \(1.3{\times}10^{-10}\) & 100.0
        & 9.84 & \(9.0{\times}10^{-21}\) & 100.0
        & 8.55 & \(1.0{\times}10^{-13}\) & 100.0 \\
      & \textbf{\tool{} (ours)}
        & \textbf{6.68} & \(\boldsymbol{4.6{\times}10^{-11}}\) & \textbf{100.0}
        & \textbf{9.93} & \(\boldsymbol{4.3{\times}10^{-22}}\) & \textbf{100.0}
        & \textbf{9.97} & \(\boldsymbol{7.2{\times}10^{-22}}\) & \textbf{100.0} \\

    \midrule
    \multirow{4}{*}{GPT-5 mini}
      & \acthook{}
        & 1.28 & \(1.59{\times}10^{-1}\) & 100.0
        & 1.82 & \(4.2{\times}10^{-2}\) & 100.0
        & 1.62 & \(1.59{\times}10^{-1}\) & 100.0 \\
      & \agentmarku{}
        & 32.99 & \(8.1{\times}10^{-218}\) & 100.0
        & 37.45 & \(1.6{\times}10^{-294}\) & 100.0
        & 12.05 & \(5.6{\times}10^{-30}\) & 100.0 \\
      & Owner-only
        & 4.54 & \(1.9{\times}10^{-5}\) & 100.0
        & 7.85 & \(5.9{\times}10^{-12}\) & 100.0
        & 5.15 & \(2.3{\times}10^{-6}\) & 100.0 \\
      & \textbf{\tool{} (ours)}
        & \textbf{5.12} & \(\boldsymbol{8.1{\times}10^{-7}}\) & \textbf{100.0}
        & \textbf{7.78} & \(\boldsymbol{3.6{\times}10^{-14}}\) & \textbf{100.0}
        & \textbf{5.45} & \(\boldsymbol{3.2{\times}10^{-5}}\) & \textbf{100.0} \\

    \midrule
    \multirow{4}{*}{MiniMax M3}
      & \acthook{}
        & 1.28 & \(1.59{\times}10^{-1}\) & 100.0
        & 1.63 & \(7.9{\times}10^{-2}\) & 100.0
        & 1.62 & \(1.59{\times}10^{-1}\) & 100.0 \\
      & \agentmarku{}
        & 28.75 & \(3.0{\times}10^{-168}\) & 100.0
        & 34.63 & \(1.8{\times}10^{-253}\) & 100.0
        & 20.70 & \(2.8{\times}10^{-89}\) & 100.0 \\
      & Owner-only
        & 7.48 & \(6.0{\times}10^{-14}\) & 100.0
        & 8.89 & \(3.0{\times}10^{-17}\) & 100.0
        & 10.01 & \(1.1{\times}10^{-17}\) & 100.0 \\
      & \textbf{\tool{} (ours)}
        & \textbf{8.02} & \(\boldsymbol{4.7{\times}10^{-15}}\) & \textbf{100.0}
        & \textbf{9.57} & \(\boldsymbol{2.0{\times}10^{-21}}\) & \textbf{100.0}
        & \textbf{10.28} & \(\boldsymbol{5.2{\times}10^{-19}}\) & \textbf{100.0} \\

    \bottomrule
  \end{tabular}%
  }
\end{table}

\endgroup

All methods reach 100\% Hit, while their evidence volumes differ: ActHook-style yields
\(z=1.28\)--1.82, \agentmarku{} yields \(z=12.05\)--37.45, and Owner-only and
complete \tool{} yield \(z=4.54\)--10.01 and \(z=5.12\)--10.28, respectively.
Because \(z\) measures accumulated carrier evidence rather than cross-method
efficiency, Section~\ref{sec:results-owner-only} uses exact owner-ID recovery as the
decision-level comparison.

\subsection{Owner-ID Recovery and Robustness}
\label{sec:results-owner-only}

We count owner-ID recovery as successful when the method-specific group-level
decoder returns the exact six-bit owner ID \texttt{0x3f}; abstentions and incorrect
IDs count as failures. Under this criterion, \agentmarku{}, Owner-only, and complete
\tool{} recover the exact owner ID in all 27 clean groups. ActHook-style is
N/A for owner-ID recovery because hook presence does not encode an owner
identity.

\subsubsection{Owner-ID Recovery under Random Corruption and White-box Carrier Deletion}
\label{sec:results-owner-only-attacks}

Figure~\ref{fig:verified-owner-id-recovery-method-comparison} evaluates exact
owner-ID recovery for the three ID-bearing methods under random visible-action
corruption and targeted owner-carrier removal. ActHook-style is included as a
detection-only reference because it does not encode an owner ID.

\begin{figure}[!ht]
  \centering
  \includegraphics[width=\linewidth]{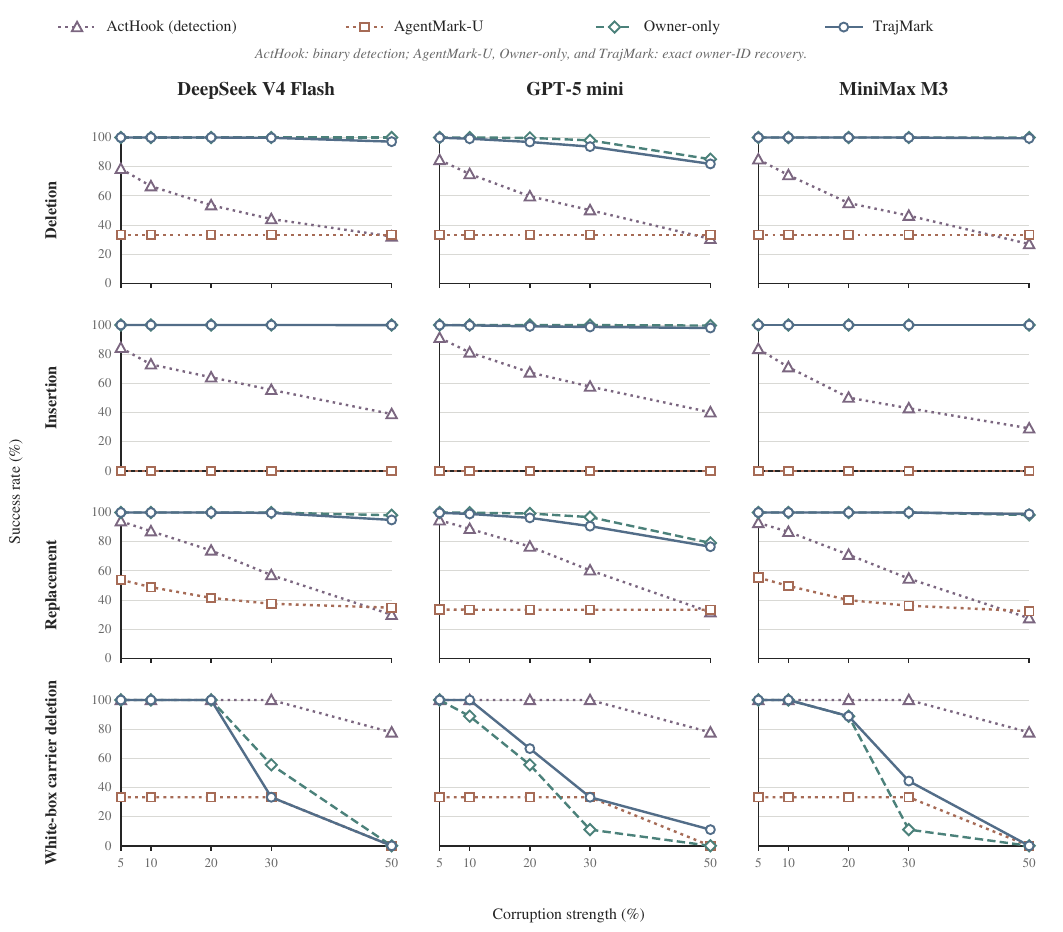}
  \Description{Twelve line charts report detection or exact owner-ID recovery
  rates. The first three rows show random visible-action deletion, insertion,
  and replacement; the fourth row shows exact-budget white-box deletion of
  owner-channel carriers. Columns are DeepSeek V4 Flash, GPT-5 mini, and
  MiniMax M3. Each panel compares ActHook-style, AgentMark-U, Owner-only, and
  TrajMark. ActHook-style reports binary detection, whereas the other three
  methods report exact owner-ID recovery.}
  \caption{Detection and exact owner-ID recovery under random visible-action
  corruption and exact-budget white-box carrier deletion. The first three rows
  aggregate 4,500 corrupted \(B=50\) batches per point. ActHook-style reports
  binary detection, whereas AgentMark-U, Owner-only, and \tool{} report exact
%   owner-ID recovery. The fourth row applies exact-budget white-box deletion to
%   owner-channel carriers; integrity-seal carriers are excluded. Agent-level
  owner-ID recovery. The fourth row reports nine \(B=50\) batches per point after
  physical carrier deletion and native re-decoding; clean integrity-seal positions
  are excluded from deletion. Agent-level
  results appear in Appendix~\ref{app:ownership-results}.}
  \label{fig:verified-owner-id-recovery-method-comparison}
\end{figure}

The filename-indexed equations used by Owner-only and \tool{} recover the owner
without adding trajectory actions. Under insertion, \agentmarku{}'s
position-indexed equations shift, whereas the filename-indexed carrier
identities used by Owner-only and \tool{} remain stable
(Figure~\ref{fig:owner-outcomes-insert-full}). ActHook-style remains a
detection-only reference at the fixed \(R=0.05\) operating point.

At 20\% random corruption, Owner-only and complete \tool{} achieve
99.3--100\% and 96.3--100\% exact owner-ID recovery, respectively. At 50\%
corruption, their minimum recovery rates are 79.1\% and 76.5\%. ActHook-style
falls to 26.9--40.2\% detection at 50\% because its binary decision depends
only on retained hook evidence.

% The exact-budget white-box attack exposes the redundancy limit. Owner-only and
% complete \tool{} fall from 85.2\% recovery at 20\% carrier deletion to 37.0\%
% at 30\% and 0\% at 50\%. The per-cell grids show that deletion broadly affects
Under white-box physical deletion, Owner-only and complete \tool{} recover the
exact ID in 22/27 and 23/27 batches at 20\% carrier deletion, 7/27 and 10/27 at
30\%, and 0/27 and 1/27 at 50\%, respectively. The per-cell grids show that
random deletion broadly affects
\agentmarku{}, while the largest filename-indexed losses under random deletion
and replacement occur in the GPT-5 mini cells
(Figures~\ref{fig:owner-outcomes-delete-full} and
\ref{fig:owner-outcomes-replace-full}). Random insertion largely preserves the
filename-indexed equations but shifts \agentmarku{}'s position index
(Figure~\ref{fig:owner-outcomes-insert-full}).
% Figure~\ref{fig:owner-outcomes-whitebox-full} reports the per-cell carrier
% budget at which rank-six recovery fails.
Figure~\ref{fig:owner-outcomes-whitebox-full} shows the native-adapter difference:
\agentmarku{} recovers no batches under white-box deletion in SWE-agent or OpenDev
from 5\% onward, whereas its OpenHands event-ID indexing preserves recovery in all
nine batches through 30\%, before falling to zero at 50\%.

\FloatBarrier
\subsection{Full-Watermark Tamper Detection}
\label{sec:local-results}

The Q6 channel supports batch-level owner recovery, whereas Q12 integrity replay
detects alarms introduced relative to the paired clean replay. Detection and localization use attack-specific eligibility conditions
determined from the paired-clean Q12 replay.

Figure~\ref{fig:integrity-detection-matrix-current} first evaluates every eligible
single-site insertion, deletion, replacement, and seal removal. Across all nine
agent--LLM cells, TrajMark detects 95.5--100\% of these edits, establishing that the
Q12 checks react reliably when one eligible site is modified.

\begin{figure}[H]
  \centering
  \includegraphics[width=0.80\linewidth]{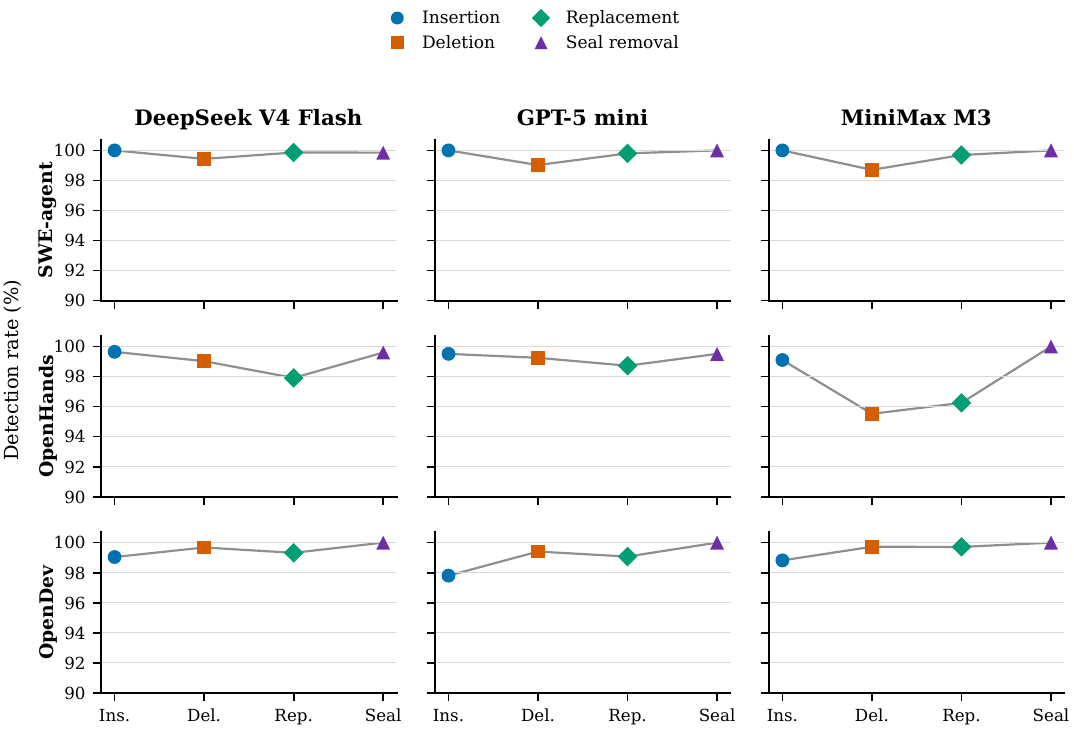}
  \Description{A three-by-three matrix of plots reports exhaustive eligible
  single-site tamper detection. Rows are SWE-agent, OpenHands, and OpenDev; columns
  are DeepSeek V4 Flash, GPT-5 mini, and MiniMax M3. Each panel contains insertion,
  deletion, replacement, and seal-removal detection rates for TrajMark.}
  \caption{Single-site integrity detection across all nine agent--LLM cells. Each
  point is the probability of at least one attack-induced Q12 alarm under exhaustive
  eligible single-site attacks, macro-averaged over the three primary seeds. The
  truncated vertical axis makes differences near 100\% visible. This is an
  any-protocol detection result.}
  \label{fig:integrity-detection-matrix-current}
\end{figure}

Figure~\ref{fig:proportional-tamper-detection-current} then applies the same any-alarm
decision under proportional attacks. Insertion, deletion, and replacement randomly
modify 5--50\% of each trajectory's visible actions; seal removal spends the same
nominal budget only on Q12 seals. Every agent--LLM condition remains separate.

\begin{figure}[!ht]
  \centering
  \includegraphics[width=\linewidth]{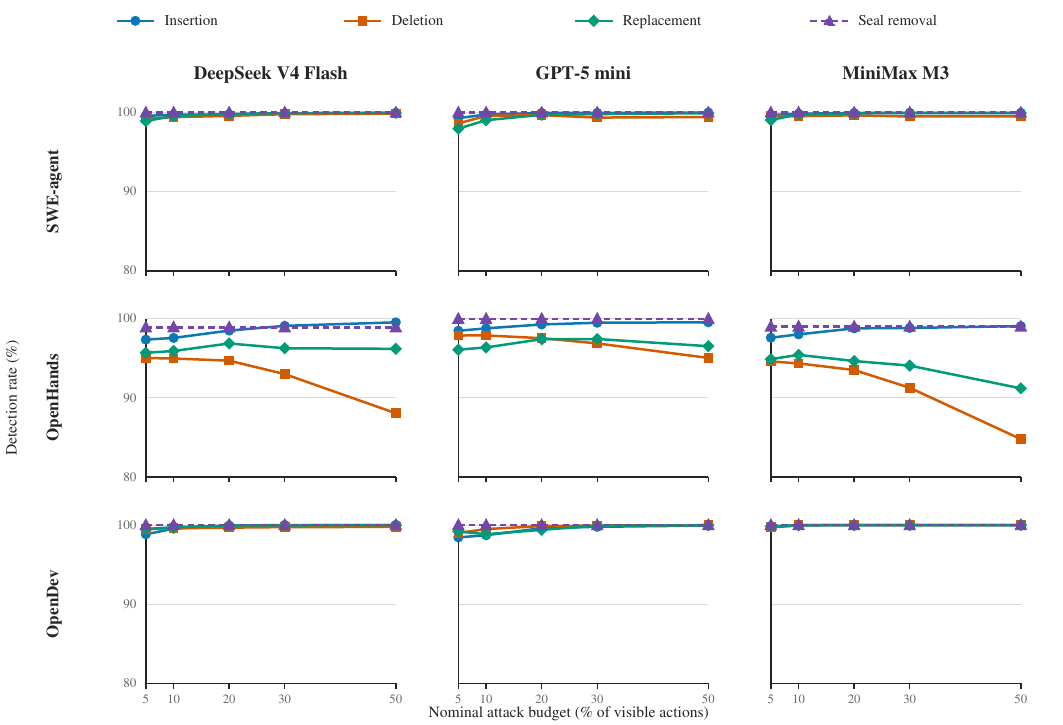}
  \Description{A three-by-three matrix reports tamper detection rates under
  proportional visible-action attacks. Rows are SWE-agent,
  OpenHands, and OpenDev; columns are DeepSeek V4 Flash, GPT-5 mini, and MiniMax M3.
  Each panel contains random insertion, deletion, and replacement curves plus a
  targeted Q12-seal-removal curve at nominal budgets of 5, 10, 20, 30, and 50 percent.}
  \caption{Tamper detection under proportional visible-action attacks across all nine
  agent--LLM conditions. The first three curves randomly insert, delete, or replace
  visible actions; seal removal directs the same nominal budget only to existing Q12
  seals and saturates once all seals are removed. Detection means that at least one
  attack-induced Q12 alarm appears relative to the paired clean replay. Each point
  macro-averages the three primary-seed rates; each seed uses 50 trajectories and 50
  deterministic trials. The denominator contains trajectories with at least one
  eligible random modification or at least one removed seal.}
  \label{fig:proportional-tamper-detection-current}
\end{figure}

At 50\% corruption, all SWE-agent and OpenDev curves remain at or above 99.4\%.
OpenHands remains at 99.1--99.6\% for insertion, but falls to 88.1\%, 95.0\%, and
84.8\% for deletion and to 96.2\%, 96.5\%, and 91.2\% for replacement across the
three LLMs. Insertion preserves the existing Q12 evidence, whereas deletion can
remove both protected content and some of its checks. The figure therefore supports
robust any-alarm detection, with OpenHands deletion as the limiting condition; it
evaluates detection rather than exact localization. Targeted seal removal is detected
in 100\% of SWE-agent and OpenDev trials and in 98.9\%, 100\%, and 99.0\% of
OpenHands trials across the three LLMs, respectively, at every nominal budget. At
50\%, all available Q12 seals are removed in every condition. The flat seal-removal
curves reflect the binary any-alarm decision: once a detectable trajectory raises an
alarm, removing additional seals cannot increase its per-trajectory score.

\FloatBarrier

Table~\ref{tab:localization-summary} then reports where those alarms point. It keeps
all LLM, agent, and random-corruption conditions separate and pairs accepted-site
coverage with the width of the narrowest reported ordinary, group, or terminal region
that covers each localized payload modification.

% Source artifacts:
%   localization_payload_site_minimum_interval_v4/*.json
% Each cell is recomputed from integer totals across seeds 42/45/48:
%   coverage = sum(accepted_localized_payload_sites) / sum(eligible_payload_sites)
%   mean width = sum(accepted_minimum_width_sum) / sum(accepted_localized_payload_sites)
\begin{table}[!t]
  \centering
  \caption{Accepted-site localization by LLM, coding agent, and random
  corruption type. Each cell reports site coverage (\%; $\uparrow$) / conditional
  mean width (paired-clean visible actions; $\downarrow$) for the narrowest newly
  reported Q12 region in the modified site's accepted ordinary, group, or terminal
  hierarchy. Coverage retains eligible misses in its denominator; width is averaged
  only over covered sites.}
  \label{tab:localization-summary}
  \begingroup
  \small
  \setlength{\tabcolsep}{2.35pt}
  \renewcommand{\arraystretch}{0.94}
  \newcommand{\loccell}[2]{#1\,/\,#2}
  \begin{adjustbox}{max width=\linewidth}
  \begin{tabular}{@{}lllrrrrrr@{}}
    \toprule
    LLM & Agent & Attack
      & \multicolumn{6}{c}{Site coverage (\%) $\uparrow$ / conditional mean width $\downarrow$} \\
    \cmidrule(lr){4-9}
      & & & 1 action & 5\% & 10\% & 20\% & 30\% & 50\% \\
    \midrule
    \multirow{9}{*}{\shortstack[l]{DeepSeek V4\\Flash}}
      & \multirow{3}{*}{SWE-agent}
        & Insertion   & \loccell{98.1}{21.7} & \loccell{97.9}{24.7} & \loccell{97.6}{25.5} & \loccell{97.5}{26.7} & \loccell{97.4}{27.9} & \loccell{97.7}{30.1} \\
      & & Deletion    & \loccell{96.8}{27.5} & \loccell{96.6}{31.6} & \loccell{95.7}{33.1} & \loccell{95.2}{37.3} & \loccell{94.8}{40.9} & \loccell{92.7}{47.0} \\
      & & Replacement & \loccell{90.0}{24.9} & \loccell{96.1}{30.8} & \loccell{96.4}{31.4} & \loccell{96.3}{34.3} & \loccell{96.3}{38.0} & \loccell{96.5}{45.6} \\
    \addlinespace[1pt]
      & \multirow{3}{*}{OpenHands}
        & Insertion   & \loccell{85.1}{20.2} & \loccell{81.6}{22.2} & \loccell{77.4}{22.6} & \loccell{69.8}{22.5} & \loccell{62.4}{22.9} & \loccell{46.4}{23.2} \\
      & & Deletion    & \loccell{90.8}{20.4} & \loccell{85.5}{23.2} & \loccell{81.2}{23.4} & \loccell{71.3}{24.0} & \loccell{61.4}{24.7} & \loccell{46.5}{25.8} \\
      & & Replacement & \loccell{90.4}{20.4} & \loccell{83.9}{22.7} & \loccell{78.2}{22.9} & \loccell{67.0}{23.6} & \loccell{54.7}{24.0} & \loccell{34.8}{24.5} \\
    \addlinespace[1pt]
      & \multirow{3}{*}{OpenDev}
        & Insertion   & \loccell{98.0}{19.2} & \loccell{98.9}{22.3} & \loccell{99.2}{22.6} & \loccell{99.6}{22.5} & \loccell{99.7}{22.5} & \loccell{99.9}{22.4} \\
      & & Deletion    & \loccell{99.3}{19.0} & \loccell{99.2}{23.1} & \loccell{99.4}{24.0} & \loccell{99.4}{25.8} & \loccell{99.3}{28.7} & \loccell{98.8}{36.1} \\
      & & Replacement & \loccell{98.9}{19.3} & \loccell{99.1}{23.4} & \loccell{99.5}{23.9} & \loccell{99.5}{26.0} & \loccell{99.3}{28.9} & \loccell{98.6}{36.4} \\
    \midrule
    \multirow{9}{*}{GPT-5 mini}
      & \multirow{3}{*}{SWE-agent}
        & Insertion   & \loccell{97.7}{12.3} & \loccell{98.2}{15.2} & \loccell{98.0}{16.0} & \loccell{97.9}{17.1} & \loccell{98.1}{18.2} & \loccell{98.1}{20.1} \\
      & & Deletion    & \loccell{96.3}{17.1} & \loccell{95.9}{21.5} & \loccell{96.5}{23.0} & \loccell{96.0}{25.3} & \loccell{95.1}{27.7} & \loccell{93.7}{31.7} \\
      & & Replacement & \loccell{93.8}{16.0} & \loccell{95.9}{20.0} & \loccell{96.3}{21.1} & \loccell{96.3}{23.1} & \loccell{96.9}{25.9} & \loccell{96.8}{31.1} \\
    \addlinespace[1pt]
      & \multirow{3}{*}{OpenHands}
        & Insertion   & \loccell{91.3}{22.9} & \loccell{89.2}{24.8} & \loccell{87.4}{25.1} & \loccell{82.4}{25.7} & \loccell{75.5}{25.8} & \loccell{59.9}{26.7} \\
      & & Deletion    & \loccell{97.2}{22.6} & \loccell{93.8}{24.9} & \loccell{91.9}{25.4} & \loccell{87.5}{26.3} & \loccell{82.1}{27.0} & \loccell{74.4}{28.5} \\
      & & Replacement & \loccell{96.2}{22.1} & \loccell{89.3}{24.1} & \loccell{87.9}{24.8} & \loccell{81.4}{25.6} & \loccell{70.4}{26.2} & \loccell{49.3}{27.0} \\
    \addlinespace[1pt]
      & \multirow{3}{*}{OpenDev}
        & Insertion   & \loccell{97.1}{11.2} & \loccell{98.2}{11.8} & \loccell{98.3}{11.8} & \loccell{99.0}{11.8} & \loccell{99.3}{11.8} & \loccell{99.6}{11.8} \\
      & & Deletion    & \loccell{98.2}{10.9} & \loccell{98.7}{11.4} & \loccell{99.1}{11.9} & \loccell{99.4}{12.7} & \loccell{99.6}{13.6} & \loccell{99.8}{15.9} \\
      & & Replacement & \loccell{99.1}{10.5} & \loccell{98.7}{11.2} & \loccell{98.1}{11.8} & \loccell{98.9}{12.6} & \loccell{99.6}{13.4} & \loccell{99.8}{15.7} \\
    \midrule
    \multirow{9}{*}{MiniMax M3}
      & \multirow{3}{*}{SWE-agent}
        & Insertion   & \loccell{98.4}{21.3} & \loccell{98.3}{22.7} & \loccell{98.2}{22.8} & \loccell{98.2}{23.9} & \loccell{98.3}{24.8} & \loccell{98.4}{26.4} \\
      & & Deletion    & \loccell{96.9}{26.2} & \loccell{97.3}{28.5} & \loccell{96.3}{29.6} & \loccell{95.1}{32.9} & \loccell{93.2}{35.7} & \loccell{89.3}{41.2} \\
      & & Replacement & \loccell{90.5}{25.7} & \loccell{97.0}{27.7} & \loccell{97.4}{28.3} & \loccell{97.1}{30.8} & \loccell{96.8}{33.5} & \loccell{97.1}{39.7} \\
    \addlinespace[1pt]
      & \multirow{3}{*}{OpenHands}
        & Insertion   & \loccell{86.9}{24.7} & \loccell{84.9}{26.6} & \loccell{80.7}{26.6} & \loccell{73.3}{26.9} & \loccell{63.8}{26.7} & \loccell{49.1}{27.2} \\
      & & Deletion    & \loccell{92.2}{25.7} & \loccell{85.2}{27.6} & \loccell{79.7}{27.8} & \loccell{70.5}{28.2} & \loccell{60.5}{28.4} & \loccell{45.0}{28.8} \\
      & & Replacement & \loccell{91.5}{25.5} & \loccell{85.1}{27.3} & \loccell{79.0}{27.5} & \loccell{67.4}{27.7} & \loccell{55.3}{27.9} & \loccell{34.2}{28.4} \\
    \addlinespace[1pt]
      & \multirow{3}{*}{OpenDev}
        & Insertion   & \loccell{98.8}{18.8} & \loccell{99.5}{21.5} & \loccell{99.6}{21.5} & \loccell{99.9}{21.4} & \loccell{99.9}{21.4} & \loccell{99.9}{21.2} \\
      & & Deletion    & \loccell{99.3}{19.3} & \loccell{99.6}{22.5} & \loccell{99.6}{23.0} & \loccell{99.2}{25.8} & \loccell{98.2}{29.8} & \loccell{95.7}{39.6} \\
      & & Replacement & \loccell{99.3}{18.6} & \loccell{99.5}{22.4} & \loccell{99.6}{23.0} & \loccell{99.2}{25.4} & \loccell{98.3}{29.0} & \loccell{95.7}{39.3} \\
    \midrule
    \multicolumn{3}{@{}l}{\textbf{Overall (all displayed cells)}}
      & \loccell{95.8}{19.8} & \loccell{96.0}{23.3} & \loccell{95.4}{23.9}
      & \loccell{93.9}{25.5} & \loccell{92.0}{27.4} & \loccell{88.6}{31.6} \\
    \bottomrule
  \end{tabular}
  \end{adjustbox}
  \vspace{2pt}
  \parbox{\linewidth}{\small Each condition--attack cell micro-pools the
  three primary seeds. ``1 action'' selects one payload action or insertion
  boundary; the remaining columns select the displayed proportion of visible
  actions. Targeted seal removal is excluded because it modifies Q12 carriers rather
  than payload sites.}
  \endgroup
\end{table}

With one selected modification, 95.8\% of eligible payload sites fall inside an
accepted newly reported region; among those sites, the narrowest region has pooled
median width 15 and mean width 19.8 visible actions. At 50\% corruption, the pooled
coverage is 88.6\%, and the corresponding median and mean widths are 25 and 31.6
actions. The disaggregated rows show that SWE-agent and OpenDev retain at least
89.3\% and 95.7\% coverage, respectively, at 50\%, whereas OpenHands is the limiting
framework, especially for replacement (34.2--49.3\%). Its conditional widths must
therefore be read together with the adjacent coverage values. Across all conditions,
terminal regions constitute 22.0\% of the narrowest accepted reports in the
one-action setting and 41.4\% at 50\% corruption, showing the shift toward coarser
enclosing evidence as corruption grows.

These widths define the resolution of the localization claim. \tool{} identifies the
narrowest accepted protocol region enclosing a modified site, not the individual action
itself. As a scale comparison, the nine full-\tool{} cells in
Table~\ref{tab:execution-cost-current} average 40.9 visible actions per trajectory;
the one-action median and mean widths therefore correspond to 36.7\% and 48.4\% of
that average. Granularity is determined jointly by \(\mathsf{Boundary}_\theta\), the
spacing of critical actions, and the surviving alarm level: ordinary reports cover one
segment, group reports two, and terminal reports a suffix. A finer boundary policy
partitions the protected stream into more ordinary segments, while
Eq.~\eqref{eq:seal-overhead} gives the corresponding seal count
\(m+\lfloor m/2\rfloor+1\).

\FloatBarrier
\subsection{Structural Cost and Task Utility}
\label{sec:utility-results}

\subsubsection{RQ4: Mechanism-Level Structural Cost}
\label{sec:results-step-cost}

Owner-only adds no action because it rewrites existing READ choices. For SWE-agent,
full \tool{} adds ordinary, group, and terminal seals: 6.98 visible actions per
trajectory on average. The mean within-condition overhead relative to the
reconstructed natural action count is 22.1\%. This is the protocol's
structural cost; the realized turns and tokens in Table~\ref{tab:execution-cost-current}
also include ordinary run-to-run variation.

\FloatBarrier
\subsubsection{Pass@1}
\label{sec:results-pass-at-one}

\begin{table}[!htb]
  \caption{Pass@1 across three coding agents and three LLMs. Each row pools
  three seed-matched batches of 50 tasks, yielding 150 task runs per
  agent--LLM condition. Superscripts give within-row changes from No-WM in
  percentage points (cyan: lower; red: higher; black: unchanged). Overall
  pools all 1,350 agent--LLM--seed task outcomes per method; the final row
  reports the pooled difference from No-WM. Pass@1 is the percentage of tasks
  whose single submitted patch passes
  the corresponding repository-level test harness; empty or invalid patches
  count as failures.}
  \label{tab:solve-rate-current}
  \centering
  \small
  \newcommand{\passup}[1]{\textcolor[HTML]{B04A3A}{\raisebox{0.35ex}{\scriptsize$\uparrow$#1}}}
  \newcommand{\passdown}[1]{\textcolor[HTML]{2A8C91}{\raisebox{0.35ex}{\scriptsize$\downarrow$#1}}}
  \begin{tabular}{@{}llrrrrr@{}}
    \toprule
    Agent & LLM & No-WM & ActHook-style & AgentMark-U & Owner-only & \tool{} \\
    \midrule
    \multirow{3}{*}{SWE-agent}
      & DeepSeek V4 Flash & 36.7\% & 34.0\%\passdown{2.7} & 36.0\%\passdown{0.7} & 34.7\%\passdown{2.0} & 35.3\%\passdown{1.3} \\
      & GPT-5 mini        & 28.0\% & 30.7\%\passup{2.7} & 4.7\%\passdown{23.3} & 30.0\%\passup{2.0} & 25.3\%\passdown{2.7} \\
      & MiniMax M3        & 27.3\% & 31.3\%\passup{4.0} & 9.3\%\passdown{18.0} & 31.3\%\passup{4.0} & 30.0\%\passup{2.7} \\
    \midrule
    \multirow{3}{*}{OpenHands}
      & DeepSeek V4 Flash & 20.0\% & 18.7\%\passdown{1.3} & 13.3\%\passdown{6.7} & 23.3\%\passup{3.3} & 18.7\%\passdown{1.3} \\
      & GPT-5 mini        & 21.3\% & 22.0\%\passup{0.7} & 6.0\%\passdown{15.3} & 21.3\%\raisebox{0.35ex}{\scriptsize$\phantom{\uparrow}$0.0} & 22.0\%\passup{0.7} \\
      & MiniMax M3        & 20.7\% & 19.3\%\passdown{1.3} & 12.0\%\passdown{8.7} & 22.0\%\passup{1.3} & 24.0\%\passup{3.3} \\
    \midrule
    \multirow{3}{*}{OpenDev}
      & DeepSeek V4 Flash & 36.0\% & 39.3\%\passup{3.3} & 40.7\%\passup{4.7} & 34.0\%\passdown{2.0} & 35.3\%\passdown{0.7} \\
      & GPT-5 mini        & 10.7\% & 8.7\%\passdown{2.0} & 13.3\%\passup{2.7} & 9.3\%\passdown{1.3} & 14.0\%\passup{3.3} \\
      & MiniMax M3        & 36.0\% & 33.3\%\passdown{2.7} & 30.7\%\passdown{5.3} & 34.7\%\passdown{1.3} & 37.3\%\passup{1.3} \\
    \midrule
    \multicolumn{2}{@{}l}{Overall}
      & 26.3\% & 26.4\% & 18.4\% & 26.7\% & 26.9\% \\
    \multicolumn{2}{@{}l}{\(\Delta\) vs. No-WM (pp)}
      & -- & +0.1 & -7.9 & +0.4 & +0.6 \\
    \bottomrule
  \end{tabular}
\end{table}

Each condition is generated as a fresh live agent run rather than by applying a
watermark to a fixed No-WM log. Although the agent, LLM, seed, task identifier, and
test harness are matched, hosted black-box inference can vary across independent runs
because of stochastic generation, retries, context handling, and provider-side
scheduling. Pooled over 1,350 outcomes, complete \tool{} reaches 26.9\% Pass@1 versus
26.3\% for No-WM, a difference of +0.6 percentage points. Across the nine matched
agent--LLM conditions, five changes are positive and four are negative;
an exact two-sided sign test over the nine non-tied conditions gives \(p=1.00\).

\section{Discussion}
\label{sec:discussion}

\subsection{Visibility and Adaptive Removal}
\label{sec:discussion-visibility}
\label{sec:discussion-security}

\tool deliberately uses two complementary visibility profiles. The owner layer rewrites
sparse filename-selected READs that the agent already intended to execute, adding no
trajectory step. The integrity layer inserts explicit READs around protected segments.
Consequently, Q6 provides action-count-neutral ownership evidence, while Q12 provides
visible, tamper-evident integrity checks.

The random-corruption evaluation measures non-targeted damage, whereas the white-box
carrier-deletion experiment intentionally grants a stronger adversary than the
post-publication editor in our main threat model. The public canonicalizer exposes the
candidate pool---the first eligible READ for each normalized basename---but not which
candidates were selected or their coefficient rows. The fourth row of
Figure~\ref{fig:verified-owner-id-recovery-method-comparison} gives the attacker this
% hidden state and exhaustively searches all 63 nonzero null directions for a minimum-cost
% rank-breaking cut. Under this oracle, Owner-only and complete \tool{} decline from
% 85.2\% exact recovery at a 20\% carrier-occurrence budget to 37.0\% at 30\% and 0\%
% at 50\%. These values quantify the rank-suppression boundary under known-carrier,
% known-equation deletion; they do not show that a keyless editor can identify the same
% cut. Raising \(p_o\) increases the equation budget, while an outer erasure code could
% increase the deletion margin.
hidden state and decoder access to select among up to 63 deterministic deletion
sets. Each candidate is applied to the visible trajectories before evidence is
re-extracted and decoded. Removing a natural READ can promote a later READ of the
same basename to carrier status; position-indexed evidence can also change when
deletion changes the positions used by the native adapter. The resulting curves capture both
lost evidence and changes to the reconstructed equations. Raising \(p_o\) increases
the initial equation budget, while an outer erasure code could increase the
deletion margin.

The ordinary, group, and terminal/root commitments make selective seal removal visible
at another level of the Q12 hierarchy. Complete READ removal produces missing
structure or owner abstention instead of accepted
provenance. A forged seal must match its ordinal, payload, preceding observation, and
enclosing commitments under the secret key.

\paragraph{Why OpenHands is the deletion boundary.}
The 84.8\% floor in Figure~\ref{fig:proportional-tamper-detection-current} is not an
inherent weakness of typed tool events; it follows from the evidence topology produced
by the evaluated OpenHands profile. OpenHands first joins each typed action to its
public observation through the cause relation and admits the action to Q12 state only
after a successful outcome. The deletion attack removes that action and all of its
cause-linked observations as one public atom. If a critical action and its adjacent
seal are both removed, the event that would have created the seal expectation and the
local check can disappear together, leaving fewer independent commitments to expose
the deletion.

The clean streams underlying Figure~\ref{fig:proportional-tamper-detection-current}
make the closure gap explicit. Among the OpenHands trajectories with at least one Q12
protector---the population that can enter the deletion denominator---a terminal
protector is present in 47/98 (48.0\%), 120/138 (87.0\%), and 44/98 (44.9\%) of the
DeepSeek V4 Flash, GPT-5 mini, and MiniMax M3 trajectories, respectively. This ordering
tracks their 88.1\%, 95.0\%, and 84.8\% detection rates at 50\% deletion; every
carrier-bearing SWE-agent trajectory has a terminal protector (150/150 in each cell).
Thus the important difference is not merely the number of local seals. SWE-agent's
textual action and outcome form one linear canonical event and, more importantly, its
always-closed root commitment remains an independent witness when a local trigger and
seal disappear. This also explains why targeted seal removal remains near 100\%:
retaining the protected action leaves an unmet commitment, whereas random deletion can
remove both its trigger and local check and, without terminal closure, leave no
higher-level evidence that the segment existed. The direct mitigation is to make every
critical typed action enter protected state before outcome resolution, map a missing or
failed observation to an explicit \(\mathsf{unknown}\) outcome, and close Q12 on every
runtime termination path, including normal finish, timeout, and error.

\subsection{Composing Trajectory and Artifact Watermarks}
\label{sec:discussion-composition}

Trajectory and artifact watermarks protect different objects. A code watermark can
associate a final patch with a provider-specific pattern, whereas \tool associates a
visible action history with an owner and checks protected regions of that history. A
deployment can combine an artifact watermark, natural-action owner equations, and overt
integrity seals.

Agreement across the trajectory and artifact layers strengthens provenance. A marked
patch without a verified trajectory owner, or an owner-verified trajectory with an
unmarked patch, flags inconsistent provenance evidence for investigation. A Q12
alarm on an owner-verified trajectory indicates an inconsistency in the protected
canonical projection. This composition also provides redundancy when basename collisions reduce
trajectory capacity or refactoring weakens an artifact mark.

\subsection{Operational Guidance}
\label{sec:discussion-guidance}

The results suggest five deployment rules:

\begin{enumerate}[leftmargin=2.2em]
  \item \textbf{Gate attribution.} Require rank six, a unique best owner, and a score
        meeting the evidence threshold; otherwise abstain. When all three gates pass,
        accept the claim if the recovered ID matches the predeclared owner ID, and
        reject it otherwise.
  \item \textbf{Calibrate each adapter.} Re-estimate clean replay and eligible-basename
        distributions after changes to models, serializers, or tool vocabularies.
  \item \textbf{Choose integrity explicitly.} Use owner-only for zero-added-action
        provenance and full mode when redundant integrity alarms justify extra READs.
  \item \textbf{Preserve canonical order.} Publish the protocol version and a
        deterministic visible index, especially for concurrent tool frameworks.
  \item \textbf{Use the strongest surviving channel.} If an authenticated sidecar is
        preserved end to end, a signed or hash-chained log can provide stronger
        integrity guarantees; additional properties depend on the authentication and
        freshness mechanisms. \tool{} targets releases that retain the human-readable action stream but
        discard wrapper metadata, such as exported trace datasets, benchmark or
        leaderboard trace artifacts, excerpts reproduced in issue threads or papers,
        and cross-organization hand-offs re-rendered by another viewer. If the action
        stream itself is not retained, \tool{} is not applicable.
\end{enumerate}

\label{sec:discussion-scope}
Keys should be domain-separated by deployment, owner, and layer, and rotated for future
batches without rewriting old trajectories. The resulting signals authenticate the
published visible process through owner attribution and tamper detection.

\section{Threats to Validity}
\label{sec:validity}

\subsection{Canonicalization and Implementation Validity}
\label{sec:validity-canonicalization}
\label{sec:validity-internal}

The canonicalizer is part of the trusted public protocol. It uses published action
requests, visible targets, exposed status, and canonical order, and excludes hidden
reasoning, debug anchors, and wrapper metadata. A small public action alphabet reduces
parser ambiguity. Semantic synonyms and complex shell pipelines can change carrier
availability, while basename collisions can repeat equations without increasing rank;
the rank and uniqueness gates prevent such batches from producing an unsupported ID.

All qualifying live-API runs are retained. Before aggregation, we verify positive API
use, valid artifacts, exact instance and model identities, clean round-trip replay,
visible order, and synchronized local recomputation. The reported tables are generated
directly from the resulting frozen artifacts.

\subsection{Statistical and Construct Validity}
\label{sec:validity-statistics}
\label{sec:validity-utility}

Owner equations are correlated because models revisit files and reuse exploration
patterns. We therefore use exact recovery, rank, uniqueness, score margin, and batch
resampling as the primary evidence, while Eq.~\eqref{eq:owner-z} summarizes carrier
agreement. The 500 owner-corruption trials and 50 integrity trials per trajectory
sample the attack distribution at each of the five reported corruption levels.

The evaluated configuration fixes owner label \texttt{0x3f}. Algebraically, this label
receives no privileged decoder treatment: Eq.~\eqref{eq:owner-mask} masks
\(b_f=\langle\mathbf a_f,\mathbf w\rangle\) with a domain-separated keyed bit before
the public Q6 choice is oriented, and the decoder applies the same rank, uniqueness,
and score gates to all 64 candidates. Under Assumption~\ref{ass:ideal-prf} and for a
fixed key-independent opportunity pool, the visible carrier choice is marginally
uniform for every fixed label. This is an idealized distributional argument rather
than an end-to-end empirical label sweep: during live generation, an early Q6
realization can alter later actions and carrier availability. Our empirical recovery
results therefore apply to \texttt{0x3f}; evaluating multiple labels under live
generation is needed to test label invariance of the adaptive online system.

The evaluated six-bit owner space permits exhaustive decoding of 64 candidates. The
construction scales to larger identifiers with more independent equations or an outer
code. Pass@1 is also compared within matched agent--LLM rows because repository
families and test harnesses
differ. Structural additions are exactly observable, whereas turns and tokens also
reflect stochastic generation, retries, context handling, and provider scheduling.

\subsection{External Validity}
\label{sec:validity-external}

Three agents, three model providers, and Python, Java, and JavaScript repositories cover
distinct interfaces and ecosystems. Deployment to IDE-native, GUI, multi-agent, or
long-horizon systems requires the same adapter calibration used in our crossed grid.
The primary evaluation uses three matched 50-task seeds, so the reported capacity
reflects their filename diversity and eligible READ supply.

Our native tool-dispatch adapters observe selected actions rather than elicited
planning-time probability lists. We therefore evaluate \agentmarku, the uniform
adaptation defined in Section~\ref{sec:setup-conditions}; all reported AgentMark
comparisons refer to this executable interface-compatible baseline.

\subsection{Security Validity}
\label{sec:validity-security}

The security model requires HMAC pseudorandomness, a secret \(K\), and an honest
embedding runtime. Public canonicalization reveals the eligible action classes, while
the key determines selected filenames, expected subtypes, and seal payloads.

The attack suite covers random non-key-aware corruption and exact-budget white-box
carrier deletion. The public protocol reveals candidate carrier classes and
first-occurrence basenames; selected carriers and coefficient rows remain keyed. Our
% white-box evaluation deliberately reveals both and grants an optimal rank-breaking
% search, thereby bounding recovery under an oracle carrier-aware attacker. The main
white-box evaluation reveals both and grants decoder access to select among up to
63 deterministic physical-deletion candidates. Recovery is measured after the
selected actions are removed and evidence is re-extracted. The main
post-publication threat model and random-corruption results do not establish how well
a keyless editor could infer selected carriers from repeated releases or verifier
feedback. Complete READ erasure removes the
authenticated substrate; the verifier then abstains or reports missing Q12 structure
instead of accepting an owner or an intact trajectory.

\section{Related Work}
\label{sec:related}

\subsection{Generated-Content and Agent Watermarking}
\label{sec:related-artifacts}
\label{sec:related-agents}

Language-model watermarks bias or partition generation choices to accumulate keyed
evidence. Prior work studies green lists, provable robustness, unbiased or
distortion-free sampling, semantic carriers, and multi-bit messages
\cite{kirchenbauer2023watermark,zhao2024provable,kuditipudi2024distortionfree,
ICLR2024_c5b00c5b,hou2024semstamp,ren2024semark,yoo2024multibit,feng2025bimark}.
WaterBench, MarkLLM, and attack studies further separate utility, detectability, and
robustness \cite{tu2024waterbench,pan2024markllm,rastogi-pruthi-2024-revisiting,
chang2025smoothing}. \tool adopts that separation but evaluates complete interactive
trajectories and executable repository tasks rather than generated text alone.

Semantic and distribution-aware schemes move beyond fixed token identities through
embedding partitions, token priors, lexical redundancy, or keyed transformations
\cite{ren2024riw,chen2024watme,lau2024waterfall}. Multi-bit methods additionally encode
traceable messages through invariant linguistic features or position allocation
\cite{yoo2023multibit}. These designs motivate our finite owner identifier, but their
evidence is accumulated over generated language. Coding-agent READ actions are sparse,
task-constrained, and correlated, which motivates exact finite-identifier decoding and
empirical batch calibration rather than an independent-token approximation.

Code watermarking must preserve syntax and behavior. Classical software watermarking
modifies program structure, while SWEET, CodeIP, and STONE constrain lexical or grammar
choices during generation \cite{collberg1999software,lee2024sweet,guan2024codeip,
kim2026stone}. These methods protect the artifact; \tool protects the visible process
that produced it. The two are complementary because different trajectories can yield
the same patch and a trajectory can be edited without changing that patch
(Section~\ref{sec:discussion-composition}).

Behavioral methods move the carrier into agent actions. AgentMark uses
utility-preserving choices conditioned on candidate-action probabilities; AGENTWM,
ActHook, and SeqWM construct or trigger marked execution behavior
\cite{huang2026agentmark,wang2026agentwm,meng2026acthook,an2026seqwm}. \tool
indexes natural carriers by normalized filenames and decodes
random linear constraints on a finite owner identifier. Its Q12 integrity layer uses
ordinary, group, and terminal commitments to turn visible-action edits into protocol
alarms. Exact owner recovery and tamper detection are separate verification decisions.

This distinction is also methodological. AgentMark elicits candidate-behavior
probabilities from the LLM for sampling and decoding. Our evaluated native interfaces
provide selected tool calls rather than these planning-time probability lists;
\agentmarku{} applies uniform-choice watermarking at this tool-dispatch boundary.
ActHook and AGENTWM demonstrate that behavior can carry ownership. \tool additionally
uses nested ordinary, group, and terminal commitments to preserve tamper evidence across
the protected trajectory structure.

\subsection{Secure Logs and Fragile Integrity Watermarks}
\label{sec:related-logs}
\label{sec:related-fragile}

Secure audit logs and provenance systems use hash chains, signatures, authenticated
structures, protected state, or trusted derivation graphs to expose modification
\cite{haber1991timestamp,bellare1997forward,schneier1999secure,crosby2009tamper,
ma2009securelogging,muniswamy2006pass,hasan2009provenance}. These mechanisms are
preferable when authenticated metadata is preserved end to end. \tool addresses the
complementary release boundary in which the visible action stream survives but wrapper
metadata or a logging service does not---for example, exported trace datasets,
benchmark or leaderboard trace artifacts, republished trace excerpts, and re-rendered
cross-organization hand-offs. In such channels the action record is the only surviving
carrier; when a sidecar can be preserved, deployments should prefer or combine signed
logging.

Data-provenance systems similarly bind outputs to derivation histories, usually through
a trusted graph or storage layer \cite{braun2008securing,waters2004encrypted,
hoang2022faster}. Our verifier may instead receive only a serialized action trace and a
key. The terminal/root commitment borrows the completeness objective of secure logs,
but represents it through a visible READ subtype rather than an opaque signature; group
commitments amortize protection over pairs of ordinary segments.

Fragile watermarks are designed to react to modification, unlike robust ownership marks
that aim to survive transformation \cite{Delp,cox2007digital}. Recent media
systems combine ownership with spatial or temporal tamper localization
% \cite{neekhara2024facesigns,zhang2024editguard,zhang2025omniguard,
% sander2025watermarkanything,sanroman2024localizedvoice,yang2025stableguard,
% petrov2025coexistence}. \tool transfers this robust-versus-fragile separation to
\cite{neekhara2024facesigns,zhang2024editguard,zhang2025omniguard,
sander2025watermarkanything,yang2025stableguard}. \tool transfers this robust-versus-fragile separation to
discrete action traces: sparse filename equations support batch ownership, while
ordinary, group, and terminal/root commitments trigger protocol alarms under action
edits and seal removal. The three commitment roles provide overlapping integrity checks
over the protected action structure.

The hierarchy is dependency-structured rather than a collection of independent tags.
Each ordinary seal commits to its segment and the previously observed seal, each group
binds two ordinary seals, and the terminal/root role commits to the protected
projection. These dependencies allow an enclosing commitment to retain tamper evidence
when local evidence is modified or removed. This differs from robust media watermarking,
where success commonly means recovering the same message after editing.

\subsection{Coding Agents and Repository-Level Benchmarks}
\label{sec:related-coding-agents}

Tool-using agents interleave model responses, actions, and environment feedback.
ReAct and Toolformer established core interaction patterns, while AgentBench, WebArena,
Mind2Web, and OSWorld evaluate heterogeneous long-horizon environments
\cite{yao2023react,schick2023toolformer,liu2023agentbench,zhou2024webarena,
deng2023mind2web,xie2024osworld}. Repository-level systems additionally require search,
cross-file reasoning, execution, and validation; RepoBench, RepoCoder, CodeAgent,
AutoCodeRover, and Agentless study these capabilities
\cite{liu2023repobench,zhang2023repocoder,zhang2024codeagentenhancingcodegeneration,
zhang2024autocoderover,xia2024agentless}.

These benchmarks also clarify why final-output watermarking is insufficient for our
setting. Repository repair contains exploration, failed attempts, tests, and revisions
that may be absent from the submitted patch but remain valuable audit evidence. At the
same time, a useful carrier must not prevent the agent from selecting tools needed to
solve the task, which motivates our separation between rewriting existing owner
choices and adding explicit integrity actions.

SWE-bench made issue resolution executable, SWE-agent highlighted the effect of the
agent-computer interface, OpenHands standardized typed tool events, and SWE-PolyBench
extended evaluation across languages \cite{jimenez2024swebench,yang2024sweagent,
wang2024openhands,rashid2025swepolybench}. These systems supply natural read and search
carriers but serialize tools, failures, concurrency, and termination differently. Our
adapter contract therefore treats public visible serialization and clean replay as part
of the watermark protocol rather than assuming one universal trajectory format.

\section{Conclusion}
\label{sec:conclusion}

\subsection{Summary}

Coding-agent trajectories are valuable process artifacts, but ownership and integrity
place opposite demands on a watermark. \tool separates them. Its sparse ownership
layer turns selected natural filename READs into masked random linear equations and
adds no actions. Its overt integrity layer commits to critical-action segments with
ordinary Q12 seals, group-of-two protection, and a terminal commitment.

Across three agents, three model providers, three task-list seeds, and three repository
languages, the visible-only decoder recovers the exact owner in
\PrimaryExactRecoveryConditions{} of \PrimaryConditions{} primary conditions.
At 20\% random corruption, Owner-only recovers the correct ID in 99.3--100\% of trials
and complete \tool{} in 96.3--100\%; their lowest rates at 50\% are 79.1\% and 76.5\%.
Q12 detects 95.5--100\% of exhaustive eligible single-site edits. Under 50\%
proportional corruption, every SWE-agent and OpenDev curve remains at least 99.4\%, and
the lowest OpenHands condition remains 84.8\%. These results establish robust owner
recovery and tamper detection across the tested agents and LLMs.

Owner-only adds no trajectory action. On SWE-agent, full \tool{} adds 6.98 visible
actions per trajectory on average. Its pooled Pass@1 across all three agents is
26.9\% versus 26.3\% for No-WM; across
the nine matched agent--LLM conditions, five changes increase and four decrease
(exact two-sided sign test over the nine
non-tied conditions, \(p=1.00\)). Thus the owner layer provides zero-added-action
attribution, and the full protocol adds explicit integrity evidence with measured
structural cost.

The resulting design is composable. Applications can deploy owner-only mode for
zero-added-action provenance, full mode for segment integrity, and artifact watermarks or
secure logs alongside either mode. Each component provides a distinct verification role
and measured cost.

\subsection{Future Work}

Two evaluation gaps remain. The primary live-generation study fixes the deployment
label at \texttt{0x3f} and, after the calibration sweep in
Table~\ref{tab:parameter-configuration}, deploys \(p_o=0.30\); live runs across labels
and deployed densities would test the mask-based symmetry of
Eq.~\eqref{eq:owner-mask} empirically. The baseline comparison uses frozen
method-specific operating points rather than an equal visible-carrier budget; a
carrier-budget-matched analysis would complement this deployment-oriented comparison
by separating evidence efficiency from carrier volume.

Three directions follow from the current results. First, probability-aware Q6
substitution could improve distribution preservation when an explicit candidate-action
distribution is available. Second, an outer erasure- and error-correcting code
could expand the owner space beyond six bits while retaining explicit abstention.
Third, adaptive evaluations should combine semantic command rewriting, targeted READ
removal, and utility-preserving trajectory compression.

For integrity, future work can learn adapter-specific natural realizations of Q12
subtypes without changing the public payload and can evaluate key-aware,
utility-preserving multi-region attacks and adaptive verifier feedback. A broader
provenance stack should also connect trajectory evidence with code watermarks, signed
execution environments, and benchmark-grounded patch validation.

\bibliographystyle{unsrtnat}
\bibliography{references}

\begin{thebibliography}{56}
\providecommand{\natexlab}[1]{#1}
\providecommand{\url}[1]{\texttt{#1}}
\expandafter\ifx\csname urlstyle\endcsname\relax
  \providecommand{\doi}[1]{doi: #1}\else
  \providecommand{\doi}{doi: \begingroup \urlstyle{rm}\Url}\fi

\bibitem[Yang et~al.(2024)Yang, Jimenez, Wettig, Lieret, Yao, Narasimhan, and
  Press]{yang2024sweagent}
John Yang, Carlos~E. Jimenez, Alexander Wettig, Kilian Lieret, Shunyu Yao,
  Karthik Narasimhan, and Ofir Press.
\newblock {SWE-agent}: Agent-computer interfaces enable automated software
  engineering, 2024.
\newblock URL \url{https://arxiv.org/abs/2405.15793}.

\bibitem[Wang et~al.(2024)Wang, Li, Song, Xu, Tang, Zhuge, Pan, Song, Li,
  Singh, Tran, Li, Ma, Zheng, Qian, Shao, Muennighoff, Zhang, Hui, Lin,
  Brennan, Peng, Ji, and Neubig]{wang2024openhands}
Xingyao Wang, Boxuan Li, Yufan Song, Frank~F. Xu, Xiangru Tang, Mingchen Zhuge,
  Jiayi Pan, Yueqi Song, Bowen Li, Jaskirat Singh, Hoang~H. Tran, Fuqiang Li,
  Ren Ma, Mingzhang Zheng, Bill Qian, Yanjun Shao, Niklas Muennighoff, Yizhe
  Zhang, Binyuan Hui, Junyang Lin, Robert Brennan, Hao Peng, Heng Ji, and
  Graham Neubig.
\newblock {OpenHands}: An open platform for {AI} software developers as
  generalist agents, 2024.
\newblock URL \url{https://arxiv.org/abs/2407.16741}.

\bibitem[Kirchenbauer et~al.(2023)Kirchenbauer, Geiping, Wen, Katz, Miers, and
  Goldstein]{kirchenbauer2023watermark}
John Kirchenbauer, Jonas Geiping, Yuxin Wen, Jonathan Katz, Ian Miers, and Tom
  Goldstein.
\newblock A watermark for large language models.
\newblock In \emph{Proceedings of the 40th International Conference on Machine
  Learning}, volume 202 of \emph{Proceedings of Machine Learning Research},
  pages 17061--17084. PMLR, 2023.
\newblock URL \url{https://proceedings.mlr.press/v202/kirchenbauer23a.html}.

\bibitem[Zhao et~al.(2024)Zhao, Ananth, Li, and Wang]{zhao2024provable}
Xuandong Zhao, Prabhanjan~Vijendra Ananth, Lei Li, and Yu-Xiang Wang.
\newblock Provable robust watermarking for {AI}-generated text.
\newblock In \emph{International Conference on Learning Representations}, 2024.
\newblock URL \url{https://openreview.net/forum?id=Bwz0fy9Hc9}.

\bibitem[Collberg and Thomborson(1999)]{collberg1999software}
Christian~S. Collberg and Clark~D. Thomborson.
\newblock Software watermarking: Models and dynamic embeddings.
\newblock In \emph{Proceedings of the 26th ACM SIGPLAN-SIGACT Symposium on
  Principles of Programming Languages}, pages 311--324. ACM, 1999.
\newblock \doi{10.1145/292540.292569}.

\bibitem[Huang et~al.(2026)Huang, Tan, Wei, Li, Zhang, Tian, Yang, and
  Zhou]{huang2026agentmark}
Kaibo Huang, Jin Tan, Yukun Wei, Wanling Li, Zipei Zhang, Hui Tian, Zhongliang
  Yang, and Linna Zhou.
\newblock {AgentMark}: Utility-preserving behavioral watermarking for agents,
  2026.
\newblock URL \url{https://arxiv.org/abs/2601.03294}.

\bibitem[Wang et~al.(2026)Wang, Li, Xie, Wang, She, Wang, and
  Rahmel]{wang2026agentwm}
Liwen Wang, Zongjie Li, Yuchong Xie, Shuai Wang, Dongdong She, Wei Wang, and
  Juergen Rahmel.
\newblock On protecting agentic systems' intellectual property via
  watermarking, 2026.
\newblock URL \url{https://arxiv.org/abs/2602.08401}.

\bibitem[Meng et~al.(2026)Meng, Gong, Zhuo, Zhang, Li, Liu, Yang, Wei, and
  Chen]{meng2026acthook}
Wenlong Meng, Chen Gong, Terry~Yue Zhuo, Fan Zhang, Kecen Li, Zheng Liu, Zhou
  Yang, Chengkun Wei, and Wenzhi Chen.
\newblock Watermarking {LLM} agent trajectories, 2026.
\newblock URL \url{https://arxiv.org/abs/2602.18700}.

\bibitem[An et~al.(2026)An, Park, Kim, and Han]{an2026seqwm}
Hyeseon An, Shinwoo Park, Dongsu Kim, and Yo-Sub Han.
\newblock Sequential behavioral watermarking for {LLM} agents, 2026.
\newblock URL \url{https://arxiv.org/abs/2605.11036}.

\bibitem[Cox et~al.(2007)Cox, Miller, Bloom, Fridrich, and
  Kalker]{cox2007digital}
Ingemar~J. Cox, Matthew~L. Miller, Jeffrey~A. Bloom, Jessica Fridrich, and Ton
  Kalker.
\newblock \emph{Digital Watermarking and Steganography}.
\newblock Morgan Kaufmann, 2 edition, 2007.

\bibitem[Schneier and Kelsey(1999)]{schneier1999secure}
Bruce Schneier and John Kelsey.
\newblock Secure audit logs to support computer forensics.
\newblock \emph{ACM Transactions on Information and System Security},
  2\penalty0 (2):\penalty0 159--176, 1999.
\newblock \doi{10.1145/317087.317089}.

\bibitem[Jimenez et~al.(2024)Jimenez, Yang, Wettig, Yao, Pei, Press, and
  Narasimhan]{jimenez2024swebench}
Carlos~E. Jimenez, John Yang, Alexander Wettig, Shunyu Yao, Kexin Pei, Ofir
  Press, and Karthik Narasimhan.
\newblock {SWE-bench}: Can language models resolve real-world {GitHub} issues?
\newblock In \emph{International Conference on Learning Representations}, 2024.
\newblock URL \url{https://openreview.net/forum?id=VTF8yNQM66}.

\bibitem[Rashid et~al.(2025)Rashid, Bock, Zhuang, Buchholz, Esler, Valentin,
  Franceschi, Wistuba, Sivaprasad, Kim, Deoras, Zappella, and
  Callot]{rashid2025swepolybench}
Muhammad~Shihab Rashid, Christian Bock, Yuan Zhuang, Alexander Buchholz, Tim
  Esler, Simon Valentin, Luca Franceschi, Martin Wistuba, Prabhu~Teja
  Sivaprasad, Woo~Jung Kim, Anoop Deoras, Giovanni Zappella, and Laurent
  Callot.
\newblock {SWE-PolyBench}: A multi-language benchmark for repository-level
  evaluation of coding agents, 2025.
\newblock URL \url{https://arxiv.org/abs/2504.08703}.

\bibitem[Kuditipudi et~al.(2024)Kuditipudi, Thickstun, Hashimoto, and
  Liang]{kuditipudi2024distortionfree}
Rohith Kuditipudi, John Thickstun, Tatsunori Hashimoto, and Percy Liang.
\newblock Robust distortion-free watermarks for language models.
\newblock \emph{Transactions on Machine Learning Research}, 2024.
\newblock URL \url{https://openreview.net/forum?id=FpaCL1MO2C}.

\bibitem[Hu et~al.(2024)Hu, Chen, Wu, Wu, Zhang, and Huang]{ICLR2024_c5b00c5b}
Zhengmian Hu, Lichang Chen, Xidong Wu, Yihan Wu, Hongyang Zhang, and Heng
  Huang.
\newblock Unbiased watermark for large language models.
\newblock In B.~Kim, Y.~Yue, S.~Chaudhuri, K.~Fragkiadaki, M.~Khan, and Y.~Sun,
  editors, \emph{International Conference on Learning Representations}, volume
  2024, pages 45408--45436, 2024.
\newblock URL
  \url{https://proceedings.iclr.cc/paper_files/paper/2024/file/c5b00c5bdcc6fe35907dbcca03d27652-Paper-Conference.pdf}.

\bibitem[Hou et~al.(2024)Hou, Zhang, He, Wang, Chuang, Wang, Shen, Durme,
  Khashabi, and Tsvetkov]{hou2024semstamp}
Abe Hou, Jingyu Zhang, Tianxing He, Yichen Wang, Yung-Sung Chuang, Hongwei
  Wang, Lingfeng Shen, Benjamin~Van Durme, Daniel Khashabi, and Yulia Tsvetkov.
\newblock {SemStamp}: A semantic watermark with paraphrastic robustness for
  text generation.
\newblock In \emph{Proceedings of the 2024 Conference of the North American
  Chapter of the Association for Computational Linguistics}, pages 4067--4082,
  Mexico City, Mexico, 2024. Association for Computational Linguistics.
\newblock \doi{10.18653/v1/2024.naacl-long.226}.

\bibitem[Ren et~al.(2024{\natexlab{a}})Ren, Xu, Liu, Cui, Wang, Yin, and
  Tang]{ren2024semark}
Jie Ren, Han Xu, Yiding Liu, Yingqian Cui, Shuaiqiang Wang, Dawei Yin, and
  Jiliang Tang.
\newblock A robust semantics-based watermark for large language model against
  paraphrasing.
\newblock In \emph{Findings of the Association for Computational Linguistics:
  NAACL 2024}, pages 613--625, Mexico City, Mexico, 2024{\natexlab{a}}.
  Association for Computational Linguistics.
\newblock \doi{10.18653/v1/2024.findings-naacl.40}.

\bibitem[Yoo et~al.(2024)Yoo, Ahn, and Kwak]{yoo2024multibit}
KiYoon Yoo, Wonhyuk Ahn, and Nojun Kwak.
\newblock Advancing beyond identification: Multi-bit watermark for large
  language models.
\newblock In \emph{Proceedings of the 2024 Conference of the North American
  Chapter of the Association for Computational Linguistics}, pages 4031--4055,
  Mexico City, Mexico, 2024. Association for Computational Linguistics.
\newblock \doi{10.18653/v1/2024.naacl-long.224}.

\bibitem[Feng et~al.(2025)Feng, Zhang, Zhang, Zhang, and Pan]{feng2025bimark}
Xiaoyan Feng, He~Zhang, Yanjun Zhang, Leo~Yu Zhang, and Shirui Pan.
\newblock {BiMark}: Unbiased multilayer watermarking for large language models.
\newblock In \emph{Proceedings of the 42nd International Conference on Machine
  Learning}, volume 267 of \emph{Proceedings of Machine Learning Research},
  pages 17049--17067. PMLR, 2025.

\bibitem[Tu et~al.(2024)Tu, Sun, Bai, Yu, Hou, and Li]{tu2024waterbench}
Shangqing Tu, Yuliang Sun, Yushi Bai, Jifan Yu, Lei Hou, and Juanzi Li.
\newblock {WaterBench}: Towards holistic evaluation of watermarks for large
  language models.
\newblock In \emph{Proceedings of the 62nd Annual Meeting of the Association
  for Computational Linguistics}, pages 1517--1542, Bangkok, Thailand, 2024.
  Association for Computational Linguistics.
\newblock \doi{10.18653/v1/2024.acl-long.83}.

\bibitem[Pan et~al.(2024)Pan, Liu, He, Gao, Zhao, Lu, Zhou, Liu, Hu, Wen, King,
  and Yu]{pan2024markllm}
Leyi Pan, Aiwei Liu, Zhiwei He, Zitian Gao, Xuandong Zhao, Yijian Lu, Binglin
  Zhou, Shuliang Liu, Xuming Hu, Lijie Wen, Irwin King, and Philip~S. Yu.
\newblock {MarkLLM}: An open-source toolkit for {LLM} watermarking.
\newblock In \emph{Proceedings of the 2024 Conference on Empirical Methods in
  Natural Language Processing: System Demonstrations}, pages 61--71, Miami,
  Florida, USA, 2024. Association for Computational Linguistics.
\newblock \doi{10.18653/v1/2024.emnlp-demo.7}.

\bibitem[Rastogi and Pruthi(2024)]{rastogi-pruthi-2024-revisiting}
Saksham Rastogi and Danish Pruthi.
\newblock Revisiting the robustness of watermarking to paraphrasing attacks.
\newblock In Yaser Al-Onaizan, Mohit Bansal, and Yun-Nung Chen, editors,
  \emph{Proceedings of the 2024 Conference on Empirical Methods in Natural
  Language Processing}, pages 18100--18110, Miami, Florida, USA, November 2024.
  Association for Computational Linguistics.
\newblock \doi{10.18653/v1/2024.emnlp-main.1005}.
\newblock URL \url{https://aclanthology.org/2024.emnlp-main.1005/}.

\bibitem[Chang et~al.(2025)Chang, Hassani, and Shokri]{chang2025smoothing}
Hongyan Chang, Hamed Hassani, and Reza Shokri.
\newblock Watermark smoothing attacks against language models.
\newblock In \emph{Findings of the Association for Computational Linguistics:
  EMNLP 2025}, pages 4915--4941, Suzhou, China, 2025. Association for
  Computational Linguistics.
\newblock \doi{10.18653/v1/2025.findings-emnlp.264}.

\bibitem[Ren et~al.(2024{\natexlab{b}})Ren, Guo, Cao, and Ma]{ren2024riw}
Yubing Ren, Ping Guo, Yanan Cao, and Wei Ma.
\newblock Subtle signatures, strong shields: Advancing robust and imperceptible
  watermarking in large language models.
\newblock In \emph{Findings of the Association for Computational Linguistics:
  ACL 2024}, pages 5508--5519, Bangkok, Thailand, 2024{\natexlab{b}}.
  Association for Computational Linguistics.
\newblock \doi{10.18653/v1/2024.findings-acl.327}.

\bibitem[Chen et~al.(2024)Chen, Bian, Deng, Cai, Li, Zhao, and
  Wong]{chen2024watme}
Liang Chen, Yatao Bian, Yang Deng, Deng Cai, Shuaiyi Li, Peilin Zhao, and
  Kam-Fai Wong.
\newblock {WatME}: Towards lossless watermarking through lexical redundancy.
\newblock In \emph{Proceedings of the 62nd Annual Meeting of the Association
  for Computational Linguistics}, pages 9166--9180, Bangkok, Thailand, 2024.
  Association for Computational Linguistics.
\newblock \doi{10.18653/v1/2024.acl-long.496}.

\bibitem[Lau et~al.(2024)Lau, Niu, Dao, Chen, Foo, and Low]{lau2024waterfall}
Gregory Kang~Ruey Lau, Xinyuan Niu, Hieu Dao, Jiangwei Chen, Chuan-Sheng Foo,
  and Bryan Kian~Hsiang Low.
\newblock Waterfall: Scalable framework for robust text watermarking and
  provenance for {LLMs}.
\newblock In \emph{Proceedings of the 2024 Conference on Empirical Methods in
  Natural Language Processing}, pages 20432--20466, Miami, Florida, USA, 2024.
  Association for Computational Linguistics.
\newblock \doi{10.18653/v1/2024.emnlp-main.1138}.

\bibitem[Yoo et~al.(2023)Yoo, Ahn, Jang, and Kwak]{yoo2023multibit}
KiYoon Yoo, Wonhyuk Ahn, Jiho Jang, and Nojun Kwak.
\newblock Robust multi-bit natural language watermarking through invariant
  features.
\newblock In \emph{Proceedings of the 61st Annual Meeting of the Association
  for Computational Linguistics}, pages 2092--2115, Toronto, Canada, 2023.
  Association for Computational Linguistics.
\newblock \doi{10.18653/v1/2023.acl-long.117}.

\bibitem[Lee et~al.(2024)Lee, Hong, Ahn, Hong, Lee, Yun, Shin, and
  Kim]{lee2024sweet}
Taehyun Lee, Seokhee Hong, Jaewoo Ahn, Ilgee Hong, Hwaran Lee, Sangdoo Yun,
  Jamin Shin, and Gunhee Kim.
\newblock Who wrote this code? watermarking for code generation.
\newblock In \emph{Proceedings of the 62nd Annual Meeting of the Association
  for Computational Linguistics}, pages 4890--4911, Bangkok, Thailand, 2024.
  Association for Computational Linguistics.
\newblock \doi{10.18653/v1/2024.acl-long.268}.

\bibitem[Guan et~al.(2024)Guan, Wan, Bi, Wang, Zhang, Zhou, and
  Sun]{guan2024codeip}
Batu Guan, Yao Wan, Zhangqian Bi, Zheng Wang, Hongyu Zhang, Pan Zhou, and
  Lichao Sun.
\newblock {CodeIP}: A grammar-guided multi-bit watermark for large language
  models of code.
\newblock In \emph{Findings of the Association for Computational Linguistics:
  EMNLP 2024}, pages 9243--9258, Miami, Florida, USA, 2024. Association for
  Computational Linguistics.
\newblock \doi{10.18653/v1/2024.findings-emnlp.541}.

\bibitem[Kim et~al.(2026)Kim, Park, and Han]{kim2026stone}
Jungin Kim, Shinwoo Park, and Yo-Sub Han.
\newblock Marking code without breaking it: Code watermarking for detecting
  {LLM}-generated code.
\newblock In \emph{Findings of the Association for Computational Linguistics:
  EACL 2026}, pages 3990--4002, Rabat, Morocco, 2026. Association for
  Computational Linguistics.
\newblock \doi{10.18653/v1/2026.findings-eacl.207}.

\bibitem[Haber and Stornetta(1991)]{haber1991timestamp}
Stuart Haber and W.~Scott Stornetta.
\newblock How to time-stamp a digital document.
\newblock \emph{Journal of Cryptology}, 3\penalty0 (2):\penalty0 99--111, 1991.
\newblock \doi{10.1007/BF00196791}.

\bibitem[Bellare and Yee(1997)]{bellare1997forward}
Mihir Bellare and Bennet~S. Yee.
\newblock Forward integrity for secure audit logs.
\newblock Technical report, University of California, San Diego, 1997.

\bibitem[Crosby and Wallach(2009)]{crosby2009tamper}
Scott~A. Crosby and Dan~S. Wallach.
\newblock Efficient data structures for tamper-evident logging.
\newblock In \emph{18th USENIX Security Symposium}, pages 317--334, Montreal,
  Canada, 2009. USENIX Association.

\bibitem[Ma and Tsudik(2009)]{ma2009securelogging}
Di~Ma and Gene Tsudik.
\newblock A new approach to secure logging.
\newblock \emph{ACM Transactions on Storage}, 5\penalty0 (1):\penalty0 1--21,
  2009.
\newblock \doi{10.1145/1502777.1502779}.

\bibitem[Muniswamy-Reddy et~al.(2006)Muniswamy-Reddy, Holland, Braun, and
  Seltzer]{muniswamy2006pass}
Kiran-Kumar Muniswamy-Reddy, David~A. Holland, Uri Braun, and Margo Seltzer.
\newblock Provenance-aware storage systems.
\newblock In \emph{USENIX Annual Technical Conference}, pages 43--56, Boston,
  Massachusetts, 2006. USENIX Association.

\bibitem[Hasan et~al.(2009)Hasan, Sion, and Winslett]{hasan2009provenance}
Ragib Hasan, Radu Sion, and Marianne Winslett.
\newblock Preventing history forgery with secure provenance.
\newblock \emph{ACM Transactions on Storage}, 5\penalty0 (4):\penalty0 1--43,
  2009.
\newblock \doi{10.1145/1629080.1629082}.

\bibitem[Braun et~al.(2008)Braun, Shinnar, and Seltzer]{braun2008securing}
Uri Braun, Avraham Shinnar, and Margo Seltzer.
\newblock Securing provenance.
\newblock In \emph{USENIX Workshop on Hot Topics in Security}, San Jose,
  California, 2008. USENIX Association.

\bibitem[Waters et~al.(2004)Waters, Balfanz, Durfee, and
  Smetters]{waters2004encrypted}
Brent~R. Waters, Dirk Balfanz, Glenn Durfee, and D.~K. Smetters.
\newblock Building an encrypted and searchable audit log.
\newblock In \emph{Network and Distributed System Security Symposium}, San
  Diego, California, 2004. Internet Society.

\bibitem[Hoang et~al.(2022)Hoang, Wu, and Yuan]{hoang2022faster}
Viet~Tung Hoang, Cong Wu, and Xin Yuan.
\newblock Faster yet safer: Logging system via fixed-key blockcipher.
\newblock In \emph{31st USENIX Security Symposium}, pages 2389--2406, Boston,
  Massachusetts, 2022. USENIX Association.

\bibitem[E~Lin(1999)]{Delp}
E~Delp E~Lin.
\newblock A review of fragile image watermarks.
\newblock pages 25--29, 1999.
\newblock URL \url{ftp://skynet.ecn.purdue.edu/pub/dist/delp/acm99/paper.pdf}.
\newblock Proceedings of the Multimedia and Security Workshop (ACM Multimedia
  '99) Multimedia Contents, October 1999, Orlando, pp. 25-29.

\bibitem[Neekhara et~al.(2024)Neekhara, Hussain, Zhang, Huang, McAuley, and
  Koushanfar]{neekhara2024facesigns}
Paarth Neekhara, Shehzeen Hussain, Xinqiao Zhang, Ke~Huang, Julian McAuley, and
  Farinaz Koushanfar.
\newblock {FaceSigns}: Semi-fragile watermarks for media authentication.
\newblock \emph{ACM Transactions on Multimedia Computing, Communications and
  Applications}, 20\penalty0 (11):\penalty0 1--21, 2024.
\newblock \doi{10.1145/3640466}.

\bibitem[Zhang et~al.(2024{\natexlab{a}})Zhang, Li, Yu, Xu, Li, and
  Zhang]{zhang2024editguard}
Xuanyu Zhang, Runyi Li, Jiwen Yu, Youmin Xu, Weiqi Li, and Jian Zhang.
\newblock {EditGuard}: Versatile image watermarking for tamper localization and
  copyright protection.
\newblock In \emph{IEEE/CVF Conference on Computer Vision and Pattern
  Recognition}, pages 11964--11974, Seattle, Washington, 2024{\natexlab{a}}.
  IEEE.

\bibitem[Zhang et~al.(2025)Zhang, Tang, Xu, Li, Xu, Chen, Gao, and
  Zhang]{zhang2025omniguard}
Xuanyu Zhang, Zecheng Tang, Zhipei Xu, Runyi Li, Youmin Xu, Bin Chen, Feng Gao,
  and Jian Zhang.
\newblock {OmniGuard}: Hybrid manipulation localization via augmented versatile
  deep image watermarking.
\newblock In \emph{IEEE/CVF Conference on Computer Vision and Pattern
  Recognition}, pages 3008--3018, Nashville, Tennessee, 2025. IEEE.

\bibitem[Sander et~al.(2025)Sander, Fernandez, Durmus, Furon, and
  Douze]{sander2025watermarkanything}
Tom Sander, Pierre Fernandez, Alain~Oliviero Durmus, Teddy Furon, and Matthijs
  Douze.
\newblock Watermark anything with localized messages.
\newblock In \emph{International Conference on Learning Representations}, 2025.
\newblock URL \url{https://openreview.net/forum?id=IkZVDzdC8M}.

\bibitem[Yang et~al.(2025)Yang, Liu, Xu, Xu, Yu, Huang, Wang, and
  He]{yang2025stableguard}
Haoxin Yang, Bangzhen Liu, Xuemiao Xu, Cheng Xu, Yuyang Yu, Zikai Huang,
  Yi~Wang, and Shengfeng He.
\newblock {StableGuard}: Towards unified copyright protection and tamper
  localization in latent diffusion models.
\newblock In \emph{Advances in Neural Information Processing Systems},
  volume~38, pages 14692--14720, 2025.

\bibitem[Yao et~al.(2023)Yao, Zhao, Yu, Du, Shafran, Narasimhan, and
  Cao]{yao2023react}
Shunyu Yao, Jeffrey Zhao, Dian Yu, Nan Du, Izhak Shafran, Karthik Narasimhan,
  and Yuan Cao.
\newblock {ReAct}: Synergizing reasoning and acting in language models.
\newblock In \emph{International Conference on Learning Representations}, 2023.
\newblock URL \url{https://openreview.net/forum?id=WE_vluYUL-X}.

\bibitem[Schick et~al.(2023)Schick, Dwivedi-Yu, Dess{\`i}, Raileanu, Lomeli,
  Hambro, Zettlemoyer, Cancedda, and Scialom]{schick2023toolformer}
Timo Schick, Jane Dwivedi-Yu, Roberto Dess{\`i}, Roberta Raileanu, Maria
  Lomeli, Eric Hambro, Luke Zettlemoyer, Nicola Cancedda, and Thomas Scialom.
\newblock Toolformer: Language models can teach themselves to use tools.
\newblock In \emph{Advances in Neural Information Processing Systems},
  volume~36, pages 68539--68551, 2023.

\bibitem[Liu et~al.(2024)Liu, Yu, Zhang, Xu, Lei, Lai, Gu, Ding, Men, Yang,
  Zhang, Deng, Zeng, Du, Zhang, Shen, Zhang, Su, Sun, Huang, Dong, and
  Tang]{liu2023agentbench}
Xiao Liu, Hao Yu, Hanchen Zhang, Yifan Xu, Xuanyu Lei, Hanyu Lai, Yu~Gu,
  Hangliang Ding, Kaiwen Men, Kejuan Yang, Shudan Zhang, Xiang Deng, Aohan
  Zeng, Zhengxiao Du, Chenhui Zhang, Sheng Shen, Tianjun Zhang, Yu~Su, Huan
  Sun, Minlie Huang, Yuxiao Dong, and Jie Tang.
\newblock {AgentBench}: Evaluating {LLMs} as agents.
\newblock In \emph{International Conference on Learning Representations}, 2024.
\newblock URL \url{https://openreview.net/forum?id=zAdUB0aCTQ}.

\bibitem[Zhou et~al.(2024)Zhou, Xu, Zhu, Zhou, Lo, Sridhar, Cheng, Ou, Bisk,
  Fried, Alon, and Neubig]{zhou2024webarena}
Shuyan Zhou, Frank~F. Xu, Hao Zhu, Xuhui Zhou, Robert Lo, Abishek Sridhar,
  Xianyi Cheng, Tianyue Ou, Yonatan Bisk, Daniel Fried, Uri Alon, and Graham
  Neubig.
\newblock {WebArena}: A realistic web environment for building autonomous
  agents.
\newblock In \emph{International Conference on Learning Representations}, 2024.
\newblock URL \url{https://openreview.net/forum?id=oKn9c6ytLx}.

\bibitem[Deng et~al.(2023)Deng, Gu, Zheng, Chen, Stevens, Wang, Sun, and
  Su]{deng2023mind2web}
Xiang Deng, Yu~Gu, Boyuan Zheng, Shijie Chen, Samuel Stevens, Boshi Wang, Huan
  Sun, and Yu~Su.
\newblock {Mind2Web}: Towards a generalist agent for the web.
\newblock In \emph{Advances in Neural Information Processing Systems},
  volume~36, pages 28091--28114, 2023.

\bibitem[Xie et~al.(2024)Xie, Zhang, Chen, Li, Zhao, Cao, Hua, Cheng, Shin,
  Lei, Liu, Xu, Zhou, Savarese, Xiong, Zhong, and Yu]{xie2024osworld}
Tianbao Xie, Danyang Zhang, Jixuan Chen, Xiaochuan Li, Siheng Zhao, Ruisheng
  Cao, Toh~Jing Hua, Zhoujun Cheng, Dongchan Shin, Fangyu Lei, Yitao Liu,
  Yiheng Xu, Shuyan Zhou, Silvio Savarese, Caiming Xiong, Victor Zhong, and Tao
  Yu.
\newblock {OSWorld}: Benchmarking multimodal agents for open-ended tasks in
  real computer environments.
\newblock In \emph{Advances in Neural Information Processing Systems}, 2024.
\newblock URL \url{https://arxiv.org/abs/2404.07972}.

\bibitem[Liu et~al.(2023)Liu, Xu, and McAuley]{liu2023repobench}
Tianyang Liu, Canwen Xu, and Julian McAuley.
\newblock {RepoBench}: Benchmarking repository-level code auto-completion
  systems, 2023.

\bibitem[Zhang et~al.(2023)Zhang, Chen, Zhang, Keung, Liu, Zan, Mao, Lou, and
  Chen]{zhang2023repocoder}
Fengji Zhang, Bei Chen, Yue Zhang, Jacky Keung, Jin Liu, Daoguang Zan, Yi~Mao,
  Jian-Guang Lou, and Weizhu Chen.
\newblock {RepoCoder}: Repository-level code completion through iterative
  retrieval and generation, 2023.

\bibitem[Zhang et~al.(2024{\natexlab{b}})Zhang, Li, Li, Shi, and
  Jin]{zhang2024codeagentenhancingcodegeneration}
Kechi Zhang, Jia Li, Ge~Li, Xianjie Shi, and Zhi Jin.
\newblock Codeagent: Enhancing code generation with tool-integrated agent
  systems for real-world repo-level coding challenges, 2024{\natexlab{b}}.
\newblock URL \url{https://arxiv.org/abs/2401.07339}.

\bibitem[Zhang et~al.(2024{\natexlab{c}})Zhang, Ruan, Fan, and
  Roychoudhury]{zhang2024autocoderover}
Yuntong Zhang, Haifeng Ruan, Zhiyu Fan, and Abhik Roychoudhury.
\newblock {AutoCodeRover}: Autonomous program improvement, 2024{\natexlab{c}}.

\bibitem[Xia et~al.(2024)Xia, Deng, Dunn, and Zhang]{xia2024agentless}
Chunqiu~Steven Xia, Yinlin Deng, Soren Dunn, and Lingming Zhang.
\newblock Agentless: Demystifying {LLM}-based software engineering agents,
  2024.

\end{thebibliography}

\clearpage
\appendix

\section{Owner Capacity Calculations}
\label{app:owner-capacity}

\subsection{Rank and collision recurrence}

% Initialize \(Q_0(0,0)=1\).  From state \((r,u)\), one further nonzero-row draw gives
Let \(Q_m(r,u)\) be the probability that \(m\) basename slots have coefficient
rank \(r\) and \(u\) distinct coefficient vectors.  Repeated vectors remain separate
matrix rows, and \(P_m(6)=\sum_uQ_m(6,u)\).  Initialize \(Q_0(0,0)=1\).
From state \((r,u)\), one further nonzero-vector draw gives
\begin{equation}
\begin{array}{rcll}
(r,u)&\to&(r,u)    &\text{with probability }u/63,\\
(r,u)&\to&(r,u+1) &\text{with probability }(2^r-1-u)/63,\\
(r,u)&\to&(r+1,u+1)&\text{with probability }(64-2^r)/63.
\end{array}
\label{eq:rank-unique-recurrence}
\end{equation}
Marginalizing over \(u\), with \(P_0(0)=1\) and out-of-range terms zero, gives
\begin{equation}
P_{m+1}(r)=P_m(r)\frac{2^r-1}{63}
+P_m(r-1)\frac{64-2^{r-1}}{63}.
\label{eq:rank-recurrence}
\end{equation}

\subsection{Proof of Theorem~\ref{thm:clean-acceptance}}

\begin{proof}
Assumption~\ref{ass:ideal-prf} makes the \(D_{\mathcal B}\) selection bits independent
Bernoulli variables with exact parameter \(p_o\).  Conditional on \(M=m\), row
generation gives \(m\) independent uniform draws from the 63 nonzero vectors.

If the current draws have rank \(r\) and contain \(u\) distinct rows, exactly \(u\)
choices repeat a row, \(2^r-1-u\) are unseen rows within the span, and \(64-2^r\) lie
outside it.  This proves Eq.~\eqref{eq:rank-unique-recurrence}.

% In a clean batch every retained equation supports \(\ownerid\).  Rank six makes that
% solution unique, and with \(u\) distinct rows its score is
% \(z_{\mathrm{owner}}=\sqrt u\).  Acceptance is therefore equivalent to rank six and
% \(u\ge\lceil\tau_z^2\rceil\).  Averaging \(Q_m(6,u)\) over \(M\) proves
% Eq.~\eqref{eq:clean-acceptance}; summing over all \(u\) proves
% Eq.~\eqref{eq:rank-six-probability}.
In a clean batch, majority aggregation retains one equation for each selected
basename, so \(U=M=m\), even when coefficient vectors repeat.  Every equation supports
\(\ownerid\), giving \(H(\ownerid)=m\) and
\(z_{\mathrm{owner}}=\sqrt m\) for \(m>0\).  Rank six makes this solution the unique
maximizer.  Acceptance is therefore equivalent to rank six and \(m\ge L_\tau\).
Conditional on \(M=m\), its probability is
\(\mathbf 1\{m\ge L_\tau\}P_m(6)\).  Averaging over the binomial distribution of \(M\)
proves Eq.~\eqref{eq:clean-acceptance}; dropping the slot-count gate gives
Eq.~\eqref{eq:rank-six-probability}.
\end{proof}

For exactly six draws, the full-rank probability is
\(\prod_{r=0}^{5}(64-2^r)/63\approx0.322\), so six rows suffice algebraically but are
not a high-probability operating point.  If the \(c_s\) observations of basename \(s\)
are independently erased with probability \(\delta\), the expected number of selected
basenames retaining at least one observation is
\begin{equation}
p_o\sum_{s\in\mathcal S}(1-\delta^{c_s}).
\label{eq:surviving-owner-rows}
\end{equation}
% This expectation measures redundancy before row collisions, not recovery probability.
This is the expected number of surviving basename slots; it does not by itself
determine matrix rank or recovery probability.

\section{Owner Decoding Proofs}
\label{app:owner-decoding}

\subsection{Robust recovery}

\begin{proof}[Proof of Theorem~\ref{thm:owner-guarantees}(i)]
The true label has residual \(d_H(y,A\ownerid)=\|\eta\|_0=e\).  Any other label is
\(\mathbf x=\ownerid\oplus v\) for some \(v\ne0\), and the Hamming triangle inequality
gives
\[
d_H(y,A\mathbf x)
\ge d_H(A\ownerid,A\mathbf x)-d_H(y,A\ownerid)
=\|Av\|_0-e\ge d(A)-e.
\]
If \(2e<d(A)\), this exceeds the true residual, so no other label ties or wins.
Deleting \(s\) rows produces \(A_{\mathrm{ret}}\) with
\(d(A_{\mathrm{ret}})\ge d(A_0)-s\), which yields the stated sufficient condition.
Rank alone is not robustness: for \(A=I_6\), \(d(A)=1\).
\end{proof}

\subsection{Fixed-claim soundness}

% \begin{proof}[Proof of Theorem~\ref{thm:owner-guarantees}(ii)]
% Condition first on selection, row identities, and basename multiplicities.  Within one
% row bucket, simultaneously complementing the mask bit of every contributing basename
% preserves its tie event and flips its non-tied majority.  Conditional on retention, the
% majority bit is therefore fair.  Different row buckets contain disjoint basename inputs,
% so their tie indicators and retained majorities are independent.  Conditioning on any
% retained set leaves \(U\) independent fair right-hand sides.
% 
% For fixed \(\mathbf w_c\), each predicted value is constant conditional on its row.
% Consequently,
% \[
% H_c\sim\operatorname{Binomial}(U,1/2).
% \]
% Solving
% \(z_{\mathrm{owner}}(\mathbf w_c)\ge\tau_z\) gives
% \(H_c\ge h_\tau(U)\); summing the binomial mass above this threshold proves
% Eq.~\eqref{eq:owner-false-accept}.  The bound is conservative because acceptance also
% requires rank six, a unique optimum, and equality with the predeclared claim.
% \end{proof}
\begin{proof}[Proof of Theorem~\ref{thm:owner-guarantees}(ii)]
Condition on the fixed stream and all non-mask domain outputs, including selection,
coefficient vectors, and family-specific orientations.  Mask invariance fixes the
replay roles and basename slots.  After orientation removal, observation \(i\) in
slot \(s\) has right-hand side \(y_{s,i}=q_{s,i}\oplus m_s\), with fixed
\(q_{s,i}\) and independent fair basename masks \(m_s\).  All observations in a slot
share \(m_s\), even across READ families.  Flipping this mask preserves ties and
complements every non-tied majority.  Retention is therefore mask-independent, and
each retained majority is a fixed bit XOR \(m_s\).  Distinct slots have independent
masks even when their coefficient vectors agree.  Conditional on the retained set,
the \(U\) right-hand sides are thus independent and fair.

The predicted bits for fixed \(\mathbf w_c\) are constant under this conditioning,
so \(H_c\sim\operatorname{Binomial}(U,1/2)\).  Solving
\(z_{\mathrm{owner}}(\mathbf w_c)\ge\tau_z\) gives \(H_c\ge h_\tau(U)\).
Summing this binomial tail proves Eq.~\eqref{eq:owner-false-accept}; the additional
rank, uniqueness, and claim-equality gates can only reduce acceptance.
\end{proof}

\section{Integrity Proofs and Boundary Cases}
\label{app:integrity-proofs}

\subsection{Clean replay}

\begin{proof}[Proof of the clean-replay clause in Theorem~\ref{thm:integrity-localization}]
We induct on public prefixes.  Generator and verifier begin with the same
genesis symbol, empty segment, and counters.  Equal states serialize the same next
payload and predict the same role subtype.  Consuming the unchanged successful carrier
preserves the common state.  The induction applies to ordinary and group roles, and the
equal final suffix and seal sequence give the same terminal role.
\end{proof}

\subsection{Single-segment and group localization}

\begin{proof}[Proof of the edit clauses in Theorem~\ref{thm:integrity-localization}]
Let \(x\) and \(x'\) be the old and recomputed ordinary inputs.  Injective canonical
encoding and the changed payload give \(x'\ne x\).  Global freshness makes the new
digest uniform conditioned on the old visible subtype, so it matches that subtype with
probability at most \(q_{12,\lambda}\), proving
Eq.~\eqref{eq:ordinary-localization-bound}.

For a completed pair, let \(d_k\) and \(d'_k\) be the old and new full ordinary digests.
The event \(d'_k=d_k\) has probability \(2^{-\lambda}\).  Outside this event, the group
input changes and is globally fresh by hypothesis.  Avoiding both reports then requires
an ordinary residue match and a conditionally fresh group residue match, each with
probability at most \(q_{12,\lambda}\).  Adding the full-digest collision event gives
Eq.~\eqref{eq:group-localization-bound}.
\end{proof}

The terminal role is deliberately excluded from this multiplication: a quantitative
factor would require both a collision bound for \(H_\theta\) and global freshness of the
terminal-domain input.  Likewise, insertion, deletion, boundary-changing, seal-carrier,
compound, and projection-preserving edits are not reduced to the single-segment theorem.

\section{Protocol Pseudocode and Replay Details}
\label{app:protocol-details}

\begin{table}[H]
  \caption{Visible-only verification contract.}
  \label{tab:visible-contract}
  \small
  \begin{tabularx}{\columnwidth}{@{}>{\bfseries}p{0.22\columnwidth}X@{}}
    \toprule
    Retained & Action class, READ subtype, normalized basename or critical target,
    public status, visible order, and a profile's public carrier header when present. \\
    Excluded & Hidden reasoning, logits, observation contents beyond public linkage
    and status, private wrapper state, debug metadata, and sidecars.  Native action
    text remains visible; this includes OpenDev's published keyed carrier header. \\
    \bottomrule
  \end{tabularx}
\end{table}

\begin{table}[H]
\centering
\caption{Public READ alphabets.}
\label{tab:alphabets}
\small
\begin{tabular}{@{}lll@{}}
\toprule
Family & Q6 owner pair & Q12-only additions \\
\midrule
Explore & \texttt{dir} / \texttt{meta} & \texttt{tree}, \texttt{files} \\
Search  & \texttt{content} / \texttt{name} & \texttt{symbol}, \texttt{import} \\
Locate  & \texttt{range} / \texttt{whole} & \texttt{head}, \texttt{tail} \\
\bottomrule
\end{tabular}
\end{table}

\begin{algorithm}[H]
\caption{\tool{} reference online generation}
\label{alg:embed}
\label{alg:owner-embed}
\label{alg:seal-embed}
\small
\begin{algorithmic}[1]
\Require online agent/environment loop; key \(K\); owner \(\ownerid\); configuration
\(\theta\); owner threshold \(\bar p_o\)
\Ensure released marked trajectory \(\tau_w\)
\State initialize filename-once owner state and public integrity state, with \(c_{\mathrm{prev}}=0\)
\While{the agent loop has not terminated}
  \State receive the next agent-proposed action \(a\)
  \If{\(a\) is a selected first eligible READ of basename \(s\)}
    \State rewrite only its Q6 subtype using the declared framework profile
  \EndIf
  \State execute and publish \(a\); update the current public segment
  \State \(c\gets\mathbf1[\mathsf{class}(a)\in\mathcal C_{\mathrm{crit}}]\);
  \(b\gets c\land\neg c_{\mathrm{prev}}\); \(c_{\mathrm{prev}}\gets c\)
  \If{\(b\)}
    \State emit its ordinary Q12 role; clear the segment buffer
    \If{two emitted ordinary roles complete a group}
      \State emit the corresponding group Q12 role
    \EndIf
  \EndIf
\EndWhile
\State emit the terminal Q12 role over the final suffix and public seal state
\State \Return \(\tau_w\)
\end{algorithmic}
\end{algorithm}

\begin{algorithm}[H]
\caption{\tool{} reference visible-only verification}
\label{alg:verify}
\label{alg:owner-decode}
\label{alg:seal-replay}
\small
\begin{algorithmic}[1]
\Require observed batch \(\widetilde{\mathcal B}\); key \(K\); claim
\(\mathbf w_c\); threshold \(\tau_z\); configuration \(\theta\)
\Ensure recovered label or abstention; owner decision; localized alarms
\State apply the public transducer and replay role placement from each stream prefix
\State recover Q6 equations from replay-assigned non-seal READs
% \State majority-aggregate equal rows; erase ties; score all \(64\) labels
\State majority-aggregate each basename slot; erase ties; score all \(64\) labels
\State apply the rank, uniqueness, claim, and evidence gates
\For{each observed trajectory}
  \State replay ordinary, group, and terminal roles from the public state machine
  \State report every mismatched, missing, or status-invalid expected role
\EndFor
\State \Return owner output, diagnostics, and \(\mathcal A_{\mathrm{int}}\)
\end{algorithmic}
\end{algorithm}

\begin{table}[H]
  \caption{Reference visible-replay transitions at an expected role.}
  \label{tab:replay-transitions}
  \small
  \begin{tabularx}{\columnwidth}{@{}p{0.26\columnwidth}X@{}}
    \toprule
    Visible condition & Replay transition \\
    \midrule
    Expected Q12 subtype, success & Consume the record and advance with its observed
    subtype. \\
    Wrong subtype or class & Consume, alarm, and use the observed subtype or
    \(\mathsf{INVALID}\). \\
    Failed or unknown status & Consume, alarm, and preserve the public role count. \\
    End of stream & Consume \(\mathsf{MISSING}\) and alarm. \\
    Record after terminal & Report a trailing-record structural alarm. \\
    \bottomrule
  \end{tabularx}
\end{table}

A deleted role therefore shifts the next visible record into that public position in
the reference replay; the
parser never silently skips it.  A non-role READ insertion is not recognizable from
private origin metadata.  It is detected only through a changed protected payload,
shifted role, or a later subtype mismatch.  Projection-preserving edits can replay
cleanly by design.

\section{Framework Adapter Details}
\label{app:framework-details}

\subsection{Core duties and failure semantics}

The reference adapter hooks the agent loop before execution, after the public outcome,
and at termination.  It must (i) classify public records as READ subtypes, critical classes,
terminal classes, or \(\mathsf{OTHER}\); (ii) extract an explicitly visible regular-file
target; (iii) realize any requested Q6/Q12 subtype through a read-only action; and
(iv) serialize action, outcome, status, and order exactly as replay observes them.

The owner hook never converts a write or test into a carrier.  Integrity carriers are
additional actions on the ordinary execution path.  In the reference contract, a failed
carrier alarms and group placement advances over attempted ordinary roles.  The evaluated
profiles below specify their actual event filters and state advancement.

\subsection{SWE-agent}

SWE-agent represents a decision as textual action handled by its agent-to-computer
interface.  Before execution, the adapter recognizes public EXPLORE, SEARCH, and LOCATE
families and accepts an owner candidate only when it explicitly names a regular file in
the repository.  It changes only the command template required by the selected Q6 member,
leaving the target and read-only semantics unchanged.

The evaluated replay scans public action text in order.  It appends the current event's
payload before closing on a critical action whose preceding boundary-state class is
noncritical, excluding malformed unexecuted actions from this trigger.  Consumed seals
leave that class unchanged.  A publicly reconstructed selected owner READ also leaves
the boundary-state class unchanged, while other non-seal actions update it; owner READs
still contribute to the segment's READ counts.  The immediately following Q12 READ is
the ordinary carrier, followed by a group READ after every second consumed ordinary.
If the expected Q12 READ is absent, replay records a missing role and continues scanning.
The final visible Q12 READ is reserved as the terminal commitment over the remaining
suffix; its absence raises a terminal report.  All roles use the ordinary computer
interface, and verification never reads \texttt{extra\_info} records.

\subsection{OpenHands}

OpenHands persists typed actions and observations in an event stream.  The owner hook
intercepts only agent-sourced events entering \texttt{EventStream.add\_event}, preserves
event identity and tool-call attributes, and realizes the selected Q6 member before
serialization.  The integrity hook runs after the corresponding public observation and
executes scheduled Q12 roles through the normal runtime.  The evaluated profile joins
observations by their public cause identifier.  A non-seal action with no observation,
an error observation, or an observation beginning with \texttt{ERROR} does not enter
Q12 state and leaves the previous-critical flag unchanged.  Every admitted READ or
critical action contributes its payload; an admitted noncritical action resets the flag.
A critical action closes the buffer immediately only when that flag was false.

Replay consumes the next action at an expected ordinary or group position, compares its
subtype, and advances even after a mismatch.  Consuming an ordinary clears the buffer;
every second ordinary schedules a group.  Expected carriers are processed before the
non-seal success filter, so this evaluated replay does not separately test their
observation status.  A visible \texttt{finish} action schedules a terminal READ over the
remaining suffix, including an empty suffix; end-of-stream raises a missing-role report
only if a role was already expected.  Runtime exits that bypass the finish hook may
therefore leave no terminal protector.  The measured closure coverage is reported in
Section~\ref{sec:discussion-visibility}.

\subsection{OpenDev}

OpenDev exposes structured tool calls at its dispatcher.  A selected natural READ is
rewritten before dispatch, and the same arguments are placed in the public tool-call
message and sent to the executor.  Natural actions and carriers receive monotonically
increasing public indices, resolving multiple calls within one response.  Diagnostic
JSON is retained for tests but excluded from verification.  Its Q12 Bash carrier also
contains a visible authenticated role/ordinal header.  The evaluated verifier retains
this header and checks its 64-bit HMAC tag over role, ordinal, and subtype against the
visible command template.  It uses the resulting role to recognize carriers.

Among successful nonterminal natural calls, READs and other noncritical calls reset the
previous-critical flag.  Critical calls append their payload and close immediately only
when the flag was false.  The next call must carry the expected ordinary header; a
missing or failed carrier is reported.  A successful ordinary carrier advances the
ordinary sequence and clears the buffer, and every second such carrier schedules the
following group call.  Replay selects the last unconsumed terminal header with ordinal
\(m+1\) after the last successful natural call and checks it against the suffix and
accumulated commitments.  A missing terminal is reported.  Thus this profile includes
an authenticated public header in addition to the Q12 subtype commitments.

\subsection{Canonicalization and validation checklist}

The adapters share semantics rather than concrete commands.  The owner selector uses a
lowercase normalized basename only after verifying an explicit regular-file target.
Critical targets come only from visible edit or execution arguments; absent targets use
the public null symbol.  A seal's operational file target does not enter
\(P_{\mathrm{int}}\), so changing only that target is projection-preserving.

Reference conformance requires unique action and observation joins, stable public order,
visible failure status, and read-only realization of all Q6/Q12 subtypes.  Profile
vectors cover selection, rows, the masked and direct owner variants,
ordinary/group/terminal seals, and empty inputs.  Round-trip tests remove auxiliary
metadata before replay and compare owner equations, rank, decoded label, and integrity
reports with the profile's original replay; published native action text is retained.

\section{Watermark Details: Released Trajectories}
\label{app:watermark-details}

This appendix shows sanitized excerpts from two real SWE-agent trajectories in the
frozen evaluation.  Machine-specific path prefixes are replaced by
\texttt{<workspace>}, task and run identifiers are omitted, and long fields are
shortened; released action text and order are otherwise retained.  Bold field labels
and colored role callouts are paper typography and are not inputs to the visible-only
verifier.

\subsection{From Prompt Context to a Released Record}

The first example connects unchanged prompt excerpts to one real released
action--observation record.  Longer interface, demonstration, and issue material is
omitted.

\begin{trajdetail}{(a) Unchanged Prompt Context (Excerpts)}
SYSTEM PROMPT (unchanged excerpt):
SETTING: You are an autonomous programmer, and you're working directly
in the command line with a special interface.
The special interface consists of a file editor that shows you 100 lines
of a file at a time.
[...]

USER TASK (unchanged excerpt):
We're currently solving the following issue within our repository.
Here's the issue text:
ISSUE:
Q(...) & Exists(...) raises a TypeError
[...]
\end{trajdetail}

\tool{} modifies neither prompt.  In this record, its intervention occurs after the
agent proposes an eligible action and before that action is executed and published.

\begin{trajdetail}{(b) One Real Prompt-to-Trajectory Record}
AGENT PROPOSAL (generation-side metadata; not released):
  open reproduce.py
        |
        |  TrajMark wrapper: action-layer rewrite only
        v
RELEASED ACTION, position 14:
  cat "$(pwd)/reproduce.py" | head -n 40
VISIBLE OBSERVATION begins:
  from django.db.models import Q
  from django.db.models.expressions import Exists
  ...
NEXT RELEASED ACTION, position 15:
  python reproduce.py
\end{trajdetail}

\noindent\textcolor{trajdetailframe}{\textbf{Owner carrier (position 14).}}
This natural \textsc{Read} targets the same file as the proposal and occupies the same
trajectory position; no extra action is introduced.

\subsection{Integrity Seals from the Patch to the Trace Tail}

The next excerpt is a continuous region from a second released trajectory.  It moves
from environment setup to reproduction, source inspection, a patch, and validation.
Positions 22 and 23 are visible read-only actions placed immediately after the edit.

\begin{trajdetail}{(c) Compact View of Released Records 16--25 (Observations Shortened)}
16  python -c "import astropy; print(astropy.__version__)" 2>&1 ||
      echo "Not installed"
17  python setup.py build_ext --inplace 2>&1 | tail -20
18  python reproduce_bug.py
19  open astropy/nddata/mixins/ndarithmetic.py 520
20  goto 500
21  edit 515:527
      [patch body omitted]
    end_of_edit
22  sed -n '1,40p' <workspace>/astropy/nddata/mixins/ndarithmetic.py
23  find $(dirname -- <workspace>/astropy/nddata/mixins/ndarithmetic.py)
      -maxdepth 2 -type d -print | head -n 40
24  python reproduce_bug.py
25  python -m pytest astropy/nddata/mixins/tests/test_arithmetic.py -v
      2>&1 | tail -40
\end{trajdetail}

\noindent\textcolor{trajdetailframe}{\textbf{Ordinary seal (position 22); group seal
(position 23).}}
In the verifier-visible projection, their records contain the command and normal
environment output: the first returns the opening lines of the edited source file, and
the second returns nearby directories.

The final excerpt comes from the tail of the same trajectory.  It retains the two
read-only records, cleanup, submission, and the final visible read in their released
order.

\begin{trajdetail}{(d) Tail of the Released Action Stream}
37  find $(dirname -- "$(pwd)/reproduce_bug.py") -maxdepth 1
      -name $(basename -- "$(pwd)/reproduce_bug.py") -print | head -n 40
38  find $(dirname -- "$(pwd)/reproduce_bug.py") -maxdepth 1
      -name $(basename -- "$(pwd)/reproduce_bug.py") -print | head -n 40
39  rm reproduce_bug.py
40  submit
41  find . -maxdepth 2 -name '*.py' -print | head -n 40
\end{trajdetail}

\noindent\textcolor{trajdetailframe}{\textbf{Ordinary seal (position 37); group seal
(position 38); terminal seal (position 41).}}
These callouts explain the otherwise ordinary-looking read commands and are not inputs
to verification.

Together, the excerpts show the verifier-visible action stream.  Verification derives
its inputs from released action text and order together with the associated public
execution status defined in Appendix~\ref{app:protocol-details}.  Ownership marking
changes one existing read, while integrity seals appear as additional read-only actions
interleaved with the agent's work.  The equations, replay rules, and adapter contracts are defined once in
Section~\ref{sec:trajmark} and Appendices~\ref{app:protocol-details}--\ref{app:framework-details};
they are not repeated here.

\clearpage
\section{Disaggregated Ownership Robustness}
\label{app:ownership-results}

Figure~\ref{fig:verified-owner-id-recovery-method-comparison} pools the three coding
agents within each LLM column to keep the main comparison compact.  The four figures
below retain every agent--LLM condition.  Each is one row of the main figure expanded to
a three-by-three grid, with agents as rows, LLMs as columns, and methods as curves.

\subsection{Random deletion}
\begin{figure}[H]
  \centering
  \includegraphics[width=\textwidth]{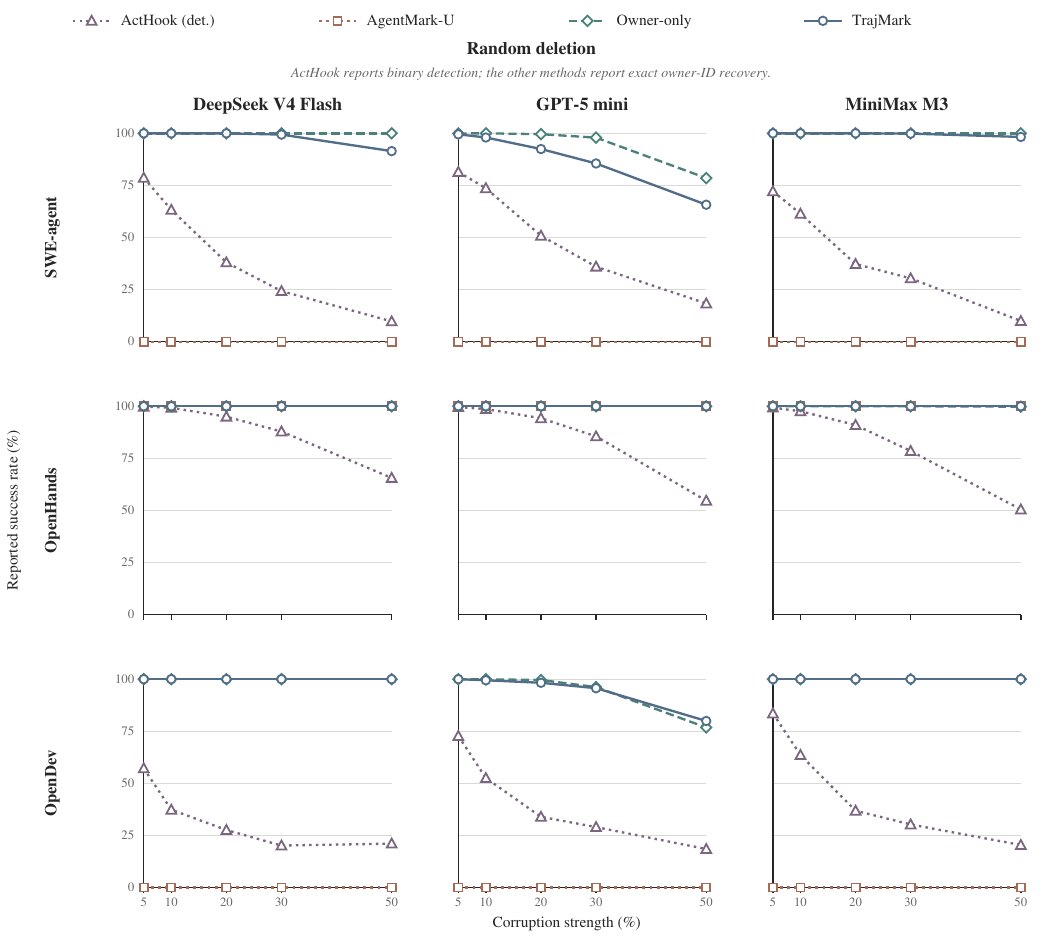}
  \Description{Nine line charts show random-deletion watermark outcomes. Rows are the
  three coding agents, columns are the three LLMs, and every panel compares ActHook-style,
  AgentMark-U, Owner-only, and TrajMark over five corruption strengths.}
  \caption{Disaggregated random-deletion results corresponding to the first row of
  Figure~\ref{fig:verified-owner-id-recovery-method-comparison}.  Each point contains
  1,500 corrupted \(B=50\) batches: 500 deterministic trials for each of three seeds.
  ActHook-style reports binary hook detection; the three ID-bearing methods require
  exact recovery of \texttt{0x3f}.}
  \label{fig:owner-outcomes-delete-full}
\end{figure}

Deletion removes visible evidence.  The separate panels show that the pooled trend is
not caused by a single agent: AgentMark-U loses indexed evidence broadly, whereas the
filename-indexed methods degrade mainly in the lowest-margin SWE-agent/GPT-5 mini and
OpenDev/GPT-5 mini conditions.
\FloatBarrier

\subsection{Random insertion}
\begin{figure}[H]
  \centering
  \includegraphics[width=\textwidth]{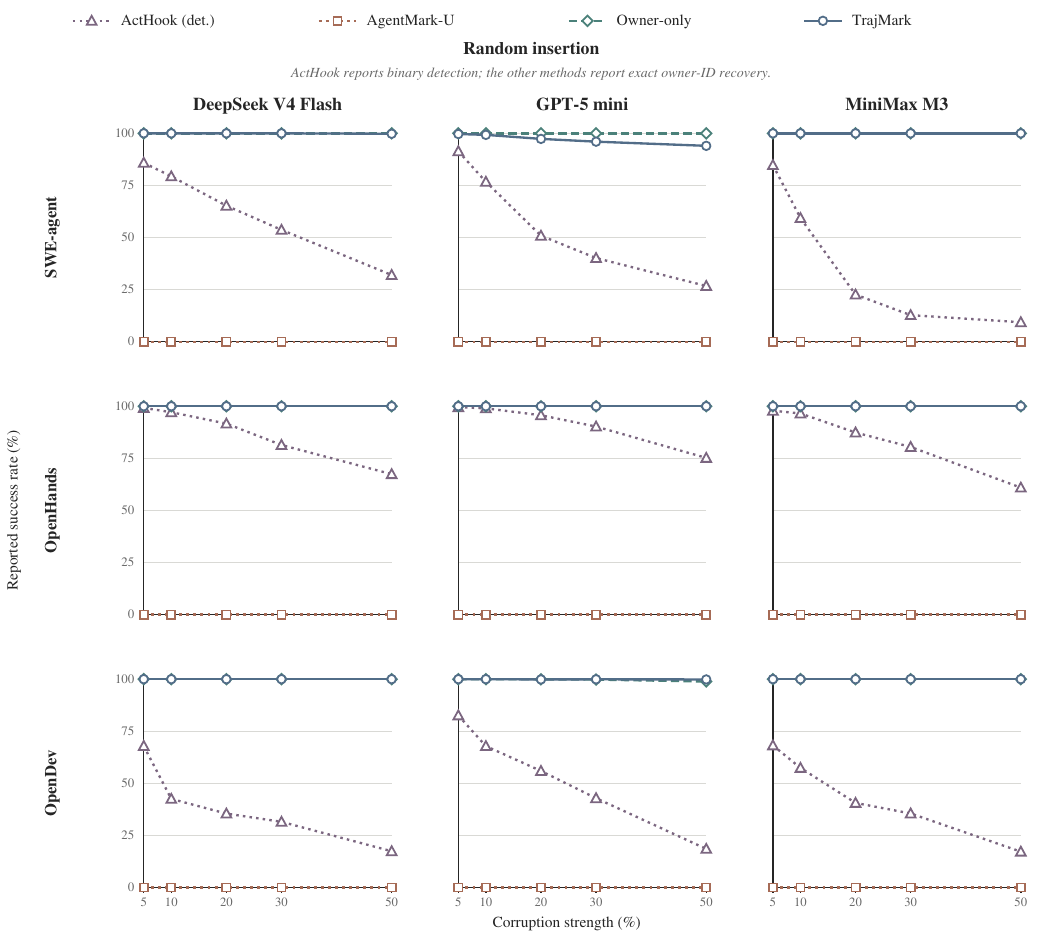}
  \Description{Nine line charts show random-insertion watermark outcomes. Rows are the
  three coding agents, columns are the three LLMs, and every panel compares the four
  evaluated watermark methods.}
  \caption{Disaggregated random-insertion results corresponding to the second row of
  Figure~\ref{fig:verified-owner-id-recovery-method-comparison}.  Each point contains
  1,500 corrupted \(B=50\) batches under the same seed and success definitions as
  Figure~\ref{fig:owner-outcomes-delete-full}.}
  \label{fig:owner-outcomes-insert-full}
\end{figure}

Unrelated insertions normally add no valid filename equation, so Owner-only and
\tool{} retain the evidence already present.  AgentMark-U instead depends on global
action positions and therefore fails under the induced index shift in every panel.
\FloatBarrier

\subsection{Random replacement}
\begin{figure}[H]
  \centering
  \includegraphics[width=\textwidth]{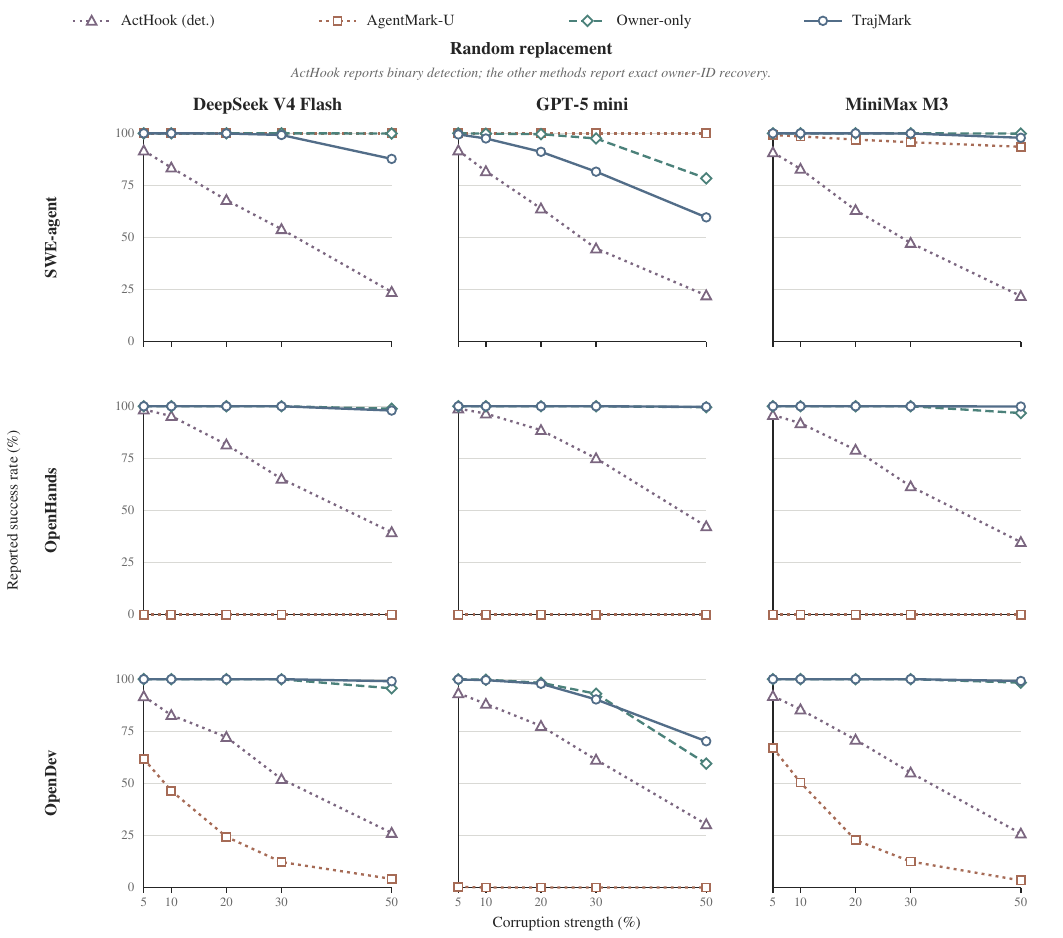}
  \Description{Nine line charts show random-replacement watermark outcomes. Rows are
  agents, columns are LLMs, and every panel compares the four evaluated methods over
  five corruption strengths.}
  \caption{Disaggregated random-replacement results corresponding to the third row of
  Figure~\ref{fig:verified-owner-id-recovery-method-comparison}.  Each point contains
  1,500 corrupted \(B=50\) batches.  Replacement may both remove an existing equation
  and introduce a conflicting visible choice.}
  \label{fig:owner-outcomes-replace-full}
\end{figure}

The full grid exposes the model-dependent evidence margin hidden by pooling: the largest
owner-recovery loss occurs in the GPT-5 mini cells, while the remaining filename-indexed
conditions stay near their clean decisions through moderate corruption.
\FloatBarrier

\subsection{Exact-budget white-box carrier deletion}
\begin{figure}[H]
  \centering
  \includegraphics[width=\textwidth]{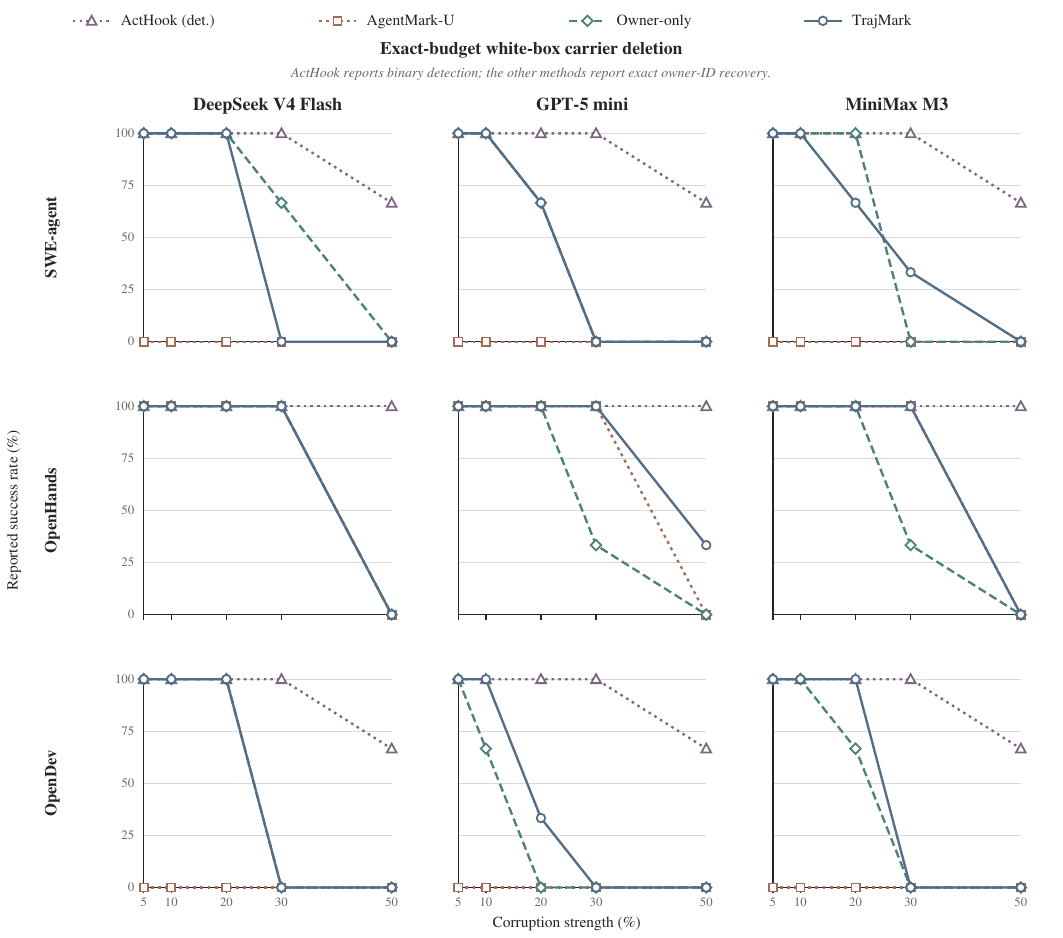}
  \Description{Nine line charts show exact-budget white-box deletion of each method's
  known owner carriers. Rows are agents, columns are LLMs, and every panel compares the
  four methods.}
  \caption{Disaggregated exact-budget white-box carrier deletion corresponding to the
  fourth row of Figure~\ref{fig:verified-owner-id-recovery-method-comparison}.  Each
  point contains the three deterministic seed batches for its agent--LLM--method cell.
  The attacker deletes exactly \(k=\lfloor\delta N+0.5\rfloor\) known owner-carrier
%   occurrences.  For an ID-bearing method it exhaustively searches the 63 nonzero
%   six-bit null directions for a minimum-cost rank-breaking cut; for ActHook-style it
%   removes keyed hooks.  Q12 seals are excluded.}
  occurrences from the visible trajectories and reruns evidence extraction and
  native decoding. Each result uses the strongest of up to 63 deterministic
  candidate deletion sets. Clean Q12 seal positions are excluded.}
  \label{fig:owner-outcomes-whitebox-full}
\end{figure}

% Unlike random corruption, the white-box attack spends every deletion on owner evidence.
% The per-cell curves show where each batch first loses enough independent equations for
% rank-six exact recovery; they are a carrier-aware boundary rather than a random-action
% robustness estimate.
For each nonzero \(v\in\mathbb{F}_2^6\), the attacker orders clean owner slots by
\(a^\top v=1\) first, then by increasing occurrence count and slot name.
Occurrences within a slot are ordered by task identifier and visible position;
the first \(k\) form a candidate deletion set. ActHook-style uses 63 deterministic
hash-based orderings of its known hooks. Duplicate deletion sets are evaluated
once. After physical deletion and native decoding, a failed recovery or hook
detection takes precedence; candidate index breaks ties. The curves therefore
measure post-deletion decoding, including changes to position indices and
first-occurrence carrier selection.
\FloatBarrier

\end{document}